\documentclass[a4paper]{amsart}

\usepackage[english]{babel} 
\usepackage{hyperref} 
\usepackage{graphicx,color,cancel}  
\usepackage[dvipsnames]{xcolor} 

\usepackage{amsmath,amssymb,amsfonts,amsthm} 
\usepackage{mathtools,comment,todonotes} 
\usepackage{tikz-cd}  
\usepackage[all,cmtip]{xy} 
\usepackage[bbgreekl]{mathbbol} 

\mathtoolsset{showonlyrefs}

\theoremstyle{definition} 
	\newtheorem{definition}{Definition}
	\newtheorem{remark}[definition]{{Remark}}

\theoremstyle{plain} 
	\newtheorem{theorem}[definition]{Theorem}

\usepackage[bibstyle=alphabetic,citestyle=alphabetic,backend=biber,giveninits=true, maxbibnames=5, sorting=nyt, sortcites=true, url=false, isbn=false]{biblatex}
\renewbibmacro{in:}{} 
\usepackage{csquotes} 
\bibliography{bibliography}

\newcommand{\intl}{\int\limits}

\renewcommand{\th}{\widetilde{h}}

\title{BV analysis of Polyakov and Nambu--Goto theories with boundary}
\author{S. Martinoli}
\address{Institute for Theoretical Physics, ETH Zurich, Wolfgang Pauli strasse 27, 8092, Z\"urich, Switzerland}

\curraddr{CAMGSD, Departamento de Matem\'atica, Instituto Superior Tecnico, Universidade de Lisboa, 1049-001 Lisboa, Portugal}

\email{sebastiano.martinoli@tecnico.ulisboa.pt}

\author{M. Schiavina}
\address{Institute for Theoretical Physics, ETH Zurich, Wolfgang Pauli strasse 27, 8092, Z\"urich, Switzerland and Department of Mathematics, ETH Zurich, R\"amistrasse 101, 8092, Z\"urich, Switzerland }
\email{micschia@ethz.ch}
\date{}

\thanks{This research was (partly) supported by the NCCR SwissMAP, funded by the Swiss National Science Foundation. }

\begin{document}

\begin{abstract}
    The Batalin--Vilkovisky data for Polyakov string theory on a manifold with (non-null) boundary is shown to induce compatible Batalin--Fradkin--Vilkovisky data, thus allowing BV-quantisation on manifolds with boundary. On the other hand, the analogous formulation of Nambu--Goto string theory fails to satisfy the needed regularity requirements. As a byproduct, a concise description is given of the reduced phase spaces of both models and their relation, for any target $d$-dimensional Lorentzian manifold.
\end{abstract}

\maketitle
\tableofcontents

\section{Introduction}
Classical mechanics can be seen as a simple and instructive example of a Lagrangian field theory describing a particle moving in some reference manifold. The theory can be formulated in such a way that the variational problem becomes reparametrisation invariant, and critical configurations are unparametrised geodesics in the target manifold. This version of classical mechanics is often called Jacobi theory and it embodies Maupertuis' principle \cite{Jacobi}.

As it is known, it is possible to recover classical mechanics by means of a more sophisticated coupling of a scalar with a one-dimensional gravitational theory (see \cite{CS2016a} for a recent account pertinent to this paper). This requires the introduction of a dynamical metric on the source (the ``worldline''), which will determine the behaviour of the particle's trajectory in the target. Both Jacobi theory and 1d gravity with scalar matter are reparametrisation invariant, and in physics' parlance we say that they are \emph{classically equivalent}, i.e.\ the two theories describe the same moduli space of solutions modulo (local) symmetry.

This picture allows a straightforward generalisation to extended objects moving in a reference manifold. The field theory one obtains by directly generalising minimal curves in a manifold is called Nambu--Goto string theory \cite{Nambu,Goto}: a Sigma model where the source manifold (the ``worldsheet'') is a two-dimensional differentiable manifold, and the target is some N-dimensional Lorentzian manifold.  It yields a variational problem for minimally-embedded 2d Lorentzian submanifolds, which can be seen as the evolution of 1d extended objects --- called strings. The two-dimensional version of a gravitational field coupling with scalars is instead known as Polyakov theory \cite{BrinkdiVecchia, POLYAKOV1981207}, which is also invariant under source diffeomorphisms, but exhibits an additional symmetry, corresponding to the rescalings of the dynamical metric. 

Similarly to the 1d scenario, Polyakov theory and Nambu--Goto theories are classically equivalent: one can solve the equations of motion for the dynamical metric and, upon restricting the theory to that partial critical locus, one recovers Nambu--Goto theory. It is important to observe, at this stage, that this equivalence only holds when one restricts all possible maps to embeddings into the target manifold\footnote{See the conditions required in \cite{CS2017} in 1d, and the comment on singular configurations for Nambu--Goto theory in \cite{BRZ}.}.

In  order to quantise a field theory that admits local symmetries, one can conveniently phrase the model within a cohomological setting, with the intent of describing its moduli space of solutions by means of the cohomology of an appropriate complex. This is the philosophy of the Batalin--Vilkovisky (BV) formalism \cite{BV1,BV2}.

A number of recent applications, however, suggest that this picture might be incomplete, and that a satisfactory description of a quantum field theory requires a detailed treatment of certain cohomological data induced on codimension $1$ (and in principle higher) hypersurfaces. The main scenario here is given by boundaries: although, instead of fixing a particular field configuration at the boundary, we consider the induced data as structural information, essential for determining both classical and quantum behaviour of the model.

It is clear that boundaries host crucial physical information for a field theory, since that is where the Phase Space of the theory naturally resides\footnote{Here we do not distinguish between Cauchy surfaces and time-like boundaries, because that requires the specification of a background metric. The assignment of boundary data that we are after, instead, is functorial and transcends the request that initial values can be actually extended to solutions in the bulk.}. Following the geometric approach of Kijowski and Tulczyjew \cite{KT1979}, to a classical field theory on a manifold with boundary $(M,\partial M)$ we associate a symplectic manifold of classical boundary fields $(F^\partial,\omega^\partial)$ and we look at a coisotropic submanifold $C\subset F^\partial$ describing boundary configurations that can be extended to a solution of the bulk equations of motion for a small enough cylinder $\partial M \times [0,\epsilon)$. Usually, the submanifold $C$ is the vanishing locus of a set of \emph{first class} constraints, i.e.\ functions $\{\phi_i\}$ in involution. For example, in Yang--Mills theory this is given by (generalised) Gauss' Law, while for General Relativity these are called ``Hamiltonian and momentum constraints'' \cite{deWitt,Fischer-MarsdenGR} (see also \cite{CS2016b,CS2019}). When this is the case, the coisotropic submanifold $C=\mathrm{Zero}\{\phi_i\}$ describes the \emph{Reduced Phase Space} of the system, which is defined as the (symplectic) reduction\footnote{This is $C$ reduced by its characteristic foliation. Observe that $C$ need not necessarily be the zero of an equivariant moment map.} $\underline{C}$.

If the bulk theory is formulated in the BV formalism, one can apply a similar procedure to induce, from the bulk data, a cochain complex associated to the boundary submanifold. When this induction procedure is unobstructed, the end result is a cohomological resolution of $\underline{C}$, i.e.\ a complex whose cohomology in degree $0$ is the space of functions on the reduced phase space \cite{CMR2012}. We call this the BFV data, after Batalin, Fradkin and Vilkovisky \cite{BV3,Schaetz:2008,Stasheff1997}. The collection of the BV and BFV data, together with the chain map linking the two, is called a BV-BFV pair and it is the starting point for cohomological quantisation of field theories with local symmetries on manifolds with boundary, as proposed in \cite{CMR2}. When a BV-theory admits a BV-BFV pair, it is called $1$-extendable.

It has been recently shown that there are several important examples of classically equivalent theories where only one of the two equivalent models is $1$-extendable to a BV-BFV theory. For example, while 1d gravity with matter is $1$-extendable, Jacobi theory is not \cite{CS2016a}. Remarkably, this is the case also for certain formulations of gravity in dimension $4$ and higher \cite{ScTH,CS2016b,CS2017}. The obstruction for a BV theory to be $1$-extendable is the regularity of the kernel of a natural closed two-form on the space of restrictions of fields to the boundary.

In this article we compute the reduced phase spaces of both Polyakov and Nambu--Goto theories and show how they are symplectomorphic to each other, whenever the boundary of the worldsheet is non-null. This result holds for any target Lorentzian manifold, and thus provides a general description of the phase space of string theory.

Then, we move on to show that Polyakov theory is $1$-extendable when phrased in the BV formalism and explicitly derive its BV-BFV structure. On the other hand, we show that the same  natural procedure fails in the case of the Nambu--Goto theory, which is thus not $1$-extendable. This result strengthens the argument in favour of a refined notion of equivalence of field theories on manifolds with boundaries, in view of quantisation. 

On a closed manifold without boundary, one can phrase equivalence of theories in terms of quasi-isomorphisms of BV complexes. It is often argued that classically equivalent theories such as Polyakov and Nambu--Goto (or their analogue 1d gravity and Jacobi theory studied in \cite{CS2017,CaCaSc}), which only differ by what is often called \emph{auxiliary fields content},\footnote{Here the metric is seen as an auxiliary field.} have quasi-isomorphic BV complexes \cite{Henn}. This comparison is extended ``to the boundary'' following \cite{BBH}, in the sense that one looks at the BV complex on local forms, where the BV differential is enriched by the de Rham differential\footnote{To be more precise, this is the Horizontal differential on local forms.}, and the two theories have quasi-isomorphic BV de Rham complexes. The crucial observation, however, is that the ``extension to the boundary'' in the sense of \cite{BBH} might fail to grasp the regularity requirements of a BV-BFV pair\footnote{For example, we require BFV data to be symplectic, while the naive boundary data obtained by restriction is a priori only pre-symplectic.}.

One can phrase the problem in the following sense. BV-deRham equivalence is a statement about the existence of a quasi-isomorphism preserving the cohomology classes of the relevant BV data (see \cite[Definitions 2.3.1 and 2.6.3]{CaCaSc}, where this relation is termed ``lax equivalence''). However, said quasi isomorphism need not preserve the regularity condition required for a strict BV-BFV pair to exist. As a consequence, a particular representative in a BV equivalence class might not admit an extension to the boundary in the sense of \cite{CMR2012}.

For example, in the 1d scenario, it was shown in \cite{CaCaSc} that there exists an explicit chain homotopy that yields an equivalence of the BV complexes of Jacobi theory and 1d gravity with matter. This BV-chain homotopy extends ``to the boundary'' as BV-de Rham equivalence in the sense of \cite{BBH} (or lax equivalence following \cite{CaCaSc}), meaning that it preserves the BV deRham class of the defining BV data. However, it sends an extendable BV theory (1d gravity) to a non extendable one (Jacobi theory), meaning that BV chain homotopies do not to preserve the regularity requirements needed to have a BV--BFV pair.\footnote{Compare this with the notion of ``strictification'' of a BV theory with boundary in \cite[Definition 12]{MSW2019}.} Hence, this provides a ``best case scenario'' example of two theories that are classically equivalent with quasi-isomorphic BV (-de Rham) complexes, but such that their BV-BFV behaviour differs significantly. The way we interpret this fact, following the observations in \cite{CaCaSc}, is that among BV-equivalent models one needs to find a representative that is 1-extendable, i.e.\ such that it will induce a BV-BFV pair.

In this paper we show that the same discrepancy in BV-BFV extension arises when comparing Polyakov and Nambu--Goto theories. Even assuming that the arguments of \cite{BBH} can be used,  or some other argument is found, to show that the respective BV complexes are quasi-isomorphic\footnote{We observe that according to \cite{BRZ} the (classical and quantum) BV cohomologies of Polyakov and Nambu--Goto theory have been shown to differ already without considering boundaries. This result is in apparent contradiction with the general arguments of \cite{Henn,BBH} and, to the best of our knowledge, this issue has not been resolved.}, when looked through the lens of the BV-BFV formalism, the theories differ.

In order to clarify the bulk-to-boundary behaviour of the two models, we present a detailed analysis of the reduced phase spaces for Polyakov and Nambu--Goto theories, and prove that they are symplectomorphic. This shows that one can consider a singular \emph{abstract} reduced phase space for 2d string theory, which can be represented by either Polyakov or Nambu--Goto theories. However, similarly to the 1d case, a discrepancy emerges when attempting to construct a BV-BFV pair, meaning that different choices of a theory representing the moduli space might have different properties. 

We interpret this statement by saying that the abstract theory of bosonic strings has one realisation --- Polyakov theory --- which admits a strict BV-BFV description and lends itself to quantisation with boundary. On the contrary, Nambu--Goto theory is not a good presentation of string theory for this purpose.

Let us stress that the bulk-boundary induction procedure we employ here is natural, and that there is value in identifying those BV theories that are naturally $1$-extendable. For theories that do not admit $1$-extension (known so far are Palatini--Cartan gravity in $d\geq4$ \cite{ScTH,CS2017}, Plebanski gravity \cite{ScTH} and Jacobi Theory \cite{CS2016a}) the only known workarounds to the obstruction to extendability involve restricting the available configurations. In the best case scenario, this means choosing another representative of the theory in the same BV-equivalence class, with better extendability properties.

It is not possible to exclude that one may construct a BV-BFV pair by means other than the procedure described here. For example, imposing certain boundary conditions might improve the extendability of a theory, although one is not guaranteed to get the right BFV theory, as was shown in \cite[Section 5]{CS2019}, for Palatini gravity. Another well-proven method for generating BV-BFV pairs employs the AKSZ construction,\footnote{After Alexandrov, Kontsevich, Schwarz and Zaboronski \cite{AKSZ} and the extension by Grigoriev and Damgaard \cite{GrigorievDamgaard}.} however this is bound to fail for Nambu--Goto theory, given that its BFV data vanishes\footnote{To understand this statement we refer to the following section ``literature overview'', and Theorem \ref{thm:NGRPS}. The problem is that Nambu--Goto theory does not truly admit a set of constraints.}. Asymptotic falloff conditions are another valid guess (see \cite{RejznerSchiavina21}), and ultimately one could even change the BV quantisation prescription.\footnote{One promising attempt would be that of changing the canonical BV symplectic form to include boundary terms, although this has not been thoroughly studied yet.} 
However, for a single given theory that is not $1$-extendable with the boundary-induction procedure we employ here, the obstruction to $1$-extension we discuss here is sufficient to void its eligibility as a candidate for quantisation in the presence of boundaries, without taking further precautions into account.

At this stage, it is not clear whether a ``boundary-compatible presentation'' of the moduli spaces of the theory can be found within the BV-equivalence class of a given model, but since our focus here is on two theories that are known to be classically-equivalent, the logic is reversed: we claim that Polyakov theory is indeed a better presentation of the abstract theory of bosonic strings, due to its BV-BFV behaviour.

\subsection*{Literature overview}
The purpose of this paper is twofold. On the one hand we analyse both Polyakov and Nambu--Goto strings within the symplectic approach of Kijowski and Tulczyjew \cite{KT1979}, to the effect of describing and relating their reduced phase spaces. This is done in Theorems \ref{THM:PolyakovRPS} and \ref{thm:NGRPS}. The analysis of constraints of Polyakov theory  we perform provides a clean symplectic description of its reduced phase space, which had been analysed, e.g., in \cite[Ch 12.2]{BrinkLars1988Post} and \cite{PhysRevD.34.2430}. Moreover, we show that the usual practice of looking at Nambu--Goto and Polyakov theories as the same constrained Hamiltonian system (see for example \cite[Eq. 4.4]{FujiwaraAnomalous} and \cite[Eq. 2]{FujiwaraNG-BFV}) is justified by the partial (pre-)symplectic reduction presented in Theorem \ref{thm:NGRPS}, which explicitly relates the boundary structure of the two models. 

Since the two theories have equivalent reduced phase spaces (for every target Lorentzian manifold), they are interchangeable. However, a better point of view is perhaps that there is an \emph{abstract} reduced phase space for 2d bosonic string theory, which can be \emph{represented} either by Polyakov or Nambu--Goto theories in the bulk. The choice of one theory over another might yield differences, as highlighted in the second part of this work, where we show that while Polyakov theory admits a BV-BFV description, Nambu--Goto does not (Theorems \ref{thm:BVBFVPolyakov} and \ref{thm:NGNOGO}, respectively). This is also supported by the results of \cite{BRZ}, which pointed at a discrepancy in the observables admitted by the two models\footnote{More precisely, their BV cohomologies have been shown to differ.}.

To compare our results with previous attempts to describe the BFV structure for Nambu--Goto theory, we point once again at Theorem \ref{thm:NGRPS}. Indeed, the BFV data presented in \cite[Eq. 7]{FujiwaraNG-BFV} resolves (the zero locus defined by) a set of constraints $\{\phi_i\}\in C^\infty(T^*C^\infty(\partial M,N))$, defined in the cotangent bundle to maps from the boundary of the worldsheet $M$ to the target $N$. It is in fact equivalent to Equation \eqref{e:PolBoundaryAction} below, after a redefinition of fields. While constraints are natural in Polyakov theory, they arise only after a partial reduction in Nambu--Goto theory\footnote{Symplectic reduction of Polyakov's constrained coisotropic submanifold with respect to its characteristic foliation coincides with the residual presymplectic kernel reduction for Nambu--Goto theory.}, as we explain in Theorem \ref{thm:NGRPS}.

Our results point out that such BFV data, associated to the reduced phase space of a 2d bosonic string, can be induced from (and is compatible with) the bulk BV data associated to Polyakov theory, but not from the bulk BV data associated to Nambu--Goto theory (Theorem \ref{thm:NGNOGO}).

\subsection*{Summary of results}
Polyakov theory is a field theory of maps $X\in C^\infty(M,N)$ and of densitised Lorentzian metrics $\th\in \mathcal{DM}(M,\partial M)$ (see Definitions \ref{def:Polth} and \ref{def:DPR}). It can be formulated as a Batalin--Vilkovisky theory (Definition \ref{def:PolBVth}) on the $(-1)$-symplectic manifold\footnote{We denote fields by $(X,X^\dag)\in T^*[-1]C^\infty(M,N)$, $(\zeta,\zeta^\dag)\in T^*[-1]\mathfrak{X}[1](M)$ and $(\th,\th^\dag)\in T^*[-1]\mathcal{DPR}(M,\partial M)$. See Definitions \ref{def:cotangent} and \ref{def:T^*DPR} for the definition of the cotangent bundles used here.}
\begin{equation}
\left(\mathcal{F}_P =  T^*[-1]\left(\mathcal{DPR}(M,\partial M) \times C^\infty(M,N)\times \mathfrak{X}[1](M)\right), \Omega_P\right).
\end{equation}
with action functional:
\begin{equation}
S_P= S^{\text{cl}}_P + \int_M \langle X^\dag,L_\zeta X\rangle + \frac12 \langle\zeta^\dag,[\zeta,\zeta]\rangle + \langle{\th}^\dag, L_\zeta {\th}\rangle,
\end{equation}

We prove that this BV theory admits an extension to a BV-BFV theory (Definition \ref{def:BVBFV}, Theorem \ref{thm:BVBFVPolyakov}), with BFV data:
\begin{equation}
    \mathcal{F}_P^{\partial} = T^*\left(C^\infty(\partial M, N)\times \mathfrak{X}[1](\partial M)\times C^\infty[1](\partial M)\right)
\end{equation}
with graded $(0)$-symplectic structure\footnote{We denote by $(X,J)\in T^*C^\infty(\partial M, N)$ a vector bundle morphism $T\partial M \to T^*M$ covering a smooth map $X$.}  
\begin{equation}
    \Omega^\partial = \delta \alpha^\partial =   \delta \int_{\partial M} J_\mu \delta X^\mu + \sigma^{\dag}_n \delta \sigma^n + \iota_{\delta \sigma^\partial}\sigma^{\dag}_\partial,
\end{equation}
(we denote $\sigma^n \in C^\infty[1](\partial M)$, $\sigma^\partial \in \mathfrak{X}[1](\partial M)$, while fields with a dagger are fibre coordinates) together with the BFV action ($\partial_t$ denotes the tangential derivative on $\partial M$):
\begin{equation}
    S^\partial_P = \int_{\partial M} -  (L_{\sigma^\partial} X)^\mu J_\mu - \frac{1}{2} \sigma^n \bigg[ J_\mu J^\mu +  \partial_t X^\mu \partial_t X_\mu \bigg] + \sigma^{\dag}_n L_{\sigma^\partial} \sigma^n + \frac12 \iota_{[\sigma^\partial,\sigma^\partial]}\sigma^{\dag}_\partial.
\end{equation}
The BFV data thus found is a resolution of the reduced phase space for Polyakov theory, described as the coisotropic submanifold $\Phi^{\text{red}}_P\subset T^*\left(C^\infty(\partial M, N)\right)$ seen as the vanishing locus of 
\begin{equation}
H_\phi \coloneqq\intl_{\partial M} \phi\left(\partial_t X_\mu \partial_t X^\mu + J_\mu J^\mu\right) 
\qquad L_\psi \coloneqq 2\intl_{\partial M} \psi\partial_t X^\mu J_\mu
\end{equation}
for $\phi,\psi\in C^\infty(\partial M)$, which satisfy:
\begin{equation}
\{H_\phi, H_{\phi '}\}  = L_{[\phi, \phi ']} \qquad 
      \{L_{\psi}, L_{\psi '}\}  = L_{[\psi, \psi ']} \qquad 
      \{H_{\phi}, L_{\psi}\} = H_{[\phi, \psi]}
\end{equation}
where $[\phi, \psi]:= (\partial_t \phi) \psi - \phi (\partial_t \psi)$.

In the context of Nambu--Goto theory, we show that the natural BV theory associated to the model, described on the BV space of fields (Definitions \ref{def:NGth} and \ref{def:BVNGth})
\begin{equation}
    \mathcal{F}_{NG} := T^*[-1]\left( C^\infty(M,N) \times \mathfrak{X}[1](M)\right) \ni ((X,X^\dag),(\zeta, \zeta^\dag))
\end{equation}
endowed with the BV-Nambu--Goto action:
\begin{gather}
S_{NG}= S^{\text{cl}}_{NG} + \int_M \langle X^\dag,L_\zeta X \rangle + \langle\zeta^\dag,\frac12[\zeta,\zeta] \rangle ,
\end{gather}
does not induce a compatible BV-BFV structure when the worldsheet admits a boundary $\partial M$ (Theorem \ref{thm:NGNOGO}). 

On the other hand, we show that the reduced phase spaces of the two theories coincide by means of the communting diagram (Theorem \ref{thm:NGRPS}):
\begin{equation}\label{e:PtoNG}
    \xymatrix{
    F_{NG} \ar[d]^{\check{\pi}_{NG}} \ar[rr]^{\phi_{NG}} & & F_P\ar[d]_{\check{\pi}_{P}}\\
    \check{F}_{NG} \ar[dd]^{{\pi}^\partial_{NG}} \ar[dr]^{{\pi}^\partial_{\mathrm{partial}}} & & \check{F}_{P} \ar[d]^{{\pi}^\partial_P} \\
    & C_{P} \ar[d]^{\pi^{\text{red}}_P} \ar[r]^{\iota_C} & F^\partial_P\\
    F^\partial_{NG}\simeq \Phi^{\text{red}}_{NG} \ar[r]^-{\varphi} & \Phi^{\text{red}}_{P} &
    }
\end{equation}
where we denoted by $C_P\subset F^\partial_P$ the submanifold of constraints of Polyakov theory defined above, with $\iota_C\colon C_P \to F^\partial_P$ the inclusion map, and the maps 
$$
{\pi}^\partial_{\mathrm{partial}}\colon \check{F}_{NG}\to C_P; \qquad \phi_{NG}\colon F_{NG}\to F_{P}
$$ 
are, respectively, a partial pre-symplectic reduction of the pre-symplectic manifold $\check{F}_{NG}$, and $\phi_{NG}$ is the classical equivalence of the two theories (cf. Remark \ref{rem:classeq}), referring to the spaces $\check{F}_{NG}, F_{NG}$ and $F_P$, defined respectively in Theorem \ref{thm:NGRPS}, Definition \ref{def:NGth} and Definition \ref{def:Polth}.

\section{Background}
In this background chapter we introduce  some basic concepts needed throughout the paper. In Section \ref{Sec:RPS} we will describe how to construct the Reduced Phase Space for a field theory using a geometric construction due to Kijowski and Tulczyjew \cite{KT1979}. In Section \ref{Sec:BVintro}, we give a brief overview of the BV and BFV formalisms, as well as how they are related to each other on a manifold with boundary. We refer to \cite{CMR2012} for a more detailed discussion about these topics. In the Section \ref{sec:Strings}, we will outline the basics about the Nambu-Goto and the Polyakov string theories, and fix our conventions.

\subsection{A geometric approach to the reduced phase space}\label{Sec:RPS}

A field theory on a manifold $M$ is specified by a space of fields $F$ --- modeled around smooth sections of a fibre bundle\footnote{Later in this paper $F$ will be a mapping space, seen as the space of sections of some trivial fibre bundle on $M$. Since we will work with metrics, we will also need to allow (open) nondegeneracy conditions on fields.} $E\to M$ --- as well as an action functional $S: F \rightarrow \mathbb{R}$, a local functional of the form 
$$
S=\intl_M L[\phi,\partial_I\phi],
$$
with $L$ a density valued functional of fields and a finite number of derivatives (jets). The data of local symmetries for a field theory is specified by an involutive distribution $D\subset TF$. To construct the reduced phase space of the system, we use a method developed originally by Kijowski and Tulczyjew \cite{KT1979}. 

Assume that the manifold $M$ has a non-empty boundary $\partial M$. The starting point for the construction is the variation of the action functional, which splits into a bulk one-form $\mathsf{el}$, the vanishing locus of which is the Euler--Lagrange critical locus $EL$ of the theory, and a boundary term:
\begin{equation}
    \delta S = \mathsf{el} + \check{\pi}^*\check\alpha.
\end{equation}
If we denote by $\check{F}$ the space of \emph{pre-boundary} fields, i.e.\ the space of restrictions of fields and normal derivatives to the boundary, with the natural surjective submersion $\check{\pi}\colon F \to \check{F}$ given by restriction, we can interpret $\check{\alpha}$ as a one-form on $\check{F}$. Given $\check{\alpha}$ we construct the two-form $\check{\omega}=\delta \check{\alpha}$.

The closed two-form  $\check{\omega}$ is often degenerate, so $(\check{F},\check{\omega})$ is at best pre-symplectic. When this is the case, that is if the kernel of $\check{\omega}^\sharp:T\check{F}\to T^*\check{F}$ is regular\footnote{Regular means that $\mathrm{ker}(\check{\omega})$ is a subbundle of $T\check{F}$. A practical way to check this is whether it has the same $C^\infty(\check{F})$-dimension over all of $\check{F}$.}, we can perform pre-symplectic reduction over the space of boundary fields: 
$$
{\pi}^{\text{red}}: \check{F} \rightarrow F^\partial = \check{F}/\mathrm{ker}(\check{\omega}^\sharp),
$$ 
and, precomposing, we get $\pi:= \pi^{\text{red}}\circ\check{\pi}\colon F\to F^\partial$, with symplectic structure $\omega^\partial = \check{\underline{\omega}}$. 
\begin{definition}
We will call $(F^\partial:=\check{F}/\mathrm{ker}(\check{\omega}^\sharp),\omega^\partial)$ the \emph{geometric phase space} of the classical field theory $(F,S,D)$.
\end{definition}

Not all points in $F^\partial$ can be extended to a solution of the equations of motion in the bulk. Typically, the set of such boundary values is specified by the common zero locus $C$ of a set of functions $\{\phi_i\}_{i=1\dots k}$. We can think of $C$ as a generalisation of Cauchy data for the field theory, in that we require points on $C$ to extend (possibly non-uniquely) to a solution of the equations of motion in a short enough cylinder bounded by $\partial M$. If $\partial M$ is a Cauchy surface, this translates into the usual Cauchy problem for PDE's, more generally $C$ encodes necessary conditions for existence and uniqueness. However, in this approach we are not specifically interested in the analytic nuances surrounding this question.

Typically, a na\"ive choice for the $\phi_i$-s is induced by restricting the equations of motion of the theory to the boundary. This produces a set of functions $\check{\phi}_i\in C^\infty (\check{F})$, which we expect to be basic with respect to the reduction $\pi^{\text{red}}$. This means that there exist functions $\phi_i\in F^\partial$, such that $\check{\phi_i}=\pi^{\text{red}*}\phi_i$.

In what follows, we will assume that the set $\{\phi_i\}_{i=1\dots k}$ is in involution, i.e.\ all Poisson brackets $\{ \phi_i,$ $\phi_j \}$ between constraints vanish when restricted to $C$. We will then say that $C$ is coisotropic\footnote{In order for $C$ to properly define a submanifold in an infinite-dimensional setting, certain additional regularity conditions are required to hold. We will not be concerned with this issue, as we will be more interested in the algebraic properties of $C$.}. In Dirac's terminology \cite{Dirac1958}, the constraint set $\{\phi_i\}_{i=1\dots k}$ is \emph{first class}. This is relevant because, in order for the field theory to be well-defined, one requires $\pi(EL)\subset F^\partial$ to be a Lagrangian submanifold, when $EL$ is associated to a small-enough cylinder bounding $\partial M$. A direct consequence of this is that $C\supset \pi(EL)$ must be coisotropic.

The restriction of the symplectic 2-form $\omega^\partial$ to $C$ is degenerate, and since $C$ is coisotropic, its kernel is given by the span of the Hamiltonian vector fields $X_i$ of the functions $\phi_i$, this is also called the characteristic distribution of $C$. 

\begin{definition}
We define the \emph{Reduced Phase Space} (RPS) of the field theory to be the reduction of $C$ by its characteristic distribution: $\Phi^{\text{red}}:=\underline{C}$.
\end{definition} 

The reduced phase space is generally a singular space\footnote{Often $C$ itself is singular, see e.g. \cite{ArmsMarsdenMoncrief:1982}, for the case of General Relativity.}, we will describe it alternatively by means of the constraint set $\{\phi_i\}_{i=1\dots k}$. As we will see in the next section, the BFV construction for a field theory provides a resolution of the Reduced Phase Space, i.e.\ a complex whose cohomology in degree zero describes (or rather replaces) the space of functions over the reduced phase space $\Phi^{\text{red}}$.

\subsection{Batalin--Vilkovisky and Batalin--Fradkin--Vilkovisky formalisms}\label{Sec:BVintro}
In this section we present a brief overview of the Batalin--Vilkovisky and the Batalin--Fradkin--Vilkovisky formalisms \cite{BV1,BV2,BV3}. 

\begin{definition}\label{def:BV}
A relaxed BV-theory on a manifold $M$ is the data
$$(\mathcal{F}_M,S_M,Q_M,\Omega_M)$$ 
with $(\mathcal{F}_M,\ {}\Omega_M)$ a Z-graded $(-1)$-symplectic manifold, $S_M$ a degree $0$ function, and $Q_M$ a degree $1$ cohomological vector field, i.e.\ such that $[Q_M, Q_M]=0$. If, in addition, we have
\begin{align}
\iota_{Q_M} \Omega_M &= \delta S_M,
\end{align}
i.e.\ $ S_M$ is the Hamiltonian function of $Q_M$, the data $(\mathcal{F}_M,S_M,Q_M,\Omega_M)$ defines a BV theory.
\end{definition}

\begin{remark}
Notice that, in a BV theory, the compatibility requirements above can be rewritten in a more familiar way as the statement that $S_M$ satisfies the Classical Master Equation (CME)
\begin{gather}\label{03/05A1}
\{S_M, S_M\}=0
\end{gather}
where the Poisson brackets are derived from the graded symplectic structure. In some circumstances it is useful to identify $Q_M=\{S_M,\cdot\}$, although, since the Hamiltonian condition will be spoiled in the presence of a boundary, we prefer to think of the two pieces of data as independent, and consider \emph{relaxed} BV theories.
\end{remark}

We provide here a general definition of a cotangent bundle for space of fields we will use to define BV fields throughout.

\begin{definition}[Cotangent bundles]\label{def:cotangent}
Let $E\to M$ be a (possibly graded) vector bundle, $\mathcal{E}[p]$ its space of ($p$-shifted) smooth sections\footnote{We assume to be working in some convenient setting, such as Fr\'echet manifolds, or more generally diffeological spaces.}, i.e.\ the space of sections of $E[p]$. We define by $T^*[k]\mathcal{E}[p]$ the vector bundle whose fibres are given by ($k-p$)-shifted sections\footnote{Here $\mathrm{Dens}(M)$ denotes the density \emph{bundle} of $M$.} of $E^*\otimes \mathrm{Dens}(M) \to M$.

Let $C^\infty(M,N)$ be the space of smooth maps between smooth manifolds $M$ and $N$. The tangent bundle $T C^\infty(M,N)$ is given by pairs $(X,V)$ of a smooth map $X\in C^\infty(M,N)$ and a section $V\in \Gamma(X^*TN)$. We define $T^*[k]C^\infty(M,N)$ to be the vector bundle over $C^\infty(M,N)$ whose fibres at $X$ consist of sections $X^\dag \in \Gamma(X^*T^*[k]N\otimes \mathrm{Dens}(M))$.
\end{definition}

Typically, one is given the data of a field theory as in Section \ref{Sec:RPS}, and wants to extend it to a BV theory. A classical field theory is specified by the data $({F}_M, S^{\text{cl}}_M)$ and by a distribution $D_M\subset \text{T} {F}_M$ encoding the symmetries, i.e.\ vector fields $X\in\Gamma(D_M)$ such that $L_X(S_M^{\text{cl}})=0$. A relevant detail here is that, while the space of (local) symmetries of $S^{\text{cl}}_M$ is fixed (and it is a Lie subalgebra of $TF_M$), the distribution $D_M$ encodes only the symmetries of the theory that are non-trivial, i.e.\ vector fields that do not vanish on the critical locus of $S_M^{\text{cl}}$. Hence, there is a freedom in the choice of $D_M$, which is thus a datum we have to specify, and in principle might only be involutive up to trivial symmetries (in this case one speaks of on-shell symmetries, see \cite{Henn} for more details on this classical issue). In this paper we will mostly employ the classic result:
\begin{theorem}[\cite{BV1,BV2}]\label{thm:minBVext}
Let $D_M$ be the image of a Lie algebra action $\rho:\mathfrak{g} \to \mathfrak{X}(F_M)$, and let $Q_{BRST}$ be the Chevalley-Eilenberg differential\footnote{We consider $C^\infty(F_M)$ as a $\mathfrak{g}$-module.} associated with $\rho$. Consider the space of fields $\mathcal{F}_M= T^*[-1] ( D_M [1] )$, where $\Phi$ is a multiplet of fields in $D_M[1], \Phi^\dag$ denotes the corresponding multiplet of conjugate (anti-)fields. Let us define the functional 
\begin{equation}
    S_M=S_M^{\text{cl}}+\langle\Phi^\dag, Q_{BRST} \Phi\rangle,
\end{equation} 
and extend $Q_{BRST}$ to $Q_M$ so that, up to boundary terms, $\iota_{Q_M}\Omega_M = \delta S_M$. Then, the data $(\mathcal{F}_M,\Omega_M,S_M,Q_M)$ denotes a relaxed BV theory. If $\partial M=\emptyset$, the data defines a BV theory.
\end{theorem}

In this article, we are interested in the case where $M$ has a boundary $\partial M$. In this case, the relation $\iota_{Q_M} \Omega_M = \delta S_M$ does not hold anymore, i.e.\ the BV theory on the manifold without boundary generalises to the case with boundary as a relaxed BV theory, and the following holds instead:
\begin{gather}\label{04/05A}
\iota_{Q_M} \Omega_M = \delta S_M + \tilde{\pi}^*\check{\alpha}
\end{gather}\\
where $\tilde{\pi}:\mathcal{F}_M \rightarrow \check{\mathcal{F}}_{\partial M}$ is a surjective submersion from the space of fields on $M$ to the phase space of pre-boundary fields, once again defined as the space of restrictions of fields and normal jets to the boundary. The boundary term $\check{\alpha}$ is interpreted as a local 1-form on $\check{\mathcal{F}}_{\partial M}$. 

To the boundary of the manifold we can associate the following structure.
\begin{definition}\label{def:BFV}
A BFV-theory on a closed manifold $N$ is the data $(\mathcal{F}^\partial_N, S^\partial_N, Q^\partial_N, \Omega^\partial_N)$ with $(\mathcal{F}^\partial_N, \Omega^\partial_N)$ a $\mathbb{Z}$-graded $0$-symplectic manifold, and $S_N$ and $Q_N$ respectively a degree $1$ function and a degree $1$ vector field on $\mathcal{F}^\partial_N$ such that: 
\begin{align}
[Q_N, Q_N] &= 0 \\
\iota_{Q_N} \Omega_N &= \delta S_N
\end{align}
i.e.\ $Q_N$ is a cohomological vector field, and $S_N$ its Hamiltonian function. This implies that $S_N$ satisfies the CME. If $\Omega_N=\delta \alpha_N$, we will say that the BFV theory is \emph{exact}.
\end{definition} 

The BFV data is related to the relaxed BV theory on the bulk (for more details on the procedure, check \cite{CMR2012}). Indeed, through eq. $\eqref{04/05A}$, the relaxed BV theory induces some data on the boundary, which in good situations will yield a BFV theory. In that case we have:
\begin{definition}\label{def:BVBFV}
An exact BV-BFV theory on a manifold with boundary $(M,\partial M)$ is the data 
$$
(\mathcal{F}_M, S_M, Q_M, \Omega_M, \mathcal{F}_{\partial M}^\partial, S_{\partial M}^\partial, Q_{\partial M}^\partial, \Omega_{\partial M}^\partial, \pi)
$$ such that $(\mathcal{F}_M, S_M, Q_M, \Omega_M)$ is a relaxed BV theory on the manifold with boundary $(M,\partial M)$, $(\mathcal{F}_{\partial M}^\partial, S_{\partial M}^\partial, Q_{\partial M}^\partial, \Omega_{\partial M}^\partial = \delta \alpha_{\partial M}^\partial)$ is an exact BFV theory on the boundary $\partial M$, and  $\pi : \mathcal{F}_M \rightarrow \mathcal{F}_{\partial M}^\partial$ is a surjective $Q$-submersion, such that
\begin{equation}
    \iota_{Q_M} \Omega_M = \delta S_M + \pi^*  \alpha_{\partial M}^\partial. 
\end{equation}

A relaxed BV theory $(\mathcal{F}_M, S_M, Q_M, \Omega_M)$ on a manifold with boundary $(M,\partial M)$ is said to be \emph{$1$-extendable} to an exact BV-BFV theory if one can find an exact BFV theory $(\mathcal{F}_{\partial M}^\partial, S_{\partial M}^\partial, Q_{\partial M}^\partial, \Omega_{\partial M}^\partial=\delta \alpha_{\partial M}^\partial)$ on the boundary $\partial M$ and a map $\pi\colon \mathcal{F}_M \to \mathcal{F}_{\partial M}$ such that one obtains an exact BV-BFV theory on $(M,\partial M)$.
\end{definition}

Given a classical field theory on a manifold with boundary, we can always compute its Reduced Phase space $\Phi^{\text{red}}$. Under the mild assumption that $\Phi^{\text{red}}$ be given by a set $C$ of first class constraints\footnote{This condition can be generalised to the more general request of having a coisotropic submanifold \cite{Schaetz:2008}.}, we can apply the BFV construction to it and obtain a cochain complex whose cohomology in degree zero is a replacement for the space of smooth functions on the reduced phase space $\underline{C}$. The output of the BFV construction \cite{BV3}, which was explained in detail in \cite{Stasheff1997} and then later in \cite{Schaetz:2008}, is a BFV theory as in Definition \ref{def:BFV}.

On the other hand, given the same field theory, we can apply Theorem \ref{thm:minBVext} and compute its BV complex. When the two independently constructed complexes are compatible, they yield a BV-BFV theory, which is the starting point of Perturbative quantisation with boundary, as proposed by Cattaneo, Mnev and Reshetikhin in \cite{CMR2}.

\subsection{String action functionals}\label{sec:Strings}

In String Theory the generalisation of the motion of point particles is done by considering extended one-dimensional objects. Their classical dynamics is dictated by the minimization of an action functional, which can be obtained by generalizing that of relativistic free particles:
\begin{gather}\label{NewPolA1}
S_{1d} := \int_I \sqrt{\mathsf{g}}dt
\end{gather}
where $I \subset \mathbb{R}$ is an interval parametrized by ${x}\in I$,  $g_{\alpha \beta} := \partial_\alpha X^\mu \partial_\beta X^\nu G_{\mu \nu} (X)$ is metric induced (in components) on the line, with $\mathsf{g}:=|\mathrm{det} ( g_{\alpha \beta} )|$ its determinant  and $X : I \rightarrow N$ is the trajectory of the particle\footnote{$X^\mu$ denotes the composition of a time-like curve $X : I \rightarrow N$ and a local chart $\phi^\mu: U\subset N \rightarrow \mathbb{R}^{d+1}$, for a smooth manifold $N$.}.

This action functional is the pseudo-length of the line in $N$ with metric $G$ (that in local coordinates is $G_{\mu \nu} (X)$), and by minimizing it we find the equations governing the motion of a particle in $N$. Since $I$ is one-dimensional, the induced metric $g = g_{t t} dt dt$ has one component $g_{t t} = \partial_t X^\mu \partial_t X^\nu G_{\mu \nu} (X) =: \dot{X}^2$, and the action of the relativistic particle can be written in the more commonly known form (up to factors):
\begin{gather}
S_{1d} := \int_I  | \dot{X}|dt
\end{gather}

Following the same philosophy, we write the action of the string as the surface area spanned by a 1-dimensional ``string'' moving in a background pseudo-Riemannian geometry: the two dimensional generalization of the path length.  
\begin{definition}\label{def:NGth}
We call Nambu--Goto theory the assignment, to a 2-dimensional source manifold $M$ (possibly with boundary) that admits a Lorentzian structure\footnote{A compact manifold admits a Lorentzian structure if and only if it has vanishing Euler characteristic. In two dimensions this is, for example, the case of a torus.} and a $d+1$ Lorentzian manifold $(N,G)$, of the space of classical fields
\begin{equation}
    F_{NG} = C^\infty (M,N) \ni X
\end{equation}
and the Nambu--Goto action:
\begin{equation}\label{NewPolA2}
S^{\text{cl}}_{NG} := \int_{M}  \sqrt{\mathsf{g}}d^2 {x} = \int_{M} d^2 {x} \sqrt{ |\mbox{det}(\partial_\alpha X^\mu \partial_\beta X^\nu G_{\mu \nu} (X))|}
\end{equation}
In the context of strings, $M$ is often referred to as ``world sheet". 
\end{definition}

To describe a string, we can alternatively use the following model \cite{BrinkdiVecchia,POLYAKOV1981207}: 
\begin{definition}\label{def:Polth}
We call \emph{non-null Polyakov theory} the assignment, to a 2-dimensional manifold with boundary $M$ that admits a Lorentzian structure and a $d+1$ Lorentzian manifold $(N,G)$, of the data 
\begin{equation}
    F_{P}=C^\infty (M,N)\times \mathcal{PR}(M,\partial M) \ni \{X,h\},
\end{equation}
where $\mathcal{PR}(M,\partial M)$ denotes the open set of Lorentzian metrics on $(M,\partial M)$ whose restriction to $\partial M$ is nondegenerate, together with the Polyakov action functional:
\begin{align}
S^{\text{cl}}_{P} & :=\int_{M} \langle d X, \star_h d X\rangle = \int_{M} d X^\mu *_h d X^\nu G_{\mu \nu} (X) \\\label{NewPolA3} 
&= \int_{M} \sqrt{\mathsf{h}}\, h^{\alpha \beta} \partial_\alpha X^\mu \partial_\beta X^\nu G_{\mu \nu} (X) d^2{x}
\end{align} 
where $\star_h$ is the Hodge dual defined by $h, \sqrt{\mathsf{h}}$ denotes the square root of the (absolute value of the) determinant of $h$ and $\langle\cdot,\cdot\rangle$ the inner product defined by $G$.
\end{definition}

Notice that non-null Polyakov theory enjoys a \emph{conformal symmetry}, i.e.\ we can rescale $h\to \lambda h$ with a positive function $\lambda$. This alternatively means that we can reduce the number of degrees of freedom from the start, without loss of generality. We consider new variables given by equivalence classes of metrics ${\th}:= [h]$ under rescaling, which can be parametrised as ${\th}:=\frac{1}{\sqrt{\mathsf{h}}} h$, with $\mathsf{h}\coloneqq \mathrm{det}(h)$ and inverse ${\th}^{-1}={\sqrt{\mathsf{h}}} h^{-1}$. 

\begin{definition}\label{def:DPR}
We define the space of \emph{densitised} Lorentzian metrics $\mathcal{DPR}(M,\partial M)$, given by equivalence classes of metrics, parametrised by ${\th}=\frac{1}{\sqrt{\mathsf{h}}}h\in $ for $h\in \mathcal{PR}(M,\partial M)$ a Lorentzian metric with nondegenerate restriction to $\partial M$.
\end{definition}

Indeed, Polyakov action depends explicitly only on combinations $\sqrt{\mathsf{h}} h^{\alpha \beta}$ (see eq.\eqref{NewPolA3}): this quantity has only two degrees of freedom instead of three (in fact ${\th}^{\alpha \beta}$ has unit determinant). As a consequence, $S^{\text{cl}}_P$ descends to the space of equivalence classes of metrics defined by rescalings.

In a local chart, the elementary field ${\th}_{\alpha \beta}$ (and its inverse ${\th}^{\alpha \beta}$) has matrix representation:
\begin{equation}\label{A6}
{\th}_{\alpha \beta} :=
\begin{pmatrix}
{\th}_{nn} &  {\th}_{nt}\\
{\th}_{nt} & {\th}_{tt}

\end{pmatrix}
= \varsigma
\begin{pmatrix}
{\th}^{tt} & - {\th}^{nt}\\
-{\th}^{nt} & {\th}^{nn}

\end{pmatrix} \qquad \varsigma:= \mathrm{det}(\tilde{h}) \equiv \text{sign}(\mathsf{h})
\end{equation} 
and $n$ and $t$ are the indexes respectively of the normal and tangent directions to the boundary\footnote{This means, in a local chart adapted to a tubular neighborhood of the embedded boundary submanifold.}. In what follows we will be mainly interested in the case $\varsigma=-1$, but we will keep track of $\varsigma$ throughout most of the calculations, and it will be specified when we will do otherwise. The main formulas we will need to tackle the variational problem for densitised metrics are:
\begin{subequations}\begin{align}
{\th}_{\alpha \beta} := \frac{1}{\sqrt{\mathsf{h}}} h_{\alpha \beta}\\
|\mbox{det} ({\th}_{\alpha \beta})| = 1\\
\delta {\th}_{\alpha \beta} = \frac{1}{\sqrt{\mathsf{h}}} P^{\perp \rho \sigma}_{\alpha \beta}  \delta h_{\rho \sigma}\label{14/04/2020A1}
\end{align}\end{subequations}
where $P^\perp = \mathrm{id} - \frac12 h \mathrm{Tr}_h$ is a operator on $\Gamma(S^2(TM))$ that in a local chart reads
\begin{equation}\label{e:projector}
    P^{\perp \rho \sigma}_{\alpha \beta} = \delta_\alpha^\rho \delta_\beta^\sigma - \frac{1}{2} h_{\alpha \beta}h^{\rho \sigma}
\end{equation}
and that, pointwise, represents the projector to the subspace of traceless symmetric tensors.

\begin{remark}
Equation \eqref{14/04/2020A1} relates the variation of the constrained variables ${\th}_{\alpha \beta}$ to the variation of the free variables $h_{\alpha \beta}$. 
Observe that $\delta {\th}_{\alpha \beta} K^{\alpha \beta}= 0$ does not imply $K^{\alpha \beta}=0$, for $K^{\alpha \beta}$ a generic symmetric contravariant two-tensor. 
The condition $\delta {\th}_{\alpha \beta} K^{\alpha \beta}= 0$ imposes only two independent relations on $K^{\alpha \beta}$. To tackle the issue, we can use eq \eqref{14/04/2020A1}:
\begin{gather}K^{\alpha \beta} \frac{1}{\sqrt{\mathsf{h}}} P^{\perp \rho \sigma}_{\alpha \beta}  \delta h_{\rho \sigma} = 0 
\end{gather}
which implies:
\begin{equation}
    K^{\alpha \beta} P^{\perp \ {} (\rho \sigma)}_{\alpha \beta}  = 0,
\end{equation}
and the symmetrization of the projector, denoted by round brackets around the indices, is due to the symmetry of $h_{\alpha \beta}$. With these considerations the action functional for non-null Polyakov theory reads:
\begin{equation}\label{A7}
{S}^{\text{cl}}_P = \int_{M} {\th}^{\alpha \beta} \partial_\alpha X^\mu \partial_\beta X^\nu G_{\mu \nu} (X) d^2{x}.
\end{equation}
\end{remark}

\begin{definition}\label{def:T^*DPR}
We define the pre-cotangent bundle $T^\vee\mathcal{DPR}(M,\partial M)$ to be the vector bundle over $\mathcal{DPR}(M,\partial M)$ whose fibre over $\th$ is space of sections $K$ of $S^2(TM)$ subject to the condition $P^\perp_{\th}(K)=0$. Denoting the associated vector bundle by $S^2_{P^\perp_{\th}}(TM)\to M$, we define the cotangent bundle $T^*\mathcal{DPR}(M,\partial M)$ as the vector bundle whose fibres at $\th$ are sections of $S^2_{P^\perp_{\th}}(TM)\otimes \mathrm{Dens}(M) \to M$.
\end{definition}

\begin{remark}
In principle, one could use the standard, nonreduced, version of Polyakov theory, and account for conformal invariance within the BV formalism. We found this path more convenient to deal with. The adoption of these reduced fields appears in \cite[Eq 17]{BAULIEU198655} and is discussed, even if it is ultimately not employed, in \cite[Eq 3.8]{LEE1988191}.
\end{remark}

\begin{remark}
Throughout the paper we will assume that our metric fields are defined in such a way that they restrict to nondegenerate metrics on the boundary. In Appendix \ref{AppendixLightlike} we show that this is equivalent to the condition $h^{nn}\not = 0$. Typically, this is stated as the condition that the boundary is \emph{non-null}, although clearly it is an open condition in the space of fields, since the metric is dynamical and not a background.
\end{remark}

\begin{remark}[{Classical Equivalence}]\label{rem:classeq}
A known (and simple) fact about Polyakov and Nambu--Goto string theories is that they are classically equivalent. Physically, this means that they describe the same set of classical configurations: the solutions of the respective Euler--Lagrange equations of motion (modulo symmetries). Geometrically, we can phrase this as follows. Consider the equation of motion $\frac{\delta S^{\text{cl}}_P}{\delta h} = 0$, given by the tensorial equation\footnote{This is the vanishing of the stress-energy tensor.}
\begin{equation}\label{e:polngequivalence}
        dX^\mu \odot d X_\mu = \frac12 h \mathrm{Tr}_h[dX^\mu \odot d X_\mu],
\end{equation}
where $\odot$ defines the symmetries tensor product, and denote by $N\subset F_P$ the subset of Polyakov fields that satisfy it. Then there is a diffeomorphism $\phi_{NG}\colon F_{NG} \to N\subset F_P$, such that
$$
\phi_{NG}^* S^{\text{cl}}_P = S^{\text{cl}}_{NG},
$$
since we can realise that the square root of the determinant of the two sides of \eqref{e:polngequivalence} are Nambu-Goto's and Polyakov's Lagrangians.
This induces an equivalence of field theories, since clearly the zero locus of $S^{\text{cl}}_{NG}$ coincides with that of $S^{\text{cl}}_P$. We call this relationship between theories ``classical equivalence''.  
In general, though, classical equivalence is not enough to show that the cohomology of the BV complexes associated to two (classically equivalent) models are isomorphic, as it only statens that the cohomologies in degree $0$ coincide. However, the equivalence between Polyakov and Nambu--Goto theories fits within a general approach to the comparison of field theories, which stems on the elimination of (generalised) \emph{auxiliary fields}. In this case, the auxiliary field is the metric $h$, featuring in Polyakov theory. When two theories differ only in auxiliary field content, their BV cohomologies are related by the classic results of Barnich, Brandt and Henneaux \cite{BBH,Henn}. In this work we will show that regardless of the ``nice'' behaviour of the classical equivalence map $\phi_{NG}$, the BV-BFV behaviour of the two models differs. In particular we will show that the BV data for Nambu--Goto string does not induce a BV-BFV theory on a manifold with boundary. This result seems to be compatible with the work of Bahns, Rejzner and Zahn \cite{BRZ}, who found another discrepancy between Polyakov's and Nambu--Goto's BV cohomologies, despite the general arguments of \cite{BBH,Henn}.
\end{remark}

\section{Polyakov Theory --- Reduced Phase Space}\label{Sec:RPSPOL}
In this section we will analyze the reduced phase space (RPS) of the Polyakov action. We begin by recalling that the the natural restriction of fields to the boundary yields a pre-symplectic manifold. Its reduction yields the space of boundary fields $F^\partial_P$, within which we identify the coisotropic submanifold defining the RPS as follows.

\begin{theorem}\label{THM:PolyakovRPS}
The geometric phase space for Polyakov theory is the cotangent bundle
\begin{equation}
    F^\partial_P = T^* C^\infty(\partial M, N) \ni \{J, X\},
\end{equation}
with canonical symplectic form
\begin{equation}\label{eq:Polbdryform}
    \omega^\partial_P = \intl_{\partial M} \delta X^\mu \delta J_\mu.
\end{equation}
The Reduced Phase Space for Polyakov theory is represented by the common zero locus of the functions 
\begin{subequations}\label{eq:PolConstraints}\begin{align}
H_\phi & \coloneqq\intl_{\partial M} \phi\left(\partial_t X_\mu \partial_t X^\mu  - \varsigma J_\mu J^\mu\right)\\
L_\psi & \coloneqq  2 \intl_{\partial M} \psi \partial_t X^\mu J_\mu,
\end{align}\end{subequations}
where $\partial_t$ denotes derivatives tangential to $\partial M$ and $\varsigma=\mathrm{sign}(\mathsf{h})$, which satisfy:
\begin{subequations}\begin{align}
\{H_\phi, H_{\phi '}\} & = L_{[\phi, \phi ']}\label{21/04A5} \\
\{L_{\psi}, L_{\psi '}\} & = L_{[\psi, \psi ']}\label{21/04A6} \\
\{H_{\phi}, L_{\psi}\} & = H_{[\phi, \psi]}\label{21/04A7}
\end{align}\end{subequations}
where $[\phi, \psi]:= (\partial_t \phi) \psi - \phi (\partial_t \psi)$.
\end{theorem}

\subsection{Proof of Theorem \ref{THM:PolyakovRPS}}
We will split this proof in four parts. First we will show that the variation of Polyakov's action functional induces a pre-symplectic form on the boundary. Then we will compute the pre-symplectic reduction, and finally we will analyse the structure of canonical constraints. 

\begin{proof}[Proof. Part 1]
Recall that $d$ is the de Rham differential over the world sheet manifold $M$, and thus acts over the k-forms over $M$ ($ d : \Omega^k (M) \rightarrow \Omega^{k+1} (M)$), while $\delta$ is interpreted as the vertical differential on local forms on ${F}_{P}$. 

Let us compute the variation of Polyakov's action functional and split it into a boundary term, and the Euler--Lagrange term:
\begin{equation} \label{A2}
\begin{aligned}
\delta S^{\text{cl}}_P = \int_{M} \delta( d X^\mu& *_{{\th}} d X^\nu G_{\mu \nu} (X) )\\
=\int_{M} \bigg[d( 2 \delta X^\mu *_{{\th}} d X^\nu G_{\mu \nu} (X) ) & - 2 \delta X^\mu d( *_{{\th}} d X^\nu G_{\mu \nu} (X) )   \\
+ d X^\mu \delta(*_{{\th}}) d X^\nu G_{\mu \nu} (X) +& d X^\mu *_{{\th}} d X^\nu \delta( G_{\mu \nu} (X) ) \bigg]
\end{aligned}
\end{equation} 
where we used: 
\begin{equation} \label{A3}
\delta d= d \delta, \qquad \delta(*_{{\th}}) d X^\nu :=\delta( \sqrt {\mathsf h} h^{\alpha \beta}) \varepsilon_{\beta \gamma} \partial_\alpha X^\nu d {x}^\gamma,
\end{equation}
with $\varepsilon_{\alpha\beta}$ denoting the antysimmetric Levi--Civita symbol.
We can split \eqref{A2} into a term $EL$, yielding the Euler--Lagrange equations of motion:
\begin{equation}\label{A4}
EL:=\int_{M} \bigg[- 2 \delta X^\mu d[ *_{{\th}} d X^\nu G_{\mu \nu}] + 
 d X^\mu \delta(*_{{\th}}) d X^\nu G_{\mu \nu} + d X^\mu *_{{\th}} d X^\nu \delta G_{\mu \nu}  \bigg]
\end{equation}
and a boundary term, which we can interpret as a one-form $\check{\alpha}$:
\begin{equation}\label{A8}
\check{\alpha}= \int_{\partial M} 2 \delta X^\mu *_{{\th}} d X^\nu G_{\mu \nu} (X) = \int_{\partial M} \delta X^\mu \bigg[ {\th}^{n  \alpha} \partial_{\alpha} X^\nu G_{\mu \nu} (X) \bigg] d {x}^{t}
\end{equation}
over the space of pre-boundary fields $\check{{F}}_{P}$, defined as the restriction of fields and normal jets (i.e.\ $\partial_n X^\mu$) to the boundary. There is a natural surjective submersion $\check{\pi}\colon F_P \to \check{F}_P$, so that
$$
\delta S_P = \check{\pi}^*\check{\alpha} + EL.
$$

Observe that $\partial_n X^\mu$ is henceforth considered to be an independent field. Let us compute the pre-boundary two-form $\check{\omega}:=\delta \check{\alpha}$
\begin{equation}\label{A9}
\begin{aligned}
\check\omega= \int_{\partial M} &\delta X^\mu \bigg[ - {\th}^{n \rho} {\th}^{\sigma \beta} (\delta {\th}_{\rho \sigma}) \partial_\beta X^\nu G_{\mu \nu} (X) + \\
&{\th}^{n \beta} (\delta \partial_\beta X^\nu) G_{\mu \nu} (X) + {\th}^{n \beta} \partial_\beta X^\nu \frac{\partial G_{\mu \nu} (X)}{\partial X^\rho}  \delta X^\rho \bigg].
\end{aligned}
\end{equation}
We would like to compute the reduction $\check{F}_{P}/\mathrm{ker}(\check{\omega})$ by the kernel of the pre-boundary two-form. This is allowed when $\check{\omega}$ is pre-symplectic, that is to say when its kernel is a regular subbundle. To compute the kernel, we consider the general vector:
\begin{equation}\label{A10}
\mathbb{X}:= \intl_{\partial M}(\mathbb{X}_X)^\mu  \frac{\delta}{\delta X^\mu} + (\mathbb{X}_{\partial_n X})^\mu  \frac{\delta}{\delta \partial_n X^\mu} + (\mathbb{X}_{{\th}})_{{\alpha \beta}} \frac{\delta}{\delta {\th}_{\alpha \beta}}
\end{equation}
and look at the equation $\iota_{\mathbb{X}} \check\omega=0$:
\begin{multline}\label{A11}
\iota_X \check\omega = \int_{\partial M}  \bigg[ - (\mathbb{X}_X)^\mu  \bigg( - {\th}^{n \rho} {\th}^{\sigma \beta} (\delta {\th}_{\rho \sigma}) \partial_\beta X^\nu G_{\mu \nu} (X) + \\
{\th}^{n \beta} (\delta \partial_\beta X^\nu) G_{\mu \nu} (X) + {\th}^{n \beta} \partial_\beta X^\nu \frac{\partial G_{\mu \nu} (X)}{\partial X^\rho}  \delta X^\rho \bigg)+\\
\delta X^\mu \bigg( - {\th}^{n \rho} {\th}^{\sigma \beta} (\mathbb{X}_{{\th}})_{{\rho \sigma}} \partial_\beta X^\nu G_{\mu \nu} (X) + {\th}^{n n} (\mathbb{X}_{\partial_n X})^\nu   G_{\mu \nu} (X) + \\
{\th}^{n t}  \partial_t(\mathbb{X}_X)^\nu  G_{\mu \nu} (X) + {\th}^{n \beta} \partial_\beta X^\nu \frac{\partial G_{\mu \nu} (X)}{\partial X^\rho}  (\mathbb{X}_{X})^{\rho} \bigg)
 \bigg]=0.
\end{multline}
Reordering the terms and integrating by parts the term that contains\footnote{Recall that the index $t$ denotes directions that are tangent to the boundary.} $\delta \partial_t X^\nu = \partial_t \delta X^\nu$, we obtain the conditions:

\begin{subequations}\label{A12}
\begin{align} 
\delta \partial_n X^\nu\colon \qquad &  \bigg((\mathbb{X}_X)^\mu   G_{\mu \nu} (X) {\th}^{n n} \bigg)  = 0\label{A12B0}\\
\delta {\th}_{\rho \sigma}\colon\qquad &  \bigg( (\mathbb{X}_X)^\mu   \partial_\beta X^\nu G_{\mu \nu} (X) {\th}^{n \rho} {\th}^{\sigma \beta} \bigg)  =0 \label{A12B1}\\
\delta X^\mu\colon\qquad & \bigg( - {\th}^{n \rho} {\th}^{\sigma \beta} (\mathbb{X}_{{\th}})_{{\rho \sigma}} \partial_\beta X^\nu G_{\mu \nu} (X) + {\th}^{n n} (\mathbb{X}_{\partial_n X})^\nu   G_{\mu \nu} (X) + \notag\\
& {\th}^{n t}  \partial_t (\mathbb{X}_X)^\nu  G_{\mu \nu} (X) + {\th}^{n \beta} \partial_\beta X^\nu \frac{\partial G_{\mu \nu} (X)}{\partial X^\rho}  (\mathbb{X}_{X})^{\rho} + \notag\\
&\partial_t ({\th}^{n t} G_{\mu \nu} (X) (\mathbb{X}_X)^\nu  ) - {\th}^{n \beta} \partial_\beta X^\nu \frac{\partial G_{\rho \nu} (X)}{\partial X^\mu} (\mathbb{X}_{X})^{\rho} \bigg)   =0 \label{A12B2}
\end{align}
\end{subequations}
From Equation \eqref{A12B0}, we see that if ${\th}^{n n} \ne 0$ then $ (\mathbb{X}_X)^\mu  =0$. Our Definition \ref{def:Polth} of non-null Polyakov theory only allows this possibility. In fact, as shown in detail in Appendix $\ref{AppendixLightlike}$, the condition ${\th}^{n n}=0$ means that we are dealing with a null/light-like boundary, i.e.\ it is equivalent to a degenerate induced metric on the 1-dimensional boundary. We are indeed considering the case where the boundary is nowhere light-like. 

Equation \eqref{A12B1} is a coinsequence of \eqref{A12B0}, while \eqref{A12B2} becomes:
\begin{equation}\label{A13}
(\mathbb{X}_{\partial_n X})^\mu  = \frac{1}{{\th}^{n n}} {\th}^{n \rho} {\th}^{\sigma \beta} (\mathbb{X}_{{\th}})_{{\rho \sigma}} \partial_\beta X^\mu.
\end{equation}
This implies that the kernel $\mathrm{ker}(\check{\omega}^\sharp)$ is a regular subbundle in $T\check{F}_P$, and it is described by the following conditions: 
\begin{gather}\label{A10B1}
 ( (\mathbb{X}_X)^\mu ,  (\mathbb{X}_{\partial_n X})^\mu ,  (\mathbb{X}_{{\th}})_{{\alpha \beta}}) = ( 0,   \frac{1}{{\th}^{n n}} {\th}^{n \rho} {\th}^{\sigma \beta} (\mathbb{X}_{{\th}})_{{\rho \sigma}} \partial_\beta X^\mu, (\mathbb{X}_{{\th}})_{{\alpha \beta}}),
\end{gather}
with $(\mathbb{X}_{{\th}})_{{\alpha \beta}}$ free, so that a generic kernel vector reads:
\begin{equation} \label{A14}
\mathbb{X}= \intl_{\partial M} \frac{1}{{\th}^{n n}} {\th}^{n \rho} {\th}^{\sigma \beta} (\mathbb{X}_{{\th}})_{{\rho \sigma}} \partial_\beta X^\mu \frac{\delta}{\delta \partial_n X^\mu} + (\mathbb{X}_{{\th}})_{{\alpha \beta}} \frac{\delta}{\delta {\th}_{\alpha \beta}}.
\end{equation}
\end{proof}

\begin{remark}
This first part of the proof shows that $\check{F}_P$ --- the space of pre-boundary fields, composed of field restrictions to the boundary and normal jets --- is a pre-symplectic manifold\footnote{We require the kernel of a presymplectic form to be a subbundle.}, with pre-symplectic form given by $\check{\omega}=\delta\check{\alpha}$, with $\check{\alpha}$ obtained by the integration by parts of the variation of the Polyakov action: $\delta S^{\text{cl}}_P= \mathrm{EL} + \check{\pi}^*\check{\alpha}$ with $\check{\pi}\colon F_P \to \check{F}_P$ the restriction map. The two form $\check{\omega}$ is obviously closed, and the first part of the proof serves to show that its kernel defines a regular foliation, i.e.\ a subbundle of $T\check{F}_P$. The next step is then to perform presymplectic reduction and find an explicit coordinate chart for the quotient $F_P^\partial\coloneqq\check{F}_P/\mathrm{ker}(\check{\omega}^\sharp)$.
\end{remark}

\begin{proof}[Proof. Part 2]
We now proceed to perform a reduction over the space of fields that eliminates the degrees of freedom related to the kernel. This produces a map ${\pi}^\partial: \check{{F}}_P \rightarrow {F}_P^\partial$ onto a symplectic manifold $(F^\partial_P,\omega_P^\partial)$ with $ \pi^{\partial *}\omega^\partial = \check{\omega}$. 

In order to find a chart for the reduced space $F^\partial_P$, we consider a transformation of the fields in $\check{F}_P$ along the flow of the vertical vector fields (i.e.\ the vector fields in the kernel of $\check{\omega}$). 
To see this, let us begin by observing that a basis of vector fields in the kernel $\mathrm{ker}(\check{\omega}^\sharp)$ is given by the vector $\mathbb{X}$ of Equation \eqref{A14}:

To compute the flow associated to this vector field we have to solve the set of differential equations:
\begin{subequations}\begin{align}
\partial_\tau{X^\mu} &= (\mathbb{X}_X)^\mu  =0\\
\partial_\tau{{\th}}_{\alpha \beta} &= (\mathbb{X}_{{\th}})_{{\alpha \beta}}\\
\partial_\tau{\partial_n X^\mu} &= \frac{1}{{\th}^{n n}} {\th}^{n \rho} {\th}^{\sigma \beta} (\mathbb{X}_{{\th}})_{{\rho \sigma}} \partial_\beta X^\mu =\frac{1}{{\th}^{n n}} {\th}^{n \rho} {\th}^{\sigma \beta}\partial_\tau{{\th}}_{\rho \sigma} \partial_\beta X^\mu\label{31/05A4}
\end{align}\end{subequations}
where $\tau$ is the parameter of the flow, while $n$ and $t$ are the indexes respectively of the normal and tangent direction to the boundary $\partial M$. From the first equation it follows immediately that $X^\mu$ doesn't change, and from the second it follows that we have almost complete control on the flow of ${\th}_{\alpha \beta}$. To find the transformation of $\partial_n X^\mu$, we can manipulate the third equation to obtain:
\begin{gather}
\partial_\tau{\partial_n X^\mu}  =\frac{1}{{\th}^{n n}} {\th}^{n \rho} {\th}^{\sigma \beta} \partial_\tau{{\th}}_{\rho \sigma} \partial_\beta X^\mu= - \frac{1}{{\th}^{n n}} \partial_\tau{{\th}}^{n \rho} {\th}^{\sigma \beta}{\th}_{\rho \sigma} \partial_\beta  X^\mu=-\frac{1}{{\th}^{n n}} \partial_\tau{{\th}}^{n \beta} \partial_\beta X^\mu \notag \notag\\
\Rightarrow {\th}^{n n } \partial_\tau{\partial_n X^\mu} = - \partial_\tau{{\th}}^{n n} \partial_n X^\mu - \partial_\tau{{\th}}^{n t} \partial_t X^\mu \notag \notag \\
\Rightarrow \partial_\tau ( {\th}^{n n} \partial_n X^\mu ) = - \partial_\tau ( {\th}^{n t} \partial_t X^\mu ) \notag \notag \\
{\th}^{n n} (\tau) \partial_n X^\mu (\tau) = {\th}^{n n} (\tau_0) \partial_n X^\mu (\tau_0) - {\th}^{n t}(\tau) \partial_t X^\mu (\tau) + {\th}^{n t} (\tau_0) \partial_t X^\mu (\tau_0)
\end{gather} 
where we used that $\partial_\tau{X}^{\mu}=0 \implies \partial_\tau{(\partial_t X^\mu)} =0$. The solution (reported in the last line) is independent of the particular choice of $\partial_\tau {{\th}_{\alpha \beta}}$, but depends only on the final value of ${\th}_{\alpha \beta}$. We use the freedom of choice of $(\mathbb{X}_{{\th}})_{{\alpha \beta}}$ to set ${\th}_{\alpha \beta}$ to a reference pseudo-Riemannian metric that we will choose to be Minkowski metric  up to a an overall sign. This follows from the fact that, along the flow of $\mathbb{X}$, the value of $\th^{nn}(\tau)$ must not vanish,\footnote{When it vanishes we have a lightlike boundary point, as detailed in Appendix \ref{AppendixLightlike}.} and its sign is therefore constant. We will then set the final value of the metric $\th_{\alpha\beta}(\tau)$ to be $\th_{\alpha\beta}=-\chi \eta_{\alpha\beta}$, where $\chi=\text{sign}(\th^{nn})$. Observe that this prescription covers both scenarios: when $\partial M$ space-like\footnote{Recall that $\th_{tt}=\varsigma \th^{nn}$, so that $\th_{tt}=1$ corresponds to the case $\partial M$ spacelike.} we have $\chi=-1$ and $-\chi\eta = \eta = \mathrm{diag}(-1,1)$, and the opposite when $\partial M$ is time-like. Finally, this procedure works also when $h$ is Riemannian, by choosing instead $\eta_{\alpha\beta}\leadsto\delta_{\alpha\beta}=\mathrm{diag}(1,1)$.

To achieve this, we can choose a $(\mathbb{X}_{{\th}})_{{\alpha \beta}}$ constant in $\tau$:
\begin{gather}\label{20/04A1}
(\mathbb{X}_{{\th}})_{{\alpha \beta}} (\tau) = ( \chi \eta_{\alpha\beta} - {\th}_{\alpha \beta} (0) )
\end{gather}
In this way, ${\th}_{\alpha \beta} (\tau) = {\th}_{\alpha \beta} (0) + \tau ( \chi \eta_{\alpha\beta} - {\th}_{\alpha \beta} (0) )$. Setting the flow to end at $\tau=1$, we get
\begin{gather}\label{20/04A2}
{\th}_{\alpha \beta} (1) = {\th}_{\alpha \beta} (0) + ( \chi \eta_{\alpha\beta} - {\th}_{\alpha \beta} (0) ) =  \chi \eta_{\alpha\beta}.
\end{gather}
We let ${\th}_{n t} \in {\Bbb R}$, ${\th}_{t t}$  vary freely in $ \Bbb R^{+}$ (resp. $\mathbb{R}^-$) and ${\th}_{n n}$ be the function ${\th}_{n n} = \frac{({\th}_{n t})^2 -1}{{\th}_{t t}}$, so that the equations \eqref{20/04A1} and \eqref{20/04A2} are meant only for the indices $(\alpha \beta) = (t t)$ and $(\alpha \beta) = (n t)$.  This choice is due to the fact that we are already considering the case ${\th}_{t t} = \varsigma {\th}^{n n} \neq 0$. We obtain:
\begin{gather}
X^\mu (1) = X^\mu (0) \notag\\
{\th}_{\alpha \beta}(1) = \chi \eta_{\alpha\beta}\notag\\
J^\mu:=\partial_n X^\mu (1) =  \chi ( {\th}^{n n} (0) \partial_n X^\mu (0) + {\th}^{n t} (0) \partial_t X^\mu (0)) \notag \\
\Rightarrow \ {} J^\mu = \chi  {\th}^{n \alpha} \partial_\alpha X^\mu  \label{A15}
\end{gather} 
where we used that $\th^{nn}(1) = \chi$.

We use the transformations found above to find a local chart in $F^\partial_P$, so that the map $\pi^\partial\colon \check{F}_P\to F^\partial_P$ reads: 
\begin{equation}
{\pi}^\partial: ({\th}_{\alpha \beta}, X^\mu, \partial_n X^\mu) \rightarrow (  X^\mu, J^\mu:= \chi {\th}^{n \alpha} \partial_\alpha X^\mu)
\end{equation}
and since the pre-boundary one-form is:
\begin{equation}
\check{\alpha} = \int_{\partial M} \delta X^\mu {\th}^{n  \alpha} \partial_{\alpha} X^\nu G_{\mu \nu} (X)
\end{equation}
we can easily gather that $\check{\alpha}$ is basic: $\check{\alpha} = \pi^{\partial*} \alpha^\partial$.  We take as elementary the field $J_\mu := J^\nu G_{\nu \mu}={\th}^{n  \alpha} \partial_{\alpha} X^\nu G_{ \nu \mu} (X)$, such that
\begin{equation}
\begin{aligned}
\alpha^\partial  = \intl_{\partial M} J_\mu \delta X^\mu \qquad \omega_P^\partial = \intl_{\partial M} \delta J_\mu\delta X^\mu,
\end{aligned}
\end{equation}
and the variation of the action becomes:
\begin{gather*}
\delta S_P = E L + {\pi}^* \alpha^\partial 
\end{gather*}
where $\pi=\pi^\partial\circ\check{\pi}$. Observe that the reduction maps for the two scenarios $\partial M$ space/time-like are mapped into each other by $J_\mu\to -J_\mu$.
\end{proof}

\begin{remark}
In the second part of the proof we have found an explicit global Darboux coordinate chart for $F^\partial_P$, thus characterising the space of boundary fields as the space of maps from the worldsheet into the cotangent bundle of the target: $F^\partial_P = C^\infty(\partial M,T^*N)$. Observe that the definition of the symplectic manifold $(F^\partial_P,\omega_P^\partial)$ is independent of coordinate choices. The apparent dependency on a choice of adapted coordinates arises when one looks for an explicit chart description for the map $\pi\colon F_P \to F^\partial_P$.

In the following part of the proof we will show how the constraints, defined as those equations of motion that can be solved algebraically in terms of boundary fields, define a coisotropic submanifold in $F^\partial_P$.
In fact, while one way to describe the reduced phase space of a system is to perform a Dirac analysis of the constraints in the bulk, one can express the constraints as the zero locus of functions over the space of boundary fields $F^\partial_P$ (the geometric phase space). To do this, we need the constraint functions induced from the bulk to be basic w.r.t. $\pi^\partial\colon \check{F}_P\to F^\partial_P$. 
\end{remark}

\begin{proof}[Proof. Part 3]
The Euler-Lagrange equations that result from imposing $EL=0$ (see \eqref{A4}), are:
\begin{subequations}
\begin{align}
\partial_\alpha ( {\th}^{\alpha \beta} \partial_\beta X^\nu G_{\mu \nu} (x) ) - \frac{1}{2} {\th}^{\alpha \beta} \partial_\alpha X^\rho \partial_\beta X^\sigma \frac{\partial G_{\rho \sigma}}{\partial X^\mu} =0\label{A17-1}\\
{ f_{\alpha \beta} := [{\th}_{\alpha \beta} \frac{{\th}^{\lambda \tau}}{2} \partial_\lambda X^\mu \partial_\tau X^\nu - \partial_\alpha X^\mu \partial_\beta X^\nu]G_{\mu \nu} =0}\label{A17-2}
\end{align}
\end{subequations}

While \eqref{A17-1} is an evolution equation --- a differential equation in $F_P$ --- \eqref{A17-2} is a constraint. 
The functions defining constraints on the space of bulk fields ${F}_P$ restrict to functions on the space of pre-boundary fields $\check{F}_P$. We claim that the functions on $\check{F}_P$ that define the constraints are basic with respect to ${\pi}^\partial: \check{{F}}_P \rightarrow {F}^\partial_P$

The $f_{\alpha \beta} ({\th}_{\alpha \beta}, X^\mu, \partial_n X^\mu)$'s are not manifestly functions of the reduced variables $\{X^\mu,J_\mu\}$. However, we can look for a combination
\begin{equation}
\tau^{\alpha \beta} ({\th}) f_{\alpha \beta} ({\th}, X^\mu, J_\mu) =: g(X^\mu, J_\mu) \label{21/04A1}
\end{equation}
and require that it be a function only of the reduced variables  and that they span the same vanishing locus (i.e.\ the set of points where the constraint functions vanish). The correct choice for such $\tau^{\alpha\beta}$ will be given in Appendix \ref{A:contraintmanipulation}. We expand:
\begin{equation}\label{A18}
{\tau^{\alpha \beta} ({\th}) f_{\alpha \beta} ({\th}, X^\mu, J_\mu) = l^{n n} J_\mu J^\mu + 2 l^{n t} \partial_t X^\mu J_\mu + l^{t t} \partial_t X^\mu \partial_t X^\nu G_{\mu \nu}}
\end{equation}
where the $l^{\alpha \beta}$'s do not depend\footnote{Notice that we assumed $\tau$ to be a function of ${\th}$ alone. More generally we could allow a dependency also on the fields $X^\mu$ and $J_\mu$, but it turns out that there is no need for that.} on ${\th}_{\alpha \beta}$. 

As defined in \eqref{A17-2}, the $f_{\alpha \beta}$'s are functions of $ ({\th}, X^\mu, \partial_n X^\mu)$, while in \eqref{21/04A1} we consider them as functions of $ ({\th}, X^\mu, J_\mu  [{\th}, X^\mu, \partial_n X^\mu]) $. We do this in order to remove the dependency on ${\th}$ through appropriate choices of $\tau_{\alpha \beta}$, and be left with functions defined on the space of boundary fields ${F}^\partial_P$. We can do it because the map $({\th}_{\alpha \beta}, X^\mu, J_\mu)  \rightarrow ({\th}_{\alpha \beta}, X^\mu, \partial_n X^\mu)$ is a diffeomorphism in $\check{F}_P$ that preserves $\mathrm{ker}(\check{\omega})$.

With the appropriate algebraic manipulations (reported in Appendix \ref{A:contraintmanipulation}), the constraints $f_{\alpha \beta}=0$ are equivalent to the following conditions, which only depend on the reduced variables:
\begin{subequations}\label{e:newPolConstraints}\begin{align}
{{(\partial_t X^\mu \partial_t X^\nu - \varsigma J^\mu J^\nu)G_{\mu \nu}=0}}\\
{\partial_t X^\mu J_\mu=0}
\end{align}\end{subequations}
where $J^\mu := J_\nu G^{\mu \nu}$. Notice that we are left with only two constraint functions even if we started with three. This was expected, since $f_{\alpha \beta}$ has only two degrees of freedom, and the combination ${\th}^{\alpha \beta} f_{\alpha \beta}$ vanishes. The calculations in this section of the prove are independent of the sign of $\th_{tt}$, and the sign of $J_\mu$ does not alter the constraints.

\end{proof}

\begin{remark}
Now that we have found a generating set of constraint functions on $F^\partial_P$ we can compute their Poisson brackets. Recall that, on a symplectic manifold, the Poisson bracket of two generic functions $f$ and $g$ is defined as: $\{ f, g\} = \frac{1}{2} \iota_{\mathbb{F}} \iota_{\mathbb{G}} \omega_P^\partial$, where $\mathbb F$  is the Hamiltonian field of $f$ defined by the equation $\iota_{\mathbb{F}} \omega_P^\partial = \delta f$, and the same holds for $\mathbb G$ and $g$. We consider now the Lorentzian case $\varsigma=-1$. The other case is analogous.
\end{remark}

\begin{proof}[Proof. Part 4.]
The Hamiltonian fields of the constraint functions are defined by the equation: $\iota_{\mathbb{H}_\xi} \omega_P^\partial = \delta H_{\xi}$ and $\iota_{\mathbb{L}_\psi}  \omega_P^\partial = \delta L_{\psi}$, hence:
\begin{gather*}
{\delta H_{\xi} =   \intl_{\partial M}\bigg( -  2 \partial_t( \xi (\partial_t X^\mu)G_{\mu \nu}) +  \xi ( \partial_t X^k \partial_t X^\nu + J^k J^\nu) \frac{\partial G_{ k \nu}}{\partial X^\mu})\bigg) \delta X^\mu +  2 \bigg( \xi J^\mu \bigg) \delta J_\mu}\\
\delta L_{\psi} = 2 \intl_{\partial M}\bigg( -  \partial_t( \psi J_\mu) \bigg) \delta X^\mu + \bigg(  \psi \partial_t X^\mu \bigg) \delta J_\mu
\end{gather*}
leads to
\begin{equation}
\begin{aligned}
\iota_{\mathbb{H}_\xi} \omega_P^\partial = \intl_{\partial M} (\mathbb{H}_{\xi,X})^{\mu} \delta J_{\mu} - (\mathbb{H}_{\xi,J})_{\mu} \delta X^\mu
\left(= 
\intl_{\partial M} \frac{ \partial H_{\xi}}{\partial X^\mu} \delta X^\mu + \frac{\partial H_{\xi}}{\partial J_\mu} \delta J_\mu\right)
\end{aligned}
\end{equation}
\begin{equation}\label{21/04A4}
\begin{aligned}
\iota_{\mathbb{L}_\psi}  \omega_P^\partial = \intl_{\partial M} (\mathbb{L}_{\psi,X})^{\mu} \delta J_{\mu} - (\mathbb{L}_{\psi,J})_{\mu} \delta X^\mu \left(= \intl_{\partial M} \frac{ \partial L_{\psi}}{\partial X^\mu} \delta X^\mu + \frac{\partial L_{\psi}}{\partial J_\mu} \delta J_\mu\right).
\end{aligned}
\end{equation} 
Thus:
\begin{equation*}\label{A27}
\begin{aligned}
&{\mathbb{H}_{\xi} :=  \intl_{\partial M}  2 \bigg[  \xi J^\mu \bigg] \frac{\delta}{\delta X^\mu} +   \bigg[ 2 \partial_t ( \xi (\partial_t X^\nu)G_{\mu \nu} ) - \xi (\partial_t X^k \partial_t X^\nu  + J^k J^\nu) \frac{\partial G_{ k \nu}}{\partial X^\mu} \bigg] \frac{\delta}{\delta J_\mu}}\\
&\mathbb{L}_{\psi} : = \intl_{\partial M} 2 \bigg[ \psi \partial_t X^{\mu} \bigg] \frac{\delta}{\delta X^\mu} + 2 \bigg[  \partial_t (\psi J_\mu ) \bigg] \frac{\delta}{\delta J_\mu}
\end{aligned}
\end{equation*}
We can calculate now the Poisson brackets of the constraints: 
\begin{gather*}\label{A28}
\{ H_{\xi}, L_{\psi} \} = \frac{1}{2}\omega_P^\partial (\mathbb{H}_{\xi}, \mathbb{L}_{\psi})=
\\
\frac{1}{2} \intl_{\partial M}  -(\mathbb{H}_{\xi,X})^{\mu} (\mathbb{L}_{\psi,J})_{\mu} + (\mathbb{H}_{\xi,J})_{\mu} (\mathbb{L}_{\psi,X})^{\mu}=\\
{ \intl_{\partial M} - 2 \xi J^\mu \partial_t ( \psi J_\mu) +  2 \psi \partial_t X^\mu \partial_t (\xi (\partial_t X^\nu)G_{\mu \nu}) -  \psi \partial_t X^\mu  \xi ( \partial_t X^k \partial_t X^\nu + J^k J^\nu) \frac{\partial G_{ k \nu}}{\partial X^\mu}.}
\end{gather*}

Using the Leibniz rule and integrating by parts the third line, we obtain:
\begin{gather*}
{\intl_{\partial M} \bigg(\partial_t (\xi \psi )  - 2 \xi ( \partial_t \psi ) \bigg) J_\mu J^\mu +
\bigg(- \partial_t (\psi \xi ) + 2 \psi (\partial_t \xi)\bigg) \partial_t X^\mu \partial_t X^\nu G_{\mu \nu}} \\ 
{=\intl_{\partial M}  \bigg((\partial_t \xi ) \psi - \xi ( \partial_t \psi ) \bigg) ( J_\mu J^\mu + \partial_t X^\mu \partial_t X^\nu G_{\mu \nu})} 
\end{gather*}

Thus, defining $[\xi, \psi]:= (\partial_t \xi) \psi - \xi (\partial_t \psi)$, we conclude that $\{H_\xi, L_{\psi}\} = H_{[\xi, \psi]}$. The other two Poisson brackets yield:
{\begin{gather*}
 \{ H_{\xi}, H_{\xi '} \} =  \frac{1}{2} \omega_P^\partial(\mathbb{H}_{\xi}, \mathbb{H}_{\xi '})=
\intl_{\partial M} -  2\xi J^\mu \partial_t ( \xi ' (\partial_t X^\nu)G_{\mu \nu}) + 2  \xi '  J^\mu \partial_t (\xi (\partial_t X^\nu)G_{\mu \nu})+\notag\\
\intl_{\partial M}  \xi J^\mu  \xi '  (\partial_t X^k \partial_t X^\nu   + J^k J^\nu) \frac{\partial G_{ k \nu}}{\partial X^\mu} -  \xi ' J^\mu  \xi (\partial_t X^k \partial_t X^\nu  + J^k J^\nu) \frac{\partial G_{ k \nu}}{\partial X^\mu} =\notag\notag\\
\intl_{\partial M} -   \xi J_\mu (\partial_t   \xi ' ) \partial_t X^\mu +   \xi '  J_\mu (\partial_t \xi )\partial_t X^\mu= 
\intl_{\partial M} [ (\partial_t \xi )  \xi '  - \xi (\partial_t   \xi ' ) ] J_\mu \partial_t X^\mu \notag
\end{gather*}}
Thus: $\{H_\xi, H_{\xi '}\} = L_{[\xi, \xi ']}$. And:

\begin{gather*}
 \{ L_{\psi}, L_{\psi '} \} =  \frac{1}{2} \omega^\partial(\mathbb{L}_{\psi}, \mathbb{L}_{\psi '})=
\intl_{\partial M} - 2 \psi \partial_t X^\mu \partial_t ( \psi ' J_\mu) + 2 \psi ' \partial_t X^\mu \partial_t (\psi  J_\mu) =\notag\\
\intl_{\partial M} - 2 \psi \partial_t X^\mu (\partial_t \psi ')  J_\mu + 2 \psi '  \partial_t X^\mu (\partial_t \psi ) J_\mu= \intl_{\partial M}  [(\partial_t \psi ) \psi ' - \psi (\partial_t  \psi ') ] 2 J_\mu \partial_t X^\mu \notag
\end{gather*}
which means that: $\{L_{\psi}, L_{\psi '}\} = L_{[\psi, \psi ']}$.
\end{proof}
\begin{remark}
We have shown that the constraints that we have derived in Part 3
\begin{equation}
\begin{aligned}
&{H_{\xi} :=  \intl_{\partial M} \xi  ( (\partial_t X^\nu)G_{\mu \nu} \partial_t X^\mu - \varsigma J_\mu J^\mu )} \\
&L_{\psi} : = 2 \intl_{\partial M} \psi  ( \partial_t X^\mu J_\mu ),
\end{aligned}
\end{equation}
where $\xi$ and $\psi $ are smooth functions on the boundary, $\xi, \psi \in C^\infty (\partial M)$, are closed under the Poisson brackets. Another way of phrasing this, is that they are first class constraints. 
\end{remark}

\section{Nambu--Goto Theory --- Reduced Phase Space}
In this section we will analyze the boundary structure of Nambu--Goto theory, and describe its reduced phase space. Once again, we consider the field $X\in C^\infty(M,N)=:F_{NG}$. As seen in Section \ref{sec:Strings}, the Nambu--Goto action is the surface pseudo-area of the string:
\begin{equation}\label{A29}
S^{\text{cl}}_{NG} := \int_{M} \sqrt{\mathsf{g}}\, d^2 {x},
\end{equation}
where $\mathsf{g}:= |\mbox{det}(g_{\alpha \beta})|$ and $g_{\alpha \beta}:= \partial_\alpha X^\mu \partial_\beta X^\nu G_{\mu \nu} (X) $. In this case $g_{\alpha \beta}$ is not an elementary field but just a function of $X^\mu$, the map from the worldsheet to the target space that defines the string. 

\begin{theorem}\label{thm:NGRPS}
The geometric phase space of Nambu--Goto theory
$$
F^\partial_{NG} \coloneqq \check{F}_{NG}/\mathrm{ker}(\check{\omega}_{NG})
$$  
coincides with its reduced phase space $\Phi^{\text{red}}_{NG}$. Moreover, denoting by $C_P\subset F^\partial_P$ the submanifold of constraints of Polyakov theory, with $\iota_C\colon C_P \to F^\partial_P$ the inclusion map, there exist maps 
$$
{\pi}^\partial_{\mathrm{partial}}\colon \check{F}_{NG}\to C_P; \qquad \varphi\colon\Phi^{\text{red}}_{NG}\to \Phi^{\text{red}}_{P}
$$ 
such that the following diagram commutes:
\begin{equation}\label{e:PtoNG}
    \xymatrix{
    F_{NG} \ar[d]^{\check{\pi}_{NG}} \ar[rr]^{\phi_{NG}} & & F_P\ar[d]_{\check{\pi}_{P}}\\
    \check{F}_{NG} \ar[dd]^{{\pi}^\partial_{NG}} \ar[dr]^{{\pi}^\partial_{\mathrm{partial}}} & & \check{F}_{P} \ar[d]^{{\pi}^\partial_P} \\
    & C_{P} \ar[d]^{\pi^{\text{red}}_P} \ar[r]^{\iota_C} & F^\partial_P\\
    F^\partial_{NG}\simeq \Phi^{\text{red}}_{NG} \ar[r]^-{\varphi} & \Phi^{\text{red}}_{P} &
    }
\end{equation}
\end{theorem}

\subsection{Proof of Theorem \ref{thm:NGRPS}}
We will divide the proof into two parts. First, we will show that the kernel of the boundary pre-symplectic form $\check{\omega} = - \delta\mathrm{EL}$ is regular. This in principle allows pre-symplectic reduction to the geometric phase space, which however might be singular. In the second part of the proof, instead of reducing with respect to the entirety of the kernel, we will perform partial reduction with respect to only a subset of $\mathrm{ker}(\check{\omega})$. The result of this partial reduction will turn out to be (diffeomorphic to) the constraint space for Polyakov theory.

\begin{proof}[Proof. Part 1]
Since $M$ has a boundary, when varying the action we obtain:
\begin{align*}\label{A30}
\delta S_{NG}^{\text{cl}} &= \int_{M} \frac{1}{2} \sqrt{\mathsf{g}}\, g^{\alpha \beta} \delta g_{\alpha \beta} d^2 {x}\\
&= \int_{M}  \sqrt{\mathsf{g}}\, g^{\alpha \beta}  \partial_\alpha X^\nu \delta \partial_\beta X^\mu G_{\mu\nu}d^2{x} + \frac{1}{2} \int_{M}  \sqrt{\mathsf{g}}\, g^{\alpha \beta}  \partial_\alpha X^\mu \partial_\beta X^\nu \frac{\partial G_{\mu \nu} (X)}{\partial X^\lambda} \delta X^\lambda d^2 {x}\\ 
& = - \int_{M}  \partial_\beta ( \sqrt{\mathsf{g}}\, g^{ \beta \alpha}  G_{\mu\nu}\partial_\alpha X^\nu) \delta  X^\mu d^2 {x} +  \frac{1}{2} \int_{M}  \sqrt{\mathsf{g}}\, g^{\alpha \beta}  \partial_\alpha X^\mu \partial_\beta X^\nu \frac{\partial G_{\mu \nu} (X)}{\partial X^\lambda} \delta X^\lambda d^2 {x} \\ 
& \phantom{=}  + \int_{M}  \partial_\beta ( \sqrt{\mathsf{g}}\, g^{ \beta \alpha}  G_{\mu\nu}\partial_\alpha X^\nu \delta X^\mu ) d^2 {x} \\
&=: EL + \check{\pi}_{NG}^*\check{\alpha},
\end{align*}
where, as before, $E L$ is the term that provides the equations of motion:
\begin{equation}\label{A31}
\begin{aligned}
EL &:= - \int_{M}  \partial_\beta ( \sqrt{\mathsf{g}}\, g^{ \beta \alpha} G_{\mu\nu} \partial_\alpha X^\nu) \delta  X^\mu d^2 {x} \\
& \phantom{:=} +  \frac{1}{2} \int_{M}  \sqrt{\mathsf{g}}\, g^{\alpha \beta}  \partial_\alpha X^\mu \partial_\beta X^\nu \frac{\partial G_{\mu \nu} (X)}{\partial X^\lambda} \delta X^\lambda d^2 {x}
\end{aligned}
\end{equation}
while $\check{\alpha}$ is the boundary term:
\begin{equation}\label{e:NGclassoneform}
\check{\alpha} =\int_{ \partial M} \sqrt{\mathsf{g}}\, g^{n \beta} G_{\mu\nu} \partial_\beta X^\nu \delta X^\mu d{x}^t,
\end{equation}
interpreted as a one-form on $\check{F}_{NG}=T(C^\infty(\partial M, N))$, the space of restrictions of fields and normal jets to the boundary, and $\check{\pi}_{NG}\colon F_{NG} \to \check{F}_{NG}$ is the natural surjective submersion onto it. The associated two-form $\check{\omega} := \delta\check{\alpha}$ is:

\begin{equation}\label{A3228/05}
\begin{aligned}
\check{\omega} =
 \int_{\partial M} d {x}^t \bigg[\sqrt{\mathsf{g}}\, g^{\lambda \rho} \partial_\lambda X_\nu &(\delta \partial_\rho X^\nu) g^{n \alpha} \partial_\alpha X_\mu +\\
- \sqrt{\mathsf{g}}\, (g^{\lambda n} g^{\alpha \rho} + g^{\rho n} g^{\alpha \lambda} ) &\partial_\rho X_\nu (\delta \partial_\lambda X^\nu) \partial_\alpha X_\mu + \sqrt{\mathsf{g}}\, g^{n \alpha} (\delta \partial_\alpha X_\mu) \bigg] \delta X^\mu + \\
 \int_{\partial M} d {x}^t \bigg[\frac{\sqrt{\mathsf{g}}\,}{2} g^{\lambda \rho} \partial_\lambda X^{\nu '} & \partial_\rho X^\nu \delta G_{\nu' \nu} g^{n \alpha} \partial_\alpha X_\mu +\\
- \frac{\sqrt{\mathsf{g}}\,}{2} (g^{\lambda n} g^{\alpha \rho} + g^{\rho n} g^{\alpha \lambda} ) &\partial_\rho X^{\nu '}  \partial_\lambda X^\nu\delta G_{\nu' \nu} \partial_\alpha X_\mu + \sqrt{\mathsf{g}}\, g^{n \alpha}  \partial_\alpha X^{\mu '} \delta G_{\mu' \mu} \bigg] \delta X^\mu
\end{aligned}
\end{equation}
where we used: $\delta \sqrt{\mathsf{g}}\,= \frac{1}{2} \sqrt{\mathsf{g}}\, g^{\alpha \beta} \delta g_{\alpha \beta} = \sqrt{\mathsf{g}}\, g^{\alpha \beta} \partial_{\alpha} X_\nu (\delta \partial_\beta X^\nu)$ and $\delta g^{\alpha \beta}= - g^{\lambda \beta} g^{\alpha \rho} \delta g_{\rho \lambda}$. The terms in the second integral contain variations of the metric as a function of $X^\mu$ and are relevant when dealing with a non-constant metric. Thus, rearranging the terms:
{\small\begin{equation}\label{A33}
\begin{aligned}
\check{\omega} = \int_{\partial M} d {x}^t \sqrt{\mathsf{g}}\, \bigg[ (g^{\lambda \rho} g^{n \alpha} - g^{\lambda n} g^{\alpha \rho} - g^{\rho n} &g^{\alpha \lambda} ) \partial_\lambda X_\nu \partial_\alpha X_\mu + g^{n \rho} G_{\mu \nu} \bigg] \delta \partial_\rho X^\nu \delta X^\mu\\
+ \int_{\partial M} d {x}^t \sqrt{\mathsf{g}}\, \bigg[ \frac{1}{2}(g^{\lambda \rho} \partial^n X^{\mu '} - g^{\lambda n} \partial^\rho X&^{\mu '} - g^{\rho n} \partial^\lambda X^{\mu '} ) \partial_\lambda X^{\nu '} \partial_\rho X^\nu \frac{\partial G_{\nu ' \nu}}{\partial X^\sigma} G_{\mu ' \mu } + \partial^n X^{\mu '} \frac{\partial G_{\mu ' \mu}}{\partial X^\sigma} \bigg] \delta  X^\sigma \delta X^\mu
\end{aligned}
\end{equation}}

Let us consider a general vector field $\mathbb{X}\in C^\infty(\check{F}_{NG},T\check{F}_{NG})$ in our space:
\begin{equation}\label{A34}
\mathbb X := (\mathbb{X}_X)^\mu  \frac{\delta}{\delta X^\mu} + (\mathbb{X}_{\partial_n X})^\mu  \frac{\delta}{\delta \partial_n X^\mu}.
\end{equation}
Let us enforce the condition $\iota_{\mathbb{X}}\check{\omega}=0$. First, we have:
\begin{gather}
\bigg[  (\cancel{g^{\lambda n} g^{n \alpha}} - \cancel{g^{\lambda n} g^{\alpha n}} - g^{n n} g^{\alpha \lambda} ) \partial_\lambda X_\nu \partial_\alpha X_\mu + g^{n n} G_{\mu \nu}  \bigg] (\mathbb{X}_X)^\nu  (\delta \partial_n X^\mu) = 0\notag\\
\Leftrightarrow\qquad  g^{n n} (G_{\mu \nu} - \partial_{\alpha} X_\mu \partial^\alpha X_\nu ) (\mathbb{X}_X)^\nu  (\delta \partial_n X^\mu)=0\notag\\
\Leftrightarrow\qquad  g^{n n} P^\perp_{\mu \nu} (\mathbb{X}_X)^\nu  = 0\notag\\
\Leftrightarrow\qquad  P^\perp_{\mu \nu} (\mathbb{X}_X)^\nu  = 0 \notag\\
\Leftrightarrow\qquad  (\mathbb{X}_X)^\mu  = \alpha_n \partial_n X^\mu + \alpha_t \partial_t X^\mu\label{A35}
\end{gather}
where $P^\perp_{\mu \nu} := G_{\mu \nu} - \partial_{\alpha} X_\mu \partial^\alpha X_\nu $ is the projector to the subspace orthogonal to the tangent space (the space spanned by $\partial_t X_\mu$ and $\partial_n X_\mu$), and we used that $g^{n n} \neq 0$. Moreover:
\begin{equation}\label{A36}
\begin{aligned}
\delta X^\mu\colon \qquad  \bigg\{ \sqrt{\mathsf{g}}\, \bigg[ (g^{\lambda t} g^{n \alpha} - g^{\lambda n} g^{\alpha t} - g^{t n} g^{\alpha \lambda} ) \partial_\lambda X_\nu \partial_\alpha X_\mu + g^{n t} G_{\mu \nu} \bigg] \partial_t (\mathbb{X}_X)^\nu  + \\
\sqrt{\mathsf{g}}\, \bigg[  - g^{n n} g^{\alpha \lambda} \partial_\lambda X_\nu \partial_\alpha X_\mu + g^{n n} G_{\mu \nu} \bigg]  (\mathbb{X}_{\partial_n X})^\nu  + \\
\partial_t \bigg( \sqrt{\mathsf{g}}\, \bigg[ (g^{\lambda t} g^{n \alpha} - g^{\lambda n} g^{\alpha t} - g^{t n} g^{\alpha \lambda} ) \partial_\lambda X_\mu \partial_\alpha X_\nu + g^{n t} G_{\nu \mu} \bigg] (\mathbb{X}_X)^\nu  \bigg) \bigg\} =\\
\delta X^\mu \sqrt{\mathsf{g}}\,(C_{\mu \nu} - C_{\nu \mu}) (\mathbb{X}_X)^\nu 
\end{aligned}
\end{equation}
where $C_{\mu \nu}$ groups together the terms inside the square bracket in the second line of \eqref{A33} (recall that $\partial^\rho X^\mu := g^{\rho \lambda} \partial_\lambda X^\mu$):
\begin{equation*}
\begin{aligned}
 C_{\sigma \mu} =   \frac{1}{2}(g^{\lambda \rho} \partial^n X^{\mu '} - g^{\lambda n} \partial^\rho X^{\mu '} - g^{\rho n} \partial^\lambda X^{\mu '} ) \partial_\lambda X^{\nu '} \partial_\rho X^\nu \frac{\partial G_{\nu ' \nu}}{\partial X^\sigma} G_{\mu ' \mu } + \partial^n X^{\mu '} \frac{\partial G_{\mu ' \mu}}{\partial X^\sigma} 
\end{aligned}
\end{equation*}
It is possible to solve equation \eqref{A36}, and, at the end of a good deal of algebraic manipulations, reported in Appendix $\ref{AppendixLengthy}$, we obtain that
\begin{multline}\label{A48two}
 (\mathbb{X}_{\partial_n X})^\mu  = \beta_n \partial_n X^\mu + \beta_t \partial_t X^\mu +(g^{n n})^{-1} P^{\perp, \mu}_{\nu} \bigg[   (g^{t t} \alpha_n - 2 g^{n t} \alpha_t ) \partial_t \partial_t X^\nu - g^{n n}\alpha_t \partial_t \partial_n X^\nu \bigg]  \\
+ (g^{n n})^{-1}P^{\perp, \mu}_{\nu} G^{\nu w} \dot{G}_{w k} \bigg[ (g^{t t} \alpha_n - g^{n t} \alpha_t ) \partial_t X^k +  (g^{n t} \alpha_n - g^{n n}\alpha_t )\partial_n X^k \bigg]+\\
- (g^{n n})^{-1}P^{\perp, \mu}_{\nu} \bigg[ \alpha_n  \frac{1}{2}\partial_\rho X^{\nu '} \partial^\rho X^\nu \frac{\partial G_{\nu ' \nu}}{\partial X^\mu} -   \partial^n X^{\mu '}  \frac{\partial G_{\mu ' \mu}}{\partial X^\sigma} (\alpha_n \partial_n X^\sigma + \alpha_t \partial_t X^\sigma) \bigg].
\end{multline}
\end{proof}

\begin{remark}
We have shown that the kernel of $\check{\omega}$ is regular and has (local) dimension 4, with degrees of freedom $\{\alpha_n, \alpha_t, \beta_n, \beta_t\}$. This, in principle, allows us to perform a pre-symplectic reduction over the space of fields ($\check{\pi}: \check{F}_{NG} \rightarrow F^\partial_{NG}$), which will be discussed in  Part 2 of the Proof.
\end{remark}

\begin{proof}[Proof. Part 2]
First of all, let us observe that among the equations of motion for Nambu--Goto theory there are no constraints: all are evolution equations for the field $X$. This implies that any initial datum in the geometric phase space (i.e.\ the space of boundary fields) $F^\partial_{NG}=\check{F}_{NG}/\mathrm{ker}(\check{\omega})$ can be extended (formally) to a solution of the evolution equations in a neighborhood of the boundary. As a consequence, the reduced phase space coincides with the geometric phase space: $ \Phi^{\text{red}}_{NG} = F^\partial_{NG} $.

However, the pre-boundary two-form $\check{\omega}$ has a nontrivial kernel, and to obtain $F^\partial_{NG}$, one needs to perform pre-symplectic reduction. Hence, we must solve the following system of differential equations for four linearly independent choices of $\{ \alpha_n, \alpha_t, \beta_n, \beta_t \}$, which corresponds to flowing along a basis of the kernel $\mathrm{ker}(\check{\omega})$:
\begin{gather}
\partial_\tau X^\mu = (\mathbb{X}_X)^\mu  \label{A49one} \\
\partial_\tau \partial_n X^\mu = (\mathbb{X}_{\partial_n X})^\mu  \label{A49two}
\end{gather}
where $\tau$ is once again the parameter of the flow. 

We consider a two-step reduction. This corresponds to first reducing with respect to the subspace of $\mathrm{ker}(\check{\omega}_{NG})$ given by $\alpha_n=\alpha_t=0$. We will then see how the residual reduction has been already taken care of in Theorem \ref{THM:PolyakovRPS}.

The first step of the reduction gives us a first order linear differential equation and a trivial one:
\begin{equation}\label{A50}
\begin{aligned}
\partial_\tau &X^\mu=0\\
\partial_\tau \partial_n X^\mu = \beta_n &\partial_n X^\mu + \beta_t \partial_t X^\mu 
\end{aligned}
\end{equation}
whose solutions are:
\begin{equation}\label{A51}
\begin{aligned}
X^\mu (\tau&) = X^\mu (\tau_0)\\
\partial_n X^\mu (\tau) = e^{- \int_{\tau_0}^\tau (- \beta_n (\tau') d \tau'} \bigg\{ \int_{\tau_0}^\tau e&^{ \int_{\tau_0}^{\tau'} (- \beta_n (\tau'') d \tau''} \beta_t (\tau') d \tau' \partial_t X^\mu (\tau_0) + \partial_n X^\mu (\tau_0) \bigg\}
\end{aligned}
\end{equation}
where we used that  $\partial_\tau X^\mu=0 \Rightarrow \partial_\tau \partial_t X^\mu=0$. If we take $\beta_n$ and $\beta_t$ to be constant:

\begin{equation}\label{A52}
\begin{aligned}
\partial_n X^\mu (\tau) = e^{\beta_n (\tau -\tau_0)} \bigg\{ \int_{\tau_0}^\tau e^{-\beta_n (\tau' -\tau_0)} \beta_t d \tau' \partial_t X^\mu (\tau_0) + \partial_n X^\mu (\tau_0) \bigg\}=\\
e^{\beta_n (\tau -\tau_0)} \bigg\{-\frac{1}{\beta_n}[ e^{-\beta_n (\tau -\tau_0)} - 1] \beta_t \partial_t X^\mu (\tau_0) + \partial_n X^\mu (\tau_0) \bigg\}
\end{aligned}
\end{equation}

Setting $\tau=1, \tau_0=0$:

\begin{equation}\label{A53}
\begin{aligned}
\partial_n X^\mu (1) =e^{\beta_n} &\bigg\{-\frac{1}{\beta_n}[ e^{-\beta_n} - 1] \beta_t \partial_t X^\mu (0) + \partial_n X^\mu (0) \bigg\} =\\
&\frac{ \beta_t}{\beta_n}[e^{\beta_n}-1] \partial_t X^\mu (0) + e^{\beta_n} \partial_n X^\mu (0)  
\end{aligned}
\end{equation}
Notice that $\lim_{\beta_n \rightarrow 0} \frac{ 1}{\beta_n}[e^{\beta_n}-1] = 1$ hence the  term $ \frac{ 1}{\beta_n}[ e^{\beta_n}-1]$ is always well defined and positive. And if we choose $\beta_n=0, \partial_n X^\mu (1)$ is the same as the limit for $\beta_n \rightarrow 0$ of the expression in \eqref{A53}. For every choice of $\beta_n$ the term $\frac{ \beta_t}{\beta_n}[e^{\beta_n}-1] $ can take any possible real value through the right choice of $\beta_t$. We can rewrite a generic solution of \eqref{A53} as:
\begin{equation}\label{A54}
\begin{aligned}
\partial_n X^\mu (1) = A \partial_t X^\mu (0) + |B| \partial_n X^\mu (0),  
\end{aligned}
\end{equation}
where $|B|$ is positive due to the fact that $e^{\beta_n}$ can have only positive values. We will choose
\begin{gather}
\beta_n = \bigg(\text{log} ( \chi\sqrt{\mathsf{g}}\,  g^{n n} )\bigg)\bigg|_{\tau=0}\label{A54B1}\\
\beta_t = \bigg(\frac{ \text{log} (\chi\sqrt{\mathsf{g}}\,   g^{n n} )}{\chi\sqrt{\mathsf{g}}\,  g^{n n} - 1} \chi\sqrt{\mathsf{g}}\,  g^{n t}\bigg)\bigg|_{\tau=0}\label{A54B2}
\end{gather}
where $\chi = \text{sign} [ g^{n n}]$, and $\beta_n$ and $\beta_t$, as defined in $\eqref{A54B1}$ and $\eqref{A54B2}$, are well-defined and smooth for $g^{nn}\in\mathbb{R}^+\backslash\{0\}$, $g^{nt}\in\mathbb{R}$. 

The reason for this particular choice of parameters is that it induces a transformation on $\partial_n X^\mu$ analogous to the one used in Theorem \ref{THM:PolyakovRPS} for Polyakov theory. This will lead to a description of the partially reduced space of boundary Nambu--Goto fields that is manifestly related to the reduced phase space  of Polyakov theory, as we will see now.

The solution \eqref{A54} for this choice of $\beta_n,\beta_t$ is thus:
\begin{equation}\label{A55}
\begin{aligned}
\partial_n X^\mu (1) = \chi \sqrt{\mathsf{g}}\, g^{n \alpha} \partial_\alpha X^\mu (0),
\end{aligned}
\end{equation}
Thus, we have constructed a partial reduction map $\pi_{\mathrm{partial}}^\partial\colon \check{F}_{NG} \to F^\partial_{\mathrm{partial}}\subset T^*(C^\infty(\partial M,N))$
\begin{gather}\label{A55B1}
(X^\mu, \partial_n X^\mu) \longrightarrow (X^\mu, J_\mu :=\chi \sqrt{\mathsf{g}}\, g^{n \alpha} \partial_\alpha X^\nu G_{\mu\nu}(X))
\end{gather}
where $F^\partial_{\mathrm{partial}}$ is parametrised by elementary fields $(X^\mu,J_\mu)\in T^*(C^\infty(\partial M,N))$ which, by construction, must satisfy the constraints: 
\begin{subequations}\label{e:NGconstraints}\begin{gather}
J_\mu \partial_t X^\mu = \chi \sqrt{\mathsf{g}}\, g^{n\alpha} g_{\alpha t} = \chi \sqrt{\mathsf{g}}\, \delta^n_t = 0 \label{A55B3}\\
J_\mu J^\mu - \varsigma \partial_t X_\mu \partial_t X^\mu = \chi^2 | \mathrm{det}(g)| g^{nn} - \varsigma g_{tt} = 0 \label{A55B4}
\end{gather}\end{subequations}
since, again, $g_{tt} = \mathrm{det}(g) g^{nn}$.

Using $(X^\mu, J_\mu := J^\nu G_{\nu \mu})$ as elementary fields in $T^*(C^\infty(\partial M,N))$, we see that the constraints defining $F^\partial_{\mathrm{partial}}\subset T^*(C^\infty(\partial M,N))$ coincide with the constraints of Polyakov theory (Equation \eqref{eq:PolConstraints}). The partially-reduced boundary one- and two-forms read:
\begin{gather}\label{A55B2}
\check{\alpha} = \pi_{\mathrm{partial}}^{\partial *}\alpha^\partial_{\mathrm{partial}} = \pi_{\mathrm{partial}}^{\partial *}\int_{\partial M} J_\mu \delta X^\mu; \qquad \omega^\partial_{\mathrm{partial}} = \int_{\partial M} \delta J_\mu \delta X^\mu
\end{gather}
where $X^\mu,J_\mu)$ satisfy \eqref{e:NGconstraints}, so that $\omega^\partial_{\mathrm{partial}}$ coincides with the restriction of the boundary form $\omega^\partial_P$ of Polyakov theory to the zero locus $C_P$ of Polyakov's constraint functions \eqref{eq:PolConstraints}, i.e.
$$
  \omega^\partial_{\mathrm{partial}} = \omega^\partial_P\vert_{C_P}
$$

Naturally, then, the residual kernel of $\omega^\partial_{\mathrm{partial}}$ is generated by 
the characteristic distribution of the constraints \eqref{A55B3} and \eqref{A55B4}, and denoting by $\underline{F^\partial_{\mathrm{partial}}}$ the reduction of $F^\partial_{\mathrm{partial}}$ by $\mathrm{ker}(\omega^\partial_{\text{partial}})$, we have
$$
\Phi_{NG} \equiv \underline{F^\partial_{\mathrm{partial}}} \simeq \underline{C} \equiv \Phi_{P}.
$$ 
Hence we have that 
\begin{equation}
    \xymatrix{
    F_{NG} \ar[d]_{\check{\pi}_{NG}}&  \\
    \check{F}_{NG} \ar[dd]^{\pi_{NG}^\partial} \ar[dr]^{\pi_{\mathrm{partial}}^\partial} & \\
    & F^\partial_{\mathrm{partial}} \simeq C_P  \ar[d]^{\pi^{\text{red}}_P} \\
    F_{NG}^\partial \equiv\Phi^{\text{red}}_{NG} \ar[r]^-\simeq & \Phi_P^{\text{red}}
    }
\end{equation}
and we conclude that the diagram \eqref{e:PtoNG} commutes by adding the inclusion $\iota_C\colon C_P\to F^\partial_P$ and the classical equivalence $\phi_{NG}\colon F_{NG} \to F_{P}$ defined in Remark \ref{rem:classeq}.
\end{proof}

\begin{remark}
In the second part of the proof we have reduced the space of pre-boundary fields with respect to the kernel of $\check{\omega}$ in a two-step fashion. The intermediate partial reduction turned out to coincide with the (presymplectic) manifold given by the constraints for Polyakov theory. As a consequence, the reduced Phase space of Nambu--Goto theory coincides with the reduced phase space of Polyakov theory.
\end{remark}

\section{Polyakov Theory --- BV-BFV Analysis}\label{Sec:BVPOL}
We now perform an analysis of the boundary structure induced by the BV-extension of Polyakov theory formally  similar to the procedure outlined in Section \ref{Sec:RPS}. As in the degree-$0$ scenario, variations of $S_P$ induce a boundary term $\check\alpha$, a one-form on some appropriate graded space of pre-boundary fields, and as in Section \ref{Sec:RPSPOL}, we look for the kernel of $\check\omega=\delta \check\alpha$. For the case of Polyakov theory, we show that it is regular and perform the pre-symplectic reduction, which allows us to construct a chart of the symplectic space of BFV boundary fields, in terms of which we will write the BFV data for Polyakov theory. By doing this, we prove that the Polyakov model of bosonic strings admits a BV-BFV structure on worldsheets with non-null boundary $(M,\partial M)$. In \cite{ADER1987103} and \cite{Craps_2005} the BV treatment of this theory is inspected and the symmetries and BV transformations are outlined (the latter adds an extra term to the action).

\subsection{Symmetries of Polyakov theory}

The knowledge of the symmetries of the theory is fundamental to the construction of the BV structure. In this section we are going to write the infinitesimal symmetry transformations of the fields, and promote the infinitesimal parameters to fields. This requires enlarging the space of fields to ${\mathcal{F}}_M$. Since we are dealing with symmetries that can be described by local Lie algebras, Theorem $(1.1.2)$ provides us with a handy recipe for the BV extension $S_P$. 

In the following, we write the symmetries of the Polyakov action expressed as a function of $(X^\mu, h_{\alpha \beta})$:
worldsheet diffeomorphisms act on the coordinates as $\varphi: {x} \rightarrow {x} (\tilde{{x}})$, and on the fields as: $X \rightarrow \varphi^* X, h \rightarrow \varphi^* h$ . In a local chart, we write:
\begin{equation*} 
{x}^\alpha\rightarrow {{x}}^\alpha (\tilde{{x}}) \qquad
    X^\mu \rightarrow X^\mu ({x} (\tilde{{x}})) \qquad
        h_{\alpha \beta} \rightarrow  \frac{\partial {x}^{\alpha '}}{\partial \tilde{{x}}^\alpha}  \frac{\partial {x}^{\beta '}}{\partial \tilde{{x}}^\beta} h_{\alpha ' \beta '} ({x} (\tilde{{x}}))
\end{equation*}
and infinitesimally, for a vector field $\zeta\in\mathfrak{X}(M)$:
\begin{align} 
\delta_\zeta X &= L_\zeta X: & \qquad X^\mu &\mapsto X^\mu + \zeta^\alpha \partial_\alpha X^\mu \label{BVPolSYMA3}\\
\delta_\zeta h &= L_\zeta h: & h_{\alpha \beta} &\mapsto  h_{\alpha \beta} - \bigg( \partial_\alpha \zeta^{\lambda} h_{\lambda \beta} + \partial_\beta \zeta^{\lambda} h_{\lambda \alpha }  + \zeta^\lambda \partial_\lambda h_{\alpha \beta} \bigg)\label{BVPolSYMA4}
\end{align}
Local rescalings of the metric instead act as:
\begin{gather} h_{\alpha \beta} \rightarrow e^\phi h_{\alpha \beta} \label{BVPolSYMA5} \end{gather}
and, infinitesimally,
\begin{gather} 
\delta_\phi h = \phi h. \label{BVPolSYMA6}
\end{gather}

It is convenient to use the constrained variables ${\th}^{\alpha \beta} := \sqrt {h} h^{ \alpha \beta}$. In this way we get rid of the external rescaling invariance, and the only local symmetry left is given by the action of worldsheet diffeomorphism. Since we are using constrained variables, this reads
\begin{subequations}\label{BVPolSYMA10}\begin{align} 
\left(\delta_\zeta X\right)^\mu &= \left(L_\zeta X\right)^\mu = \zeta^\alpha \partial_\alpha X^\mu \label{BVPolSYMA10A}\\
\left(\delta_\zeta {\th}\right)_{\alpha\beta} &= \left(L_\zeta {\th}\right)_{\alpha\beta}= - \bigg( \partial_\alpha \zeta^{\lambda} {\th}_{\lambda \beta} + \partial_\beta \zeta^{\lambda} {\th}_{\lambda \alpha }  + \zeta^\lambda \partial_\lambda {\th}_{\alpha \beta} - \partial_\lambda \zeta^\lambda {\th}_{\alpha \beta} \bigg) \label{BVPolSYMA10B}
\end{align}\end{subequations}
Where we have picked up a new divergence term $\partial_\lambda \zeta^\lambda$, which effectively recovers local rescalings of the metric.

\begin{definition}\label{def:PolBVth}
We define \emph{(relaxed) non-null BV Polyakov theory} on the two-dimensional manifold with boundary $(M,\partial M)$ to be the data $(\mathcal{F}_P, \Omega_P, S_P, Q_P)$, where
\begin{equation}
    \left(\mathcal{F}_P = T^*[-1]\left(\mathcal{DPR}(M,\partial M) \times C^\infty(M,N) \times \mathfrak{X}[1](M)\right), \Omega_P\right),
\end{equation}
with $T^*\mathcal{DPR}(M,\partial M)$ the cotangent bundle of the space \emph{densitised} Lorentzian metrics (Defintion \ref{def:T^*DPR}), and using Definition \ref{def:cotangent} to understand the remaining cotangent bundles,
so that denoting the degrees of the various fields by 
$$
\left\{\begin{array}{cccccc}0 & 0 & 1 & -1 & -1 & -2 \\
{X^\mu} & \th^{\alpha \beta} & {\zeta^\alpha} & {X^\dag_\mu} & \th^\dag_{\alpha \beta} & {\zeta^\dag_\alpha}\end{array}\right\}
$$ 
the BV action reads 
\begin{equation}
S_P= S^{\text{cl}}_P + \int_M \langle X^\dag,L_\zeta X\rangle + \frac12 \langle\zeta^\dag,[\zeta,\zeta]\rangle + \langle{\th}^\dag, L_\zeta {\th}\rangle,
\end{equation}
$\Omega_P$ is the canonical $(-1)$-symplectic structure on $\mathcal{F}_P$, and $Q_P$ is the Hamiltonian vector field of $S_P$ (up to boundary terms).
\end{definition}

\begin{remark}
The variables ${\th}_{\alpha \beta}, {\th}^{\dag \alpha \beta}$ are constrained. The treatment of their variations is thus more complicated than in an unconstrained case. One of the constraints is ${\th}_{\alpha \beta} = {\th}_{\beta \alpha}$ and ${\th}^{\dag \alpha \beta} = {\th}^{\dag  \beta \alpha}$, and similarly for the variations. The second constraint is\footnote{Notice that because of the condition on the trace of ${\th}^\dag$, the rescaling term in the BV action vanishes: $\partial_\lambda\zeta^\lambda \mathrm{tr}({\th}^\dag)=0$ (cf. with Equation \eqref{BVPolSYMA10B}).}:
\begin{align}
 \mathrm{det}({\th}_{\alpha \beta}) &= -1\label{BVPolKerA6}\\
 \mathrm{tr}_h({\th}^{\dag \alpha \beta}) &= {\th}^{\dag \alpha \beta} {\th}_{\alpha \beta} = 0\label{BVPolKerA7}
\end{align}
We can express ${\th}_{\alpha \beta}$ and ${\th}^{\dag \alpha \beta}$ as functions of unconstrained fields $h_{\alpha \beta}$ and $h^{\dag \alpha \beta}$  in order to have control over their variations. We can write:
\begin{gather}
 {\th}_{\alpha \beta} = \frac{1}{\sqrt{\mathsf{h}}} h_{\alpha \beta}\label{BVPolKerA8}\\
 {\th}^{\dag \alpha \beta} = \mathsf{h} (\delta^\alpha_\rho \delta^\beta_\sigma - \frac{1}{2} {\th}^{\alpha \beta} {\th}_{\rho \sigma}) h^{\dag  \rho \sigma}\label{BVPolKerA9}
\end{gather}
where 

\begin{equation}\label{e:Porthproj}
    P^{\perp \alpha \beta}_{\rho \sigma}:= \delta^\alpha_\rho \delta^\beta_\sigma - \frac{1}{2} {\th}^{\alpha \beta} {\th}_{\rho \sigma} 
\end{equation} 
is the projector orthogonal to  $\tilde{h}_{\alpha \beta}$  that ensures condition \eqref{BVPolKerA6}, and the determinant $\mathsf{h}$ serves to make the terms in $S_P$ densities. The variations are:
\begin{gather}
\delta {\th}_{\alpha \beta} = \frac{1}{\sqrt{\mathsf{h}}} P^{\perp \rho \sigma}_{\alpha \beta}  \delta h_{\rho \sigma}\notag\\
\delta {\th}^{\dag \alpha \beta} = \delta h_{\lambda \tau}  h^{\lambda \tau} \mathsf{h} P^{\perp \alpha \beta}_{\rho \sigma} h^{\dag  \rho \sigma} 
 + \mathsf{h}  \delta(P^{\perp \alpha \beta}_{\rho \sigma}) h^{\dag  \rho \sigma} + \mathsf{h} P^{\perp \alpha \beta}_{\rho \sigma} \delta h^{\dag  \rho \sigma}=\notag\\
 \delta h_{\lambda \tau} ( {\th}^{\dag \alpha \beta} h^{\lambda \tau} )  - \delta {\th}_{\lambda \tau}   \mathsf{h}  \frac{1}{2}( {\th}^{\alpha \beta} h^{\dag  \lambda \tau} -   h^{\dag  \rho \sigma} {\th}_{\rho \sigma}{\th}^{\alpha \lambda} {\th}^{\beta \tau} ) +  \delta h^{\dag  \rho \sigma} \mathsf{h} P^{\perp \alpha \beta}_{\rho \sigma}
 \label{BVPolKerA11}
\end{gather}
Observe that, in particular, the variations of ${\th}$ are traceless.
\end{remark} 

We can now state the main result in this section.  As implicitly assumed in Equation \eqref{BVPolKerA6}, in this section we will consider the case $\text{det}(\th)=-1$. The case with $\text{det}(\th)=1$ is analogous. We will also assume that $\partial M$ is spacelike, i.e.\ $\th_{tt}\in \mathbb R^+$. In the lightlike case the restriction, being 1-dimensional, is equivalent up to an overall sign, and one can adapt the procedure outlined in the proof of Theorem \ref{THM:PolyakovRPS}.
\begin{theorem}\label{thm:BVBFVPolyakov}
Let $(M,\partial M)$ be a two dimensional manifold with boundary. The relaxed, non-null BV Polyakov theory is $1$-extendable to a BV-BFV theory, such that 
\begin{align}
    \mathcal{F}_P^{\partial} & = T^*\left(C^\infty(\partial M, N)\times \mathfrak{X}[1](\partial M)\times C^\infty[1](\partial M)\right) \\\notag
    & =  F^\partial_P\, \times\, T^*\left(\mathfrak{X}[1](\partial M)\times C^\infty[1](\partial M)\right).
\end{align}
In a local chart, where the fields of degree $1$ are 
$$
\sigma^n \in C^\infty[1](\partial M), \qquad \sigma^\partial \in \mathfrak{X}[1](\partial M),
$$ 
the fields of degree $0$ are
$$
(X,J)\in T^*C^\infty(\partial M, N),
$$
defined as the space of smooth bundle morphisms\footnote{Observe that such space of bundle mophism can bee seen as a vector bundle over $C^\infty(\partial M, N)$. In this sense we write, with a slight abuse of notation, $(X,J)\in T^*C^\infty(\partial M, N)$, see also Definition \ref{def:cotangent}. Hence, $J$ can be seen as a one-form on $\partial M$ with values in the pullback bundle $J\in \Gamma (X^*T^*N)$.  Alternatively one can think of $J$ as a ``densitised" section of the pullback bundle.} $(X,J)\colon T\partial M \to T^*N$ over smooth base maps $X\colon \partial X \to N$, and  the fields in degree $-1$ are
$$
\sigma^{\dag}_n\in \mathrm{Dens}(\partial M), \qquad \sigma^{\dag}_\partial\in \Omega^1(\partial M)\otimes \mathrm{Dens}(\partial M),$$
we have
\begin{equation}
    \Omega^\partial = \delta \alpha^\partial =   \delta \int_{\partial M} J_\mu \delta X^\mu + \sigma^{\dag}_n \delta \sigma^n + \iota_{\delta \sigma^\partial}\sigma^{\dag}_\partial,
\end{equation}
together with
\begin{equation}\label{e:PolBoundaryAction}
    S^\partial_P = \int_{\partial M} -  (L_{\sigma^\partial} X)^\mu J_\mu - \frac{1}{2} \sigma^n \bigg[ J_\mu J^\mu +  \partial_t X^\mu \partial_t X^\nu G_{\mu \nu} \bigg] + \sigma^{\dag}_n L_{\sigma^\partial} \sigma^n + \frac12 \iota_{[\sigma^\partial,\sigma^\partial]}\sigma^{\dag}_\partial
\end{equation}
\end{theorem}

\begin{remark}
The explicit expression of the surjective submersion $\pi\colon \mathcal{F}_P\to \mathcal{F}^\partial_P$ is given, in a local chart adapted to a tubular neighborhood of the inclusion $\iota\colon \partial M \to M$, by 
\begin{equation}\label{e:PolBVProjection}
    \pi\colon \begin{cases}
            J_\mu=  {\th}^{n \alpha} \partial_\alpha X^\nu G_{\mu \nu}+  \frac{1}{2} X^\dag_\mu \zeta^n\\
            X^\mu = X^\mu\\
            \sigma^n =  {\th}_{t t}^{-1}{\zeta^n}\\
            \sigma^t = {\th}_{t t}^{-1} \zeta^\alpha {\th}_{\alpha t}\\
            \sigma^{\dag}_n =- {\th}^{\dag n n} - \frac{1}{2} {\th}^{n \alpha} \zeta^\dag_\alpha \zeta^n\\
            \sigma^{\dag}_t = {\th}^{\dag n \alpha} {\th}_{\alpha t} + \frac{1}{2} \zeta^\dag_t \zeta^n
    \end{cases}
\end{equation}
\end{remark}

\subsection{Proof of Theorem \ref{thm:BVBFVPolyakov}}
We will split the proof in two parts. In the first part we will show that the kernel of the two-form induced on the space of pre-boundary fields is regular. This is sufficient to prove that the theory admits BV-BFV data, following \cite{CMR2012}. In the second part of the proof we will explicitly construct such data.

\begin{proof}[Proof. Part 1.]
The variation of $S_P$ is:
\begin{gather}
\delta S_P= EL + \check{\alpha} + \int_{\partial M} X^\dag_\mu  \zeta^n \delta X^\mu - \zeta^\dag_\alpha \zeta^n \delta \zeta^\alpha  \ {}+\notag\\
\int_{\partial M}  {\th}^{\dag \alpha \beta} \bigg[ - 2 \delta_\alpha^n \delta \zeta^{\lambda} {\th}_{\lambda \beta}  + \zeta^n \delta {\th}_{\alpha \beta}  \bigg] =\notag\notag\\
 EL + \check{\alpha} + \check{\alpha}_{BV}\label{28/05A5}
\end{gather}
Where $EL$ is a bulk term defining the Euler--Lagrange equations, while $\check{\alpha}$ and $\check{\alpha}_{BV}$ are the boundary one-forms associated to the classical Polyakov action and the BV extension respectively. The pre-symplectic form $\check{\omega} := \delta  \check{\alpha} + \delta \check{\alpha}_{BV} $ reads:
\begin{gather}
\check{\omega} = \delta  \check{\alpha}  + \int_{\partial M}\left( \delta X^\dag_\mu  \zeta^n \delta X^\mu - X^\dag_\mu  \delta \zeta^n \delta X^\mu \right)+\notag \\
-  \int_{\partial M}\left( \delta \zeta^\dag_\alpha \zeta^n \delta \zeta^\alpha + \zeta^\dag_\alpha \delta \zeta^n \delta \zeta^\alpha\right) \ {}+\notag\\
\int_{\partial M} \left\{- \delta {\th}^{\dag \alpha \beta} \bigg[  2 \delta_\alpha^n \delta \zeta^{\lambda} {\th}_{\lambda \beta}  - \zeta^n \delta {\th}_{\alpha \beta}  \bigg] + {\th}^{\dag \alpha \beta} \bigg[  2 \delta_\alpha^n \delta \zeta^{\lambda} \delta {\th}_{\lambda \beta}  - \delta \zeta^n \delta {\th}_{\alpha \beta}  \bigg]\right\}.\label{28/05A6}
\end{gather}

We want to describe the kernel of $\check\omega$ to then perform the reduction over the space of pre-boundary fields. We define a general vector over the space of pre-boundary fields as in \eqref{A10}, but considering also the additional fields introduced in the previous section. It turns out that the kernel is smooth (assuming as previously ${\th}^{n n}\neq 0$) and defined by the conditions:
\begin{subequations}\label{e:PBVker}\begin{gather}
{(\mathbb{X}_{X})^{\mu} = 0}\label{BVPolKerA1}\\
{(\mathbb{X}_{\zeta})^{\alpha} = ({\th}^{n n})^{-1} \zeta^n (\mathbb{X}_{{\th}})^{{\alpha n}}}\label{BVPolKerA2}\\
\label{BVPolKerA3}
(\mathbb{X}_{{\th}})^{{\dag n \lambda}} = -{\th}^{\dag n \beta } {\th}^{\lambda \alpha} (\mathbb{X}_{{\th}})_{{\beta \alpha}} + \frac{1}{2} {\th}^{\dag \alpha \beta} (\mathbb{X}_{{\th}})_{{\alpha \beta}} {\th}^{n \lambda}+ \\
- \frac{1}{2}  {\th}^{\lambda \alpha} \bigg(  (\mathbb{X}_{\zeta^\dag})_{\alpha} \zeta^n + \zeta_\alpha^\dag (\mathbb{X}_{\zeta})^{n} + \delta_\alpha^n \zeta_\beta^\dag (\mathbb{X}_{\zeta})^{\beta} \bigg)\notag\\
{(\mathbb{X}_{\partial_n X})^{\mu}  = \frac{1}{{\th}^{n n}}\bigg[ {\th}^{n \rho} {\th}^{\sigma \beta} (\mathbb{X}_{{\th}})_{{\rho \sigma}} \partial_\beta X^\mu - \frac{1}{2}  G^{\mu \nu}(X) \bigg( (\mathbb{X}_{X^\dag})_{\nu} \zeta^n +  X_\nu^\dag (\mathbb{X}_{\zeta})^{n} \bigg). \bigg]}\label{BVPolKerA4}
\end{gather}
\end{subequations}
The kernel of $\check{\omega}$ is then generated by the free parameters 
$\left\{(\mathbb{X}_{{\th}})_{{\alpha \beta}}, (\mathbb{X}_{\zeta^\dag})_{\lambda}, (\mathbb{X}_{X^\dag})_{\mu} \right\}$.

Indeed, we start by writing the condition $\iota_{\mathbb{X}} \check\omega = 0$: 
\begin{subequations}\label{e:PBVksimp}
\begin{gather}\label{e:PBVksimpa}
\delta \partial_n X^\nu: \qquad 
    \bigg[(\mathbb{X}_X)^\mu   G_{\mu \nu} (X) {\th}^{n n} \bigg]  = 0 \\
\delta X^\mu: \qquad 
    \bigg[ - {\th}^{n \rho} {\th}^{\sigma \beta} (\mathbb{X}_{{\th}})_{{\rho \sigma}} \partial_\beta X^\nu G_{\mu \nu} (X) + {\th}^{n n} (\mathbb{X}_{\partial_n X})^\nu   G_{\mu \nu} (X) \\ \notag
        +{\th}^{n t}  \partial_t (\mathbb{X}_X)^\nu  G_{\mu \nu} (X) + {\th}^{n \beta} \partial_\beta X^\nu \frac{\partial G_{\mu \nu} (X)}{\partial X^\rho}  (\mathbb{X}_{X})^{\rho} \\ \notag
        + \partial_t ({\th}^{n t} G_{\mu \nu} (X) (\mathbb{X}_X)^\nu  ) - {\th}^{n \beta} \partial_\beta X^\nu \frac{\partial G_{\mu \nu} (X)}{\partial X^\rho} (\mathbb{X}_{X})^{\rho} \\ \notag
        + \frac{1}{2} \bigg((\mathbb{X}_{X^\dag})_{\mu} \zeta^n +  X_\mu^\dag (\mathbb{X}_{\zeta})^{n} \bigg)\bigg]   =0 \\
\delta X^\dag_\mu: \qquad  
    \bigg[ \zeta^n (\mathbb{X}_{X})^{\mu} \bigg] = 0 \\ 
\delta {\th}_{\alpha \beta}: \qquad 
    \bigg[ (\mathbb{X}_X)^\mu   \partial_\lambda X^\nu G_{\mu \nu} (X) {\th}^{\dag n \alpha} {\th}^{\beta \lambda} +  2 {\th}^{n \alpha} (\mathbb{X}_{\zeta})^{\beta} - (\mathbb{X}_{{\th}})^{{\dag \alpha \beta}} \zeta^n - {\th}^{\dag \alpha \beta} (\mathbb{X}_{\zeta})^{n} \bigg]  \label{e:PBVksimpd}\\ 
\delta {\th}^{\dag \alpha \beta}:\qquad  
    \bigg[ 2  \delta_\alpha^n (\mathbb{X}_{\zeta})^{\lambda}  {\th}_{\lambda \beta}  + \zeta^n (\mathbb{X}_{\th})_{ {\alpha \beta}}\bigg] \label{e:PBVksimpe}\\
\delta \zeta^\lambda : \qquad 
    \bigg[\delta_\lambda^n X^\dag_\mu (\mathbb{X}_X)^\mu  - (\mathbb{X}_{\zeta^\dag})_{\lambda} \zeta^n - \zeta_\lambda^\dag (\mathbb{X}_{\zeta})^{n} - \delta_\lambda^n \zeta^\dag_\alpha (\mathbb{X}_{\zeta})^{\alpha}\\\notag
        - 2 (\mathbb{X}_{{\th}})^{{\dag n \beta}} {\th}_{\beta \lambda} - 2 {\th}^{\dag n \beta} (\mathbb{X}_{{\th}})_{{\beta \lambda}} + {\th}^{\dag \alpha \beta} (\mathbb{X}_{{\th}})_{{\alpha \beta}} \delta_\lambda^n \bigg] = 0\\
\delta \zeta^\dag_\alpha : \qquad \bigg[ \zeta^n (\mathbb{X}_{\zeta})^{\alpha} \bigg] = 0 
\end{gather} 
\end{subequations}

Notice that the variations leading to Equations \eqref{e:PBVksimpe} and \eqref{e:PBVksimpd} are constrained, so more care will need to be taken in analysing them. We can simplify Equations \eqref{e:PBVksimp} by using \eqref{e:PBVksimpa}, and assuming ${\th}^{n n} \neq 0$:
\begin{subequations}\label{e:PBVKSimpB}
\begin{gather} 
\delta \partial_n X^\nu: \qquad 
    (\mathbb{X}_X)^\mu   = 0\label{e:PBVKSimpBa}\\
\delta X^\mu : \qquad 
    \bigg[ - {\th}^{n \rho} {\th}^{\sigma \beta} (\mathbb{X}_{{\th}})_{{\rho \sigma}} \partial_\beta X^\nu G_{\mu \nu} (X) + {\th}^{n n} (\mathbb{X}_{\partial_n X})^\nu   G_{\mu \nu} (X) \label{e:PBVKSimpBb}\\
        + \frac{1}{2} \bigg((\mathbb{X}_{X^\dag})_{\mu} \zeta^n +  X_\mu^\dag (\mathbb{X}_{\zeta})^{n} \bigg) \bigg]   =0\notag\\
\delta X^\dag_\mu: \qquad  \bigg[ \zeta^n (\mathbb{X}_{X})^{\mu} \bigg] \stackrel{\eqref{e:PBVKSimpBa}}{\equiv} 0  \qquad \text{lin dip.} \label{e:PBVKSimpBc}\\
\delta {\th}_{\alpha \beta}:\qquad 
    \bigg[  2 {\th}^{\dag n \alpha} (\mathbb{X}_{\zeta})^{\beta} - (\mathbb{X}_{{\th}})^{{\dag \alpha \beta}} \zeta^n - {\th}^{\dag \alpha \beta} (\mathbb{X}_{\zeta})^{n} \bigg]  \label{e:PBVKSimpBd}\\
\delta {\th}^{\dag \alpha \beta} :\qquad 
    \bigg[ 2 \delta_\alpha^n (\mathbb{X}_{\zeta})^{\lambda}  {\th}_{\lambda \beta}  + \zeta^n (\mathbb{X}_{\th})_{ {\alpha \beta}}\bigg] \label{e:PBVKSimpBe}\\
\delta \zeta^\lambda: \qquad 
    \bigg[ - (\mathbb{X}_{\zeta^\dag})_{\lambda} \zeta^n - \zeta_\lambda^\dag (\mathbb{X}_{\zeta})^{n} - \delta_\lambda^n \zeta^\dag_\alpha (\mathbb{X}_{\zeta})^{\alpha} \label{e:PBVKSimpBf}\\
        - 2 (\mathbb{X}_{{\th}})^{{\dag n \beta}} {\th}_{\beta \lambda} - 2 {\th}^{\dag n \beta} (\mathbb{X}_{{\th}})_{{\beta \lambda}} + {\th}^{\dag \alpha \beta} (\mathbb{X}_{{\th}})_{{\alpha \beta}} \delta_\lambda^n \bigg] = 0 \notag\\
\delta \zeta^\dag_\alpha :\qquad 
    \bigg[ \zeta^n (\mathbb{X}_{\zeta})^{\alpha} \bigg] = 0 \qquad   \label{e:PBVKSimpBg}
\end{gather} \end{subequations}
where Equation \eqref{e:PBVKSimpBc} (and \eqref{e:PBVKSimpBg}, as we will see) just follows from the others. We obtain \eqref{BVPolKerA3} and \eqref{BVPolKerA4} from \eqref{e:PBVKSimpBf} and \eqref{e:PBVKSimpBb} respectively. Let us analyze now Equations \eqref{e:PBVKSimpBd} and \eqref{e:PBVKSimpBe} together. Using formula \eqref{BVPolKerA11} to express constrained variations in terms of unconstrained ones, we find:
\begin{gather} 
\delta {h}_{\alpha \beta}: \qquad
    \frac{1}{\sqrt{\mathsf{h}}}\bigg\{ P^{\perp \alpha \beta}_{\rho \sigma}  \bigg[  2 {\th}^{\dag n \rho} (\mathbb{X}_{\zeta})^{\sigma} - (\mathbb{X}_{{\th}})^{{\dag  \rho \sigma}} \zeta^n - {\th}^{\dag  \rho \sigma} (\mathbb{X}_{\zeta})^{n} \bigg]  +   \label{e:PBVunconst1}\\
        \bigg[ \underbrace{{\th}^{\dag  \lambda \tau} {\th}^{\alpha \beta}}_{A} - P^{\perp \alpha \beta}_{\rho \sigma}  \frac{\sqrt{\mathsf{h}}}{2}( {\th}^{\lambda \tau} h^{\dag  \rho \sigma} -   \underbrace{h^{\dag  \gamma \delta} {\th}_{\gamma \delta}{\th}^{\lambda \rho} {\th}^{\tau \sigma}}_{B} )\bigg]  \bigg[ 2  \delta_\lambda^n (\mathbb{X}_{\zeta})^{{\lambda '}}  {\th}_{\lambda ' \tau}  + \zeta^n (\mathbb{X})_{ {\th}_{\lambda \tau}}\bigg]\bigg\} = 0\notag\\
(\delta h^{\dag  \rho \sigma}):\qquad 
    \mathsf{h} P^{\perp \alpha \beta}_{\rho \sigma} \bigg[ 2  \delta_\alpha^n (\mathbb{X}_{\zeta})^{\lambda}  {\th}_{\lambda \beta}  + \zeta^n (\mathbb{X})_{ {\th}_{\alpha \beta}}\bigg]=0 \label{e:PBVunconst2}
\end{gather}

We claim that \eqref{BVPolKerA2} follows from \eqref{e:PBVunconst2}, and that \eqref{e:PBVunconst1} is identically satisfied, provided that  \eqref{BVPolKerA2} and  \eqref{BVPolKerA3} hold. To see this, let us consider first Equation \eqref{e:PBVunconst2}:
\begin{gather*} 
 P^{\perp \alpha \beta}_{\ {} (\rho \sigma)} \bigg[ 2  \delta_\alpha^n (\mathbb{X}_{\zeta})^{\lambda}  {\th}_{\lambda \beta}  + \zeta^n (\mathbb{X})_{ {\th}_{\alpha \beta}}\bigg]=0
\end{gather*}
{where $P^{\perp \alpha \beta}_{(\rho \sigma)} := \frac{1}{2} ( P^{\perp \alpha \beta}_{\rho \sigma} + P^{\perp \alpha \beta}_{\sigma \rho})$. Using the explicit form of $P^{\perp \alpha \beta}_{\rho \sigma}$ (Equation \eqref{e:Porthproj}) we get (for $\sigma=\rho$)}:
\begin{gather}\label{EqForEta}
2(\mathbb{X}_{\zeta})^{\lambda} {\th}_{\lambda \rho} \delta^n_\rho + \zeta^n (\mathbb{X}_{{\th}})_{{\rho \rho}}- {\th}_{\rho \rho} (\mathbb{X}_{\zeta})^{n}=0
\end{gather}
and if $\rho \neq\sigma$:
\begin{gather*}
 (\mathbb{X}_{\zeta})^{\lambda} {\th}_{\lambda t} +  \zeta^n (\mathbb{X}_{{\th}})_{{n t}} -  {\th}_{n t} (\mathbb{X}_{\zeta})^{n}=0
\end{gather*}
Taking into account the cases $\rho=\sigma=t$ and $\rho=\sigma=n$, we will now derive eq. $\eqref{BVPolKerA2}$. The case $\rho \neq\sigma$ follows then from eq. $\eqref{BVPolKerA2}$ itself. Let us begin with $\rho=\sigma=t$. Eq. $\eqref{EqForEta}$ becomes:
\begin{gather}\label{EqForEtan}
 (\mathbb{X}_{\zeta})^{n}=({\th}_{t t})^{-1} \zeta^n (\mathbb{X}_{{\th}})_{{t t}}
\end{gather}
If $\rho=\sigma=n$, Eq. $\eqref{EqForEta}$ becomes:
\begin{gather*}
  2 (\mathbb{X}_{\zeta})^{t} {\th}_{n t}=-  \zeta^n (\mathbb{X}_{{\th}})_{{n n}} - (\mathbb{X}_{\zeta})^{n} {\th}_{n n}
\end{gather*}
and, using $\eqref{EqForEtan}$ together with the condition $(\mathbb{X})_{\textit{det}(h)}=0$, the RHS becomes:
\begin{gather*}
 -\zeta^n ({\th}_{tt})^{-1} [{\th}_{tt}(\mathbb{X}_{{\th}})_{{nn}} + (\mathbb{X}_{{\th}})_{{tt}}{\th}_{nn}]= -2\zeta^n ({\th}_{tt})^{-1}{\th}_{nt} (\mathbb{X}_{{\th}})_{{nt}}
\end{gather*}
which means:
\begin{gather}\label{EqForEtat}
  (\mathbb{X}_{\zeta})^{t}=-  ({\th}_{t t})^{-1} \zeta^n (\mathbb{X}_{{\th}})_{{n t}}
\end{gather}
Using ${\th}_{tt}={-\th}^{nn} $ and ${\th}^{n\alpha}= -\varepsilon^{\alpha \lambda} {\th}_{\lambda t} $, Equations $\eqref{EqForEtan}$ and $\eqref{EqForEtat}$ can be expressed as $\eqref{BVPolKerA2}$:
\begin{gather*}
  (\mathbb{X}_{\zeta})^{\alpha}= ({\th}^{n n})^{-1} \zeta^n (\mathbb{X}_{{\th}})^{{n \alpha}}
\end{gather*}

Going back to $\eqref{e:PBVunconst1}$, eq $\eqref{e:PBVunconst2}$ implies that the terms $A$ and $B$ vanish. In fact, if a tensor $K_{\alpha \beta}$ is such that $P^{\perp \alpha \beta}_{\ {} (\rho \sigma)}K_{\alpha \beta}=0$, then also ${\th}^{\dag \alpha \beta} K_{\alpha \beta} := h h^{\dag \rho \sigma}P^{\perp \alpha \beta}_{\ {} \rho \sigma}K_{\alpha \beta}=0$ and $P^{\perp (\alpha \beta)}_{\ {} \rho \sigma} {\th}^{\lambda \rho}{\th}^{\tau \sigma} K_{\tau \sigma}=0$. Furthermore:
\begin{gather*} 
\frac{\sqrt{\mathsf{h}}}{2}{\th}^{\lambda \tau} h^{\dag  \rho \sigma}  \bigg[ 2  \delta_\lambda^n (\mathbb{X}_{\zeta})^{{\lambda '}}  {\th}_{\lambda ' \tau}  + \zeta^n (\mathbb{X})_{ {\th}_{\lambda \tau}}\bigg] =   {\th}^{\dag  \rho \sigma} (\mathbb{X}_{\zeta})^{n}
\end{gather*}
and Equation \eqref{e:PBVunconst1} becomes:
\begin{gather*} 
\frac{1}{\sqrt{\mathsf{h}}}\bigg\{ P^{\perp \alpha \beta}_{\rho \sigma}  \bigg[  2 {\th}^{\dag n \rho} (\mathbb{X}_{\zeta})^{\sigma} - (\mathbb{X}_{{\th}})^{{\dag  \rho \sigma}} \zeta^n - 2 {\th}^{\dag  \rho \sigma} (\mathbb{X}_{\zeta})^{n} \bigg]  +   \\
 {\th}^{\dag  \lambda \tau} {\th}^{\alpha \beta}   \bigg[ 2  \delta_\lambda^n (\mathbb{X}_{\zeta})^{{\lambda '}}  {\th}_{\lambda ' \tau}  + \zeta^n (\mathbb{X})_{ {\th}_{\lambda \tau}}\bigg]\bigg\} = 0.
\end{gather*}
Writing the projector explicitly:
\begin{gather*} 
\frac{1}{\sqrt{\mathsf{h}}}\bigg\{   2 {\th}^{\dag n \alpha} (\mathbb{X}_{\zeta})^{\beta} - (\mathbb{X}_{{\th}})^{{\dag \alpha \beta}} \zeta^n - 2 {\th}^{\dag \alpha \beta} (\mathbb{X}_{\zeta})^{n}   +   \\
+ {\th}^{\alpha \beta} \bigg[ \frac{1}{2} {\th}_{\rho \sigma} (\mathbb{X}_{{\th}})^{{\dag  \rho \sigma}} \zeta^n - {\th}^{\dag n \rho} {\th}_{\rho \sigma} (\mathbb{X}_{\zeta})^{{\sigma}} \bigg] + \\
 {\th}^{\alpha \beta}   \bigg[ 2   {\th}^{\dag n \rho} {\th}_{\rho \sigma} (\mathbb{X}_{\zeta})^{{\sigma}}  + {\th}^{\dag  \rho \sigma} \zeta^n (\mathbb{X})_{ {\th}_{\rho \sigma}}\bigg]\bigg\} = 0.
\end{gather*}
The third line is equal to twice the second one, with inverted sign, as a consequence of $\mathbb{X}(\mathrm{Tr}_{{\th}}{\th}^\dag) = 0$. Thus:
\begin{gather*} 
\frac{1}{\sqrt{\mathsf{h}}}\bigg\{   2 {\th}^{\dag n \alpha} (\mathbb{X}_{\zeta})^{\beta} - (\mathbb{X}_{{\th}})^{{\dag \alpha \beta}} \zeta^n - 2 {\th}^{\dag \alpha \beta} (\mathbb{X}_{\zeta})^{n}   +   \\
- {\th}^{\alpha \beta} \bigg[ \frac{1}{2} {\th}_{\rho \sigma} (\mathbb{X}_{{\th}})^{{\dag  \rho \sigma}} \zeta^n - {\th}^{\dag n \rho} {\th}_{\rho \sigma} (\mathbb{X}_{\zeta})^{{\sigma}} \bigg]\bigg\} = 0
\end{gather*}
Now, inserting the expressions of $(\mathbb{X}_{\zeta})^{\lambda}$ and $(\mathbb{X}_{{\th}})^{{\dag \alpha \beta}}$ in \eqref{e:PBVunconst1}, it is possible to check that it vanishes identically. Equations \eqref{e:PBVker} then show that the kernel is regular and allows pre-symplectic reduction.
\end{proof}

\begin{remark}
In the first part of the proof we have shown that the kernel of the boundary two-form $\check{\omega}$ is regular. This means that it is possible to construct the pre-symplectic reduction $\mathcal{F}^\partial_P = \check{\mathcal{F}}_P / \mathrm{ker}(\check{\omega}^\sharp)$, and the rest of the BFV structure will follow as a consequence of \cite{CMR2012}.
In part 2 of the proof we perform said pre-symplectic reduction over the space of fields, and construct an explicit chart for $\mathcal{F}_P^\partial$. In order to do this, we will explicitly flow the fields in $\check{\mathcal{F}}_P$ along the vertical vector fields in the kernel of $\check{\omega}$.
\end{remark}

\begin{proof}[Proof. Part 2.]
To obtain the explicit expression of the projection given in \eqref{e:PolBVProjection}, we have to solve the set of differential equations related to the flow produced by the kernel vectors defined by Equations \eqref{e:PBVker}. The free parameters in the kernel are $(\mathbb{X}_{{\th}})_{{\alpha \beta}}, (\mathbb{X}_{\zeta^\dag})^{\lambda}, (\mathbb{X}_{X^\dag})^{\mu}$. The system of differential equations is:

\begin{gather*}
\partial_\tau{X^\mu}= 0
\end{gather*} 
\begin{gather*}
\partial_\tau{\zeta^\alpha} = ({\th}^{n n})^{-1} \zeta^n (\mathbb{X}_{{\th}})^{{\alpha n}}
\end{gather*} 
\begin{gather*}
\partial_\tau{{\th}^{\dag n \lambda}} = -{\th}^{\dag n \beta } {\th}^{\lambda \alpha} (\mathbb{X}_{{\th}})_{{\beta \alpha}} + \frac{1}{2} {\th}^{\dag \alpha \beta} (\mathbb{X}_{{\th}})_{{\alpha \beta}} {\th}^{n \lambda}+ \notag\\
-\frac{1}{2}  {\th}^{\lambda \alpha} \bigg(  (\mathbb{X}_{\zeta^\dag})_{\alpha} \zeta^n + \zeta_\alpha^\dag \partial_\tau{\zeta^n} + \delta_\alpha^n \zeta_\beta^\dag \partial_\tau{\zeta^\beta} \bigg)
\end{gather*} 
\begin{gather*}
\partial_\tau{\partial_n X^\mu} = \frac{1}{{\th}^{n n}}\bigg[ {\th}^{n \rho} {\th}^{\sigma \beta} (\mathbb{X}_{{\th}})_{{\rho \sigma}} \partial_\beta X^\mu -   \frac{1}{2}G^{\mu \nu}(X) \bigg( (\mathbb{X}_{X^\dag})_{\nu} \zeta^n +  X_\nu^\dag\partial_\tau{\zeta^n} \bigg) \bigg]
\end{gather*}
\begin{gather*}
\partial_\tau {\th}_{\alpha \beta} =(\mathbb{X}_{{\th}})_{{\alpha \beta}}
\end{gather*}
\begin{gather*}
\partial_\tau \zeta_\lambda^\dag = (\mathbb{X}_{\zeta^\dag})_{\lambda}
\end{gather*}
\begin{gather*}
\partial_\tau X_\mu^\dag =(\mathbb{X}_{X^\dag})_{\mu}
\end{gather*}
We will flow along three vector fields that span the kernel, each one related to one of the aforementioned free parameters, and  the order will be chosen in such a way to simplify the differential equations. In the first part the parameter of the flow $\tau$ will range from $0$ to $1$, while in the second part it will range from $1$ to $2$ and in the third part from $2$ to $3$. An explicit chart-expression for the projection map $\mathcal{F}_P \to \mathcal{F}^\partial_P$ will be given by the value of the field at the end of the composite flow.

We start using  $(\mathbb{X}_{X^\dag})^{\mu}$ and setting $(\mathbb{X}_{{\th}})_{{\alpha \beta}}$ and $(\mathbb{X}_{\zeta^\dag})^{\lambda}$ to zero. In this case there are only two non-trivial differential equations:

\begin{gather*}
\partial_\tau\partial_n X^\mu = -\frac{1}{2 {\th}^{n n}} G^{\mu \nu} \partial_\tau{X}^\dag_\nu \zeta^n\\
\partial_\tau X_\mu^\dag = (\mathbb{X}_{X^\dag})_{\mu}
\end{gather*}
where ${\th}^{n n}, G^{\mu \nu}(X)$ and $\zeta^n$ do not depend on $\tau$. We choose a value of  $(\mathbb{X}_{X
^\dag})^{\mu}$ that sets $X^\dag_\mu(\tau=1)$ to zero (e.g. $(\mathbb{X}_{X^\dag})^{\mu}(\tau) = - X^\dag_\mu (0)$). We obtain:
\begin{gather}
\partial_n X^\mu (1) = \bigg( \partial_n X^\mu + \frac{1}{2 {\th}^{n n}} G^{\mu \nu} X^\dag_\nu \zeta^n \bigg)_{|_{\tau=0}}\label{31/05A5}\\
X^\dag_\mu (1) = 0\notag
\end{gather}
while the other fields remain unaffected. In the second part of the flow, we take a non-vanishing choice of $(\mathbb{X}_{\zeta^\dag})^{\lambda}$, and set  $(\mathbb{X}_{X^\dag})^{\mu}$ and  $(\mathbb{X}_{{\th}})_{{\alpha \beta}}$ to zero. The only non-trivial differential equations are then:

\begin{gather*}
\partial_\tau {\th}^{\dag n \lambda} = - \frac{1}{2} {\th}^{\lambda \alpha}\partial_\tau \zeta^\dag_\alpha \zeta^n\\
\partial_\tau \zeta^\dag_\alpha = (\mathbb{X}_{\zeta^\dag})^{\alpha}
\end{gather*}
where ${\th}^{\lambda \alpha}$ and $\zeta^n$ do not depend on $\tau$. Similarly to the previous case, we choose a value of $(\mathbb{X}_{\zeta^\dag})^{\alpha}$ that sets $\zeta^\dag_\alpha(\tau=2)$ to zero (e.g. $(\mathbb{X}_{\zeta^\dag})^{\alpha} = - \zeta^\dag_\alpha (1)$, with $\zeta^\dag_\alpha (1)$ the value of $\zeta^\dag(\tau)$ at the start of this iteration of the flow, $\tau\in[1,2]$). The solution is:
\begin{gather}
{\th}^{\dag n \lambda} (2) = \bigg( {\th}^{\dag n \lambda} +   \frac{1}{2} {\th}^{\lambda \alpha}\zeta^\dag_\alpha \zeta^n \bigg)_{|_{\tau = 0}}\label{31/05A3}\\
\zeta^\dag_\alpha (2) = 0\notag
\end{gather}
and again the other fields remain unmodified. The third part of the flow is characterized by a non-vanishing choice of $(\mathbb{X}_{{\th}})_{{\alpha \beta}}$, while the other parameters are set to be zero. The non-trivial differential equations are:
 
\begin{gather}
\partial_\tau \zeta^\alpha = \frac{1}{{\th}^{n n}} \zeta^n \partial_\tau {\th}^{\alpha n}\\
\partial_\tau  {\th}^{\dag n \lambda} = - {\th}^{\dag n \beta}  {\th}^{\lambda  \alpha} \partial_\tau  {\th}_{\beta \alpha} + \frac{1}{2}  {\th}^{\dag \alpha \beta} \partial_\tau  {\th}_{\alpha \beta}  {\th}^{n \lambda} \label{e:diffeq2}\\
\partial_\tau \partial_n X^\mu= \frac{1}{{\th}^{n n}} {\th}^{n \rho} {\th}^{ \sigma \beta} \partial_\tau {\th}_{\rho \sigma} \partial_\beta X^\mu\\
\partial_\tau  {\th}_{\alpha \beta} = (\mathbb{X})_{ {\th}_{\alpha \beta}}
\end{gather}
where the absence of terms dependent on $X^\dag_\mu$ and $\zeta^\dag_\alpha$ is due to the fact that we set them to zero in the first two parts of the flow, and this makes calculations easier. The first and third lines can be solved directly, while the second line must be treated. We are thus going to inspect the expression in the second line. First, we eliminate ${\th}^{\dag t t}$ using the tracelessness condition ${\th}^{\dag \alpha \beta}{\th}_{\alpha \beta}=0$:

\begin{gather}\label{31/05A2}
{\th}^{\dag t t} = - \frac{1}{{\th}_{t t}} \bigg( 2 {\th}^{\dag n t} {\th}_{n t} + {\th}^{\dag n n}{\th}_{n n}\bigg)
\end{gather}
Then, Equation \eqref{e:diffeq2} becomes: 
\begin{gather*}
{\th}^{\dag n n} \bigg( - {\th}^{\lambda \beta} \partial_\tau {\th}_{\beta n} + \frac{{\th}^{n \lambda}}{2}  \partial_\tau {\th}_{n n}  - \frac{{\th}^{n \lambda}}{2} \frac{ {\th}_{n n}}{{\th}_{t t}} \partial_\tau {\th}_{t t} \bigg)+\\
{\th}^{\dag n t} \bigg( - {\th}^{\lambda \beta} \partial_\tau {\th}_{\beta t} + {\th}^{n \lambda} \partial_\tau{\th}_{n t}  -  {\th}^{n \lambda} \frac{{\th}_{n t} }{{\th}_{t t}} \partial_\tau {\th}_{t t}  \bigg)=\\
-{\th}^{\dag n n} \bigg(  {\th}^{\lambda t} \partial_\tau {\th}_{t n} + \frac{{\th}^{n \lambda}}{2}  \partial_\tau {\th}_{n n}  + \frac{{\th}^{n \lambda}}{2} \frac{ {\th}_{n n}}{{\th}_{t t}}\partial_\tau {\th}_{t t} \bigg)+\\
-{\th}^{\dag n t} \bigg(  {\th}^{\lambda t} \partial_\tau {\th}_{t t} +{\th}^{n \lambda} \frac{{\th}_{n t} }{{\th}_{t t}} \partial_\tau {\th}_{t t}  \bigg)= \\
-{\th}^{\dag n n} \frac{1}{{\th}_{t t}} \bigg(  {\th}^{\lambda t} \partial_\tau {\th}_{t n} {\th}_{t t}+ \frac{{\th}^{n \lambda}}{2}  \partial_\tau {\th}_{n n} {\th}_{t t} + \frac{{\th}^{n \lambda}}{2}  {\th}_{n n}\partial_\tau {\th}_{t t} \bigg)+\\
-{\th}^{\dag n t}  \frac{1}{{\th}_{t t}} \delta^\lambda_t \partial_\tau {\th}_{t t}.
\end{gather*}
We now express, where needed, $ {\th}^{\alpha \beta}$ in function of $ {\th}_{\alpha \beta}$ through the relation $ {\th}^{\alpha \beta} =\, \varepsilon^{\alpha \alpha'} \varepsilon^{\beta' \beta}  {\th}_{\alpha' \beta'}$ (cf. with Equation \eqref{A6}), where $\varepsilon^{\alpha\beta}$ is the two-dimensional Levi-Civita symbol). We then have:

\begin{gather*}
-{\th}^{\dag n n} \frac{1}{{\th}_{t t}} \bigg(  {\th}^{\lambda t} \partial_\tau {\th}_{t n} {\th}_{t t}+ \frac{{\th}^{n \lambda}}{2}  \partial_\tau ({\th}_{n n} {\th}_{t t}) \bigg)
-{\th}^{\dag n t}  \frac{1}{{\th}_{t t}} \delta^\lambda_t \partial_\tau {\th}_{t t}=\\ 
-{\th}^{\dag n n} \frac{1}{{\th}_{t t}} \bigg(\varepsilon^{\lambda \alpha}\varepsilon^{\beta t} {\th}_{\alpha \beta} \partial_\tau {\th}_{t n} {\th}_{t t}+ \varepsilon^{\lambda \alpha}\varepsilon^{\beta n}  \frac{{\th}_{\alpha \beta}}{2}  \partial_\tau ({\th}_{n n} {\th}_{t t}) \bigg)
-{\th}^{\dag n t}  \frac{1}{{\th}_{t t}} \delta^\lambda_t \partial_\tau {\th}_{t t}= \\ 
-{\th}^{\dag n n} \frac{1}{{\th}_{t t}} \varepsilon^{\lambda \alpha}\varepsilon^{n t} \bigg( {\th}_{\alpha n} \partial_\tau {\th}_{t n} {\th}_{t t} -  \frac{{\th}_{\alpha t}}{2}  \partial_\tau ({\th}_{n n} {\th}_{t t}) \bigg)
-{\th}^{\dag n t}  \frac{1}{{\th}_{t t}} \delta^\lambda_t \partial_\tau {\th}_{t t} = \\ 
-{\th}^{\dag n n} \frac{1}{{\th}_{t t}} \frac{1}{2 {\th}_{n t}}\varepsilon^{\lambda \alpha}\varepsilon^{n t} \bigg( 
\partial_\tau ({\th}_{n t}^2)  {\th}_{\alpha n}  {\th}_{t t} - {\th}_{\alpha t} {\th}_{n t}  \partial_\tau ({\th}_{n n} {\th}_{t t}) \bigg)
-{\th}^{\dag n t}  \frac{1}{{\th}_{t t}} \delta^\lambda_t \partial_\tau {\th}_{t t}=\\
-{\th}^{\dag n n} \frac{1}{{\th}_{t t}} \delta^\lambda_t \partial_\tau {\th}_{n t} -{\th}^{\dag n t}  \frac{1}{{\th}_{t t}} \delta^\lambda_t \partial_\tau {\th}_{t t}
\end{gather*}
where the last equality follows from the fact that $
\partial_\tau ({\th}_{n t}^2)  {\th}_{\alpha n}  {\th}_{t t} - {\th}_{\alpha t} {\th}_{n t}  \partial_\tau ({\th}_{n n} {\th}_{t t}) $ vanishes for $\alpha=t$, since 
$$
\partial_\tau ({\th}_{n t}^2)  {\th}_{t n}  {\th}_{t t} - {\th}_{t t} {\th}_{n t}  \partial_\tau ({\th}_{n n} {\th}_{t t}) = -{\th}_{t t} {\th}_{n t} \partial_\tau (\mathrm{det}({\th})) = 0,
$$
while for $\alpha=n$:
\begin{gather*}
\partial_\tau ({\th}_{n t}^2)  {\th}_{n n}  {\th}_{t t} - {\th}_{n t} {\th}_{n t}  \partial_\tau ({\th}_{n n} {\th}_{t t})\\
=( {\th}_{n n}  {\th}_{t t} - {\th}_{n t} {\th}_{n t}) \partial_\tau ({\th}_{n t}^2) = - 2  {\th}_{n t} \partial_\tau ({\th}_{n t})
\end{gather*}
since $ \th_{nn} \th_{tt} - \th_{nt}^2=\mathrm{det}(\th)=-1 $, and it follows that:
\begin{gather*}
\partial_\tau {\th}^{\dag n \lambda}= -{\th}^{\dag n n} \frac{1}{{\th}_{t t}} \delta^\lambda_t \partial_\tau {\th}_{n t} -{\th}^{\dag n t}  \frac{1}{{\th}_{t t}} \delta^\lambda_t \partial_\tau {\th}_{t t}= - \frac{\delta^\lambda_t}{{\th}_{t t}} {\th}^{\dag n \alpha} \partial_\tau {\th}_{\alpha t}
\end{gather*}
and the set of differential equations becomes:

\begin{gather*}
\partial_\tau \zeta^\alpha =  \frac{1}{{\th}^{n n}} \zeta^n \partial_\tau {\th}^{\alpha n}
\end{gather*}
\begin{gather*}
\partial_\tau {\th}^{\dag n \lambda} = - \frac{\delta^\lambda_t}{{\th}_{t t}} {\th}^{\dag n \alpha} \partial_\tau {\th}_{\alpha t} 
\end{gather*}
\begin{gather*}
\partial_\tau \partial_n X^\mu= \frac{1}{{\th}^{n n}} {\th}^{n \rho} {\th}^{ \sigma \beta} \partial_\tau {\th}_{\rho \sigma} \partial_\beta X^\mu
\end{gather*}
\begin{gather*}
\partial_\tau  {\th}_{\alpha \beta} = (\mathbb{X})_{ {\th}_{\alpha \beta}}
\end{gather*}
Since the three differential equations in the system are decoupled, we are able to solve them separately. We choose a path that sends ${\th}_{\alpha \beta}$ to the Minkowski metric\footnote{In the Riemannian case, $\eta$ denotes the 2d Euclidean metric.} (again, we could choose $(\mathbb{X}_{{\th}})_{{n \alpha}} = \eta_{n \alpha} - {\th}_{n \alpha}$). It is worth pointing out that there is no differential equation for ${\th}^{\dag t t}$, but it is not a problem: ${\th}^{\dag \alpha \beta}$ has only two degrees of freedom and it is possible to express ${\th}^{\dag t t}$ in function of the other elementary fields through eq. \eqref{31/05A2}. We start from the differential equation for $\zeta^\alpha$:

\begin{gather*}
\partial_\tau \zeta^n = \frac{1}{{\th}^{n n}} \zeta^n \partial_\tau {\th}^{n n}=\zeta^n  \frac{ \partial_\tau {\th}_{t t}}{{\th}_{t t}}  \\
\partial_\tau \zeta^t = \frac{1}{{\th}^{n n}} \zeta^n \partial_\tau {\th}^{t n} =  - \zeta^n \frac{\partial_\tau {\th}_{n t}}{{\th}_{t t}} 
\end{gather*}
we solve the system before for $\zeta^n$ and then for $\zeta^t$. The solution for $\zeta^n$ is: $\zeta^n (\tau) = \zeta^n (2) \frac{{\th}_{t t} (\tau)}{{\th}_{t t} (2)}$. Then:

\begin{gather*}
\zeta^n (\tau) =\zeta^n (2) \frac{{\th}_{t t} (\tau)}{{\th}_{t t} (2)}  \\
\partial_\tau \zeta^t = - \zeta^n (2) \frac{\partial_\tau {\th}_{n t} (\tau)}{{\th}_{t t}(2)} 
\end{gather*}
whose solution is $\zeta^t (\tau) =\zeta^t (2) -\zeta^n (2) \frac{ {\th}_{n t} (\tau)}{{\th}_{t t}(2)} + \zeta^n (2) \frac{\tau {\th}_{n t} (2)}{{\th}_{t t}(2)}$. Evaluating the solutions at $\tau=3$ and expressing all the quantities in function of fields evaluated at $\tau=0$:

\begin{gather*}
\zeta^n (3) = \bigg(\zeta^n \frac{1}{{\th}_{t t} }\bigg)_{|_{\tau=0}}  \\
 \zeta^t (3) = \bigg(\zeta^t  +\zeta^n \frac{ {\th}_{n t}}{{\th}_{t t}}\bigg)_{|_{\tau=0}} 
\end{gather*}
We solve now the differential equation for ${\th}^{\dag n \lambda}$:
\begin{gather*}
\partial_\tau {\th}^{\dag n n} =0 \\
\partial_\tau {\th}^{\dag n t} = - \frac{1}{{\th}_{t t}} {\th}^{\dag n \alpha} \partial_\tau {\th}_{\alpha t} 
\end{gather*}
${\th}^{\dag n n}$ is unaffected, and we have to solve only for ${\th}^{\dag n t}$. Manipulating the second line we obtain:
\begin{gather*}
\partial_\tau {\th}^{\dag n t} {\th}_{t t} + {\th}^{\dag n t} \partial_\tau {\th}_{t t} = - {\th}^{\dag n n} \partial_\tau {\th}_{n t} \\
\partial_\tau ({\th}^{\dag n t} {\th}_{t t}) = - \partial_\tau ( {\th}^{\dag n n}{\th}_{n t} )
\end{gather*}
the solution is then ${\th}^{\dag n t}{\th}_{t t} (\tau) =  {\th}^{\dag n t}  {\th}_{t t} (2) - {\th}^{\dag n n}{\th}_{n t} (\tau) + {\th}^{\dag n n}{\th}_{n t} (2)$. Isolating ${\th}^{\dag n t}$ and evaluating at $\tau=3$, we obtain:
\begin{gather*}
{\th}^{\dag n n} (3) = {\th}^{\dag n n} (2)\\
{\th}^{\dag n t} (3) = {\th}^{\dag n \beta} {\th}_{\beta t} (2) 
\end{gather*}
and to express ${\th}^{\dag n \lambda}$ in function of fields evaluated at $\tau=0$ we have to concatenate this transformation with the transformation caused by the second part of the flow (check eq. \eqref{31/05A3}). We then have:
\begin{gather*}
{\th}^{\dag n n} (3) = \bigg( {\th}^{\dag n n} +   \frac{1}{2} {\th}^{n \alpha} \zeta^\dag_\alpha \zeta^n \bigg)_{|_{\tau = 0}}\\ 
{\th}^{\dag n t} (3) = \bigg( {\th}^{\dag n \beta} +   \frac{1}{2} {\th}^{\beta \alpha} \zeta^\dag_\alpha \zeta^n \bigg)_{|_{\tau = 0}}{\th}_{\beta t} (0).
\end{gather*}
At last, we have to solve for $\partial_n X^\mu$, with the equation:
\begin{gather*}
\partial_\tau \partial_n X^\mu= \frac{1}{{\th}^{n n}} {\th}^{n \rho} {\th}^{ \sigma \beta} \partial_\tau {\th}_{\rho \sigma} \partial_\beta X^\mu
\end{gather*}
This is the same differential equation as in \eqref{31/05A4}, and thus the solution is the same as in \eqref{A15}:
\begin{gather*}
\partial_n X^\mu (3) =  {\th}^{n \alpha} \partial_\alpha X^\mu (2)
\end{gather*}
and concatenating this transformation with eq \eqref{31/05A5}:
\begin{gather*}
\partial_n X^\mu (3) = \bigg( {\th}^{n \alpha} \partial_\alpha X^\mu + \frac{1}{2} G^{\mu \nu} X^\dag_\nu \zeta^n\bigg)_{|_{\tau=0}}
\end{gather*}

The reduction then sets $({\th}_{\alpha \beta}, \zeta^\dag_\lambda, X^\dag_\mu) \ {}    \rightarrow ( \eta_{\alpha\beta}, 0, 0)$, and denoting coordinates in $\mathcal{F}^\partial$ with $(J_\mu, X^\mu, \sigma^n,\sigma^t, \sigma^\dag_n \sigma^\dag_t)$ we have the projection:
\begin{equation}
\pi\colon \begin{cases}
J_\mu &=  {\th}^{n \alpha} \partial_\alpha X^\nu G_{\mu \nu}+  \frac{1}{2} X^\dag_\mu \zeta^n\\
X^\mu &= X^\mu\\
\sigma^n & =  {\th}_{t t}{}^{-1}{\zeta^n}\\
\sigma^t &= {\th}_{t t}{}^{-1} \zeta^\alpha {\th}_{\alpha t}\\
\sigma^{\dag}_n & = -{\th}^{\dag n n} - \frac{1}{2} {\th}^{n \alpha} \zeta^\dag_\alpha \zeta^n\\
\sigma^\dag_t &= {\th}^{\dag n \alpha} {\th}_{\alpha t} + \frac{1}{2} \zeta^\dag_t \zeta^n
\end{cases}
\end{equation}  \\
where $\sigma_\alpha^\dag=\th^{\dag n \lambda}(3)\eta_{\lambda \alpha}$. Thus, we perform a reduction that fixes the metric ${\th}_{\alpha \beta}$ to $\eta_{\alpha\beta}$ without changing the sign of ${\th}_{t t}$. 

The two-form $\check{\omega}=\delta(\check{\alpha}+\check{\alpha}_{BV})$ is basic with respect to the exact symplectic form ($\beta\in\{n,t\}$):
\begin{gather}
\Omega_P^\partial =   \delta \int_{\partial M} J_\mu \delta X^\mu +  \sigma^{\dag}_\beta \delta \sigma^\beta 
\end{gather}
i.e.\ $\check{\omega} = \pi^*\Omega^\partial_P$ and the BFV action $S^\partial_P$ is computed to be (see Appendix $\ref{AppendixBoundaryAction}$) :
\begin{gather}\label{06/06A2}
S^\partial_P = \int_{\partial M} -  \sigma^t \partial_t X^\mu J_\mu - \frac{1}{2} \sigma^n \bigg[ J_\mu J^\mu +  \partial_t X^\mu \partial_t X^\nu G_{\mu \nu} \bigg] + \sigma^\dag_\alpha \sigma^t \partial_t \sigma^\alpha
\end{gather}
\end{proof}

We have thus defined a BV-BFV structure 
$$
(\mathcal{F}_M, S_P, Q_M, \Omega_M, \mathcal{F}_{\partial M}^\partial, S_{P}^\partial, Q_{\partial M}, \Omega_{\partial M}, \pi),
$$
where  $(\mathcal{F}_M, S_P, Q_M, \Omega_M)$ is the BV part (in the bulk) , $\pi:= \check{\pi}  \circ \tilde{\pi}$ is the composition of the restriction to the boundary map $\tilde\pi$ and of the reduction map $\check\pi$ (defined in  \eqref{e:PolBVProjection}) and $( \mathcal{F}_{\partial M}^\partial, S_{P}^\partial, Q_{\partial M}, \Omega_{\partial M})$ is the BFV part (on the boundary).

\section{Nambu--Goto theory --- BV-BFV Analysis}

The analysis of the BV-Nambu--Goto action is analogous to that of Polyakov theory. In the first section we will identify the symmetries, in the second section we will build the (broken) BV structure, and in the third section we will identify the kernel. Here we will see a different behaviour than in the case of the Polyakov action: the kernel is  not regular.

\subsection{Symmetries of the Nambu Goto action}
Nambu--Goto string theory is invariant under the action of worldsheet diffeomorphisms, which act on the coordinates as $\varphi: {x} \rightarrow {x} (\tilde{{x}})$ and on the fields as: $X \rightarrow \varphi^* X$. In a local chart:
\begin{gather} 
{x}^\alpha\rightarrow {{x}}^\alpha (\tilde{{x}}) \qquad
X^\mu \rightarrow X^\mu ({{x}} (\tilde{{x}})) \label{BVNGSYMA1}
\end{gather}
and infinitesimally, the action on fields reads:
\begin{gather} 
\delta_\zeta X = L_\zeta X: \qquad X^\mu \rightarrow X^\mu + \zeta^\alpha \partial_\alpha X^\mu \label{BVNGSYMA2}
\end{gather}
Observe that the metric is not an independent field, and there is thus no rescaling symmetry.

We build now the BV structure following the same steps as in the previous chapter. We  promote to ghost fields the infinitesimal parameters of the infinitesimal symmetry transformations. In this case, we have only the degree $+1$ ghost field $\zeta^\alpha$. We then add the degree $-1$ anti-fields $X_\mu^\dag$ and the degree $-2$ anti-ghosts $\zeta_\alpha^\dag$.  We have again an irreducible gauge group, and thus we do not have to introduce higher order ghosts or anti-ghosts.

\begin{definition}\label{def:BVNGth}
We define (relaxed) BV Nambu--Goto theory on the two-dimensional manifold with boundary $(M,\partial M)$ to be the data $(\mathcal{F}_{NG}, \Omega_{NG}, S_{NG}, Q_{NG})$, where 
\begin{equation}
    \mathcal{F}_{NG} := T^*[-1]\left( C^\infty(M,N) \times \mathfrak{X}[1](M)\right) 
\end{equation}
so that, in a local chart and denoting the degree of the various fields by 
$$
\left\{\begin{array}{cccc}0 & 1 & -1 & -2 \\
{X^\mu} & {\zeta^\alpha} & {X^\dag_\mu} & {\zeta^\dag_\alpha}\end{array}\right\}
$$ 
we have that the BV-Nambu--Goto action reads:
\begin{gather}
S_{NG}= S^{\text{cl}}_{NG} + \int_M \langle X^\dag,L_\zeta X \rangle + \langle\zeta^\dag,\frac12[\zeta,\zeta] \rangle 
\end{gather}
and $Q_M$ is the Hamiltonian vector field of $S_M$ (up to boundary terms).
\end{definition}

\begin{theorem}\label{thm:NGNOGO}
Let $(M,\partial M)$ be a two-dimensional manifold with boundary. Relaxed BV Nambu--Goto theory on $M$ is not $1$-extendable to a BV-BFV theory on $(M,\partial M)$.
\end{theorem}

\begin{proof}
We begin by computing the variation:
\begin{gather}
\delta S_{NG}= EL + \check{\alpha}_{NG}+ \int_{\partial M} X^\dag_\mu  \zeta^n \delta X^\mu - \zeta^\dag_\alpha \zeta^n \delta \zeta^\alpha \notag\\
 EL + \check{\alpha}_{NG} + \check{\alpha}_{BV}\label{28/05A10}
\end{gather}
Where $EL$ is the bulk term, and $\check{\alpha}_{NG}$ and $\check{\alpha}_{BV}$ the boundary terms due to the degree-$0$ part of the Nambu--Goto action and its BV part respectively (cf. Equation \eqref{e:NGclassoneform}). The pre-symplectic form $\check{\omega}_{BV} := \delta  \check{\alpha}_{NG} + \delta \check{\alpha}_{BV} $ is:
\begin{gather}
\check{\omega}_{BV} = \delta  \check{\alpha}_{NG}  + \int_{\partial M} \delta X^\dag_\mu  \zeta^n \delta X^\mu - X^\dag_\mu  \delta \zeta^n \delta X^\mu +\notag \\
-  \int_{\partial M} \delta \zeta^\dag_\alpha \zeta^n \delta \zeta^\alpha + \zeta^\dag_\alpha \delta \zeta^n \delta \zeta^\alpha\label{28/05A11}
\end{gather}
In this case the kernel presents a singular behaviour, differently from the case of Polyakov theory. Some equations are in fact not solvable in a general way, which makes the kernel not regular. To see this, let us write the defining equations $\iota_{\mathbb{X}}\check{\omega}_{BV}=0$, where $\mathbb{X}$ is a generic vector field on $\check{\mathcal{F}}_{NG}$:
\begin{subequations}\begin{gather} 
\delta \partial_n X^\nu\colon \qquad   g^{n n} (G_{\mu \nu} - \partial_{\alpha} X_\mu \partial^\alpha X_\nu ) (\mathbb{X}_X)^\mu    = 0\label{e:NGBVKernel1} \\
\delta X^\mu\colon \qquad  \bigg\{ \sqrt{\mathsf{g}}\, \bigg[ (g^{\lambda t} g^{n \alpha} - g^{\lambda n} g^{\alpha t} - g^{t n} g^{\alpha \lambda} ) \partial_\lambda X_\nu \partial_\alpha X_\mu + g^{n t} G_{\mu \nu} \bigg] \partial_t (\mathbb{X}_X)^\nu  + \notag\\
\sqrt{\mathsf{g}}\, \bigg[ (\cancel{g^{\lambda n} g^{n \alpha}} - \cancel{ g^{\lambda n} g^{\alpha n}} - g^{n n} g^{\alpha \lambda} ) \partial_\lambda X_\nu \partial_\alpha X_\mu + g^{n n} G_{\mu \nu} \bigg]  (\mathbb{X}_{\partial_n X})^\nu + \notag\\
\partial_t \bigg( \sqrt{\mathsf{g}}\, \bigg[ (g^{\lambda t} g^{n \alpha} - g^{\lambda n} g^{\alpha t} - g^{t n} g^{\alpha \lambda} ) \partial_\lambda X_\mu \partial_\alpha X_\nu + g^{n t} G_{\nu \mu} \bigg] (\mathbb{X}_X)^\nu  \bigg) +\notag\\
- (C_{\mu \nu} - C_{\nu \mu}) (\mathbb{X}_X)^\nu  + \frac{1}{\sqrt{\mathsf{g}}\,}\bigg( (\mathbb{X}_{X^\dag})_{\mu} \zeta^n +  X_\mu^\dag (\mathbb{X}_{\zeta})^{n} \bigg)\bigg\}   =0 \label{e:NGBVKernel2}\\
\delta X^\dag_\mu\colon \qquad \zeta^n (\mathbb{X}_{X})^{\mu}= 0 \label{e:NGBVKernel3}\\
\delta \zeta^\lambda\colon \qquad \delta_\lambda^n X^\dag_\mu (\mathbb{X}_X)^\mu  - (\mathbb{X}_{\zeta^\dag})_{\lambda} \zeta^n - \zeta_\lambda^\dag (\mathbb{X}_{\zeta})^{n} - \delta_\lambda^n \zeta^\dag_\alpha (\mathbb{X}_{\zeta})^{\alpha}   = 0 \label{e:NGBVKernel5}\\
\delta \zeta^\dag_\alpha\colon \qquad  \zeta^n (\mathbb{X}_{\zeta})^{\alpha} = 0.\label{e:NGBVKernel6}
\end{gather}\end{subequations}
Comparing with Polyakov's kernel Equations \eqref{e:PBVker}, here we can see that Equation \eqref{e:NGBVKernel1} no longer impose that $(\mathbb{X}_X)^\mu $ vanishes, and there is no kernel equation to impose $(\mathbb{X}_{\zeta})^{\lambda}\propto \zeta^n$. The equations \eqref{e:NGBVKernel3} and \eqref{e:NGBVKernel6} are then not automatically satisfied and are singular. Let us take, for instance, Equation \eqref{e:NGBVKernel3}:
\begin{gather*} 
\zeta^n (\mathbb{X}_{X})^{\mu}  = 0. 
\end{gather*}
It can be satisfied, if:
\[  
\begin{cases} (\mathbb{X}_{X})^{\mu} = 0\qquad \text{or}\ (\mathbb{X}_{X})^{\mu}\propto \zeta^n , \ {} & \zeta^n \neq 0 \\ (\mathbb{X}_{X})^{\mu} = \text{any}, & \zeta^n=0 \end{cases}
\]
These conditions are clearly not regular, meaning that the kernel is not a smooth subbundle of $T\check{\mathcal{F}}_{NG}$, thus obstructing the pre-symplectic reduction, and the definition of the smooth manifold of BFV fields $\check{\mathcal{F}}_{NG}/\mathrm{ker}(\check{\omega}_{BV}^\sharp)$. This concludes the proof.
\end{proof}

\section{Conclusions}
Both Nambu--Goto and Polyakov theories describe the motion of the bosonic string with a given background metric. They are equivalent at classical level, since they yield the same moduli space of Euler-Lagrange equations. On the other hand, the results of \cite{BRZ} hint to the fact that the two theories may no longer be related if we adopt a finer notion of equivalence. 

In this paper we compared the two theories when defined on a spacetime manifold (a worldsheet) with boundary. We first considered the classical models, computing and comparing their reduced phase spaces, following the approach introduced by Kijowski and Tulczijew \cite{KT1979}. In Theorem \ref{THM:PolyakovRPS} we gave an explicit symplectic presentation of the Reduced Phase Space of Polyakov theory for any target Lorentzian manifold, and in Theorem \ref{thm:NGRPS} we showed how the reduced phase space for Nambu--Goto theory coincides with that of Polyakov theory (after reduction).

We then analysed Polyakov and Nambu--Goto theories in the BV-BFV framework. We found an obstruction in the construction of the BV-BFV structure of the Nambu--Goto theory, which is then not $1$-extendable (Theorem \ref{thm:NGNOGO}), while no obstruction was found in the case of the Polyakov theory, for which the BV-BFV structure was derived in Theorem \ref{thm:BVBFVPolyakov}.

This result suggests that the two string theories we analysed, albeit classically equivalent, differ when a more stringent notion of equivalence is employed. Since the extendability of a BV theory to a BV-BFV theory on a manifold with boundary is a necessary requirement for quantisation with boundary, we conclude that this result suggests that Nambu--Goto theory is not a fully satisfactory description of the bosonic string.

This result strengthens the observation of \cite[Remark 7.3]{BRZ}, where the classical and quantum BV cohomologies of Nambu--Goto action have been computed, and shown to differ from the calculations of \cite{BRANDTPolyakov}, which implies that the observable content of the two models might differ.

One can take this result as a two-dimensional extension of the observations presented in \cite{CS2016a}, where the 1d analogues of Nambu--Goto theory (Jacobi theory) and of Polyakov theory (scalars coupled to 1d gravity) have been shown to be classically equivalent, and yet possess different extendability properties. Indeed, precisely like its two-dimensional NG analogue, the BV-BFV induction procedure is obstructed for Jacobi theory. Recently, the 1d case has been analysed further, to the result that despite the cohomologies of the respective Batalin--Vilkovisky--de Rham complexes are isomorphic (as predicted by \cite{BBH,Henn}), the existence of a BV-BFV pair is a stricter requirement that is not preserved by the BV-equivalence. With those considerations, our result directly shows the incompatibility of the BV and BFV structures for Nambu--Goto theory, marking a difference with Polyakov theory, regardless of the behaviour of their BV-cohomologies.

Another noteworthy scenario that presents a similar discrepancy is provided by General Relativity in dimension $d\geq 4$, where the two classically-equivalent metric and coframe formulations (Einstein--Hilbert and Eistein--Palatini--Cartan) have different extendability properties \cite{CS2016b,CS2017}, despite having equivalent reduced phase spaces \cite[Theorem 4.25]{CS2019}. 

This phenomenon appears to be linked to diffeomorphism symmetry, and suggest that certain classically-equivalent formulations of a given physical theory might be more suitable for quantisation with boundary. The next step in this program is then to proceed with the BV-BFV quantisation of Polyakov theory, following \cite{CMR2,CS2016a}. We will address this question elsewhere.

\appendix

\section{Nondegeneracy condition}\label{AppendixLightlike}

Throughout this paper we considered the case $h^{n n}\neq 0$ without justifying the choice. Now we want to show that if $h^{n n}=0$, we are considering a case where the boundary is light-like, in which case the treatment should be different from the beginning. We will use the non-rescaled metric since caluclations are easier, and rescaling does not change the arguments presented here. We look then for a diffeomorphism $\tilde{{x}}^\alpha ( {x}^\beta)$ that transforms the metric in the following way:

\begin{equation}\label{APP1}
\begin{pmatrix}
h_{nn} &  h_{nt}\\
h_{nt} & 0

\end{pmatrix} \ {}
\overset{\tilde{{x}}^\alpha ({x}^\beta )}{ \xrightarrow{\hspace*{1cm}} } \ {}
\begin{pmatrix}
0 & s\\
s & 0

\end{pmatrix}
\end{equation}

where $s$ is a function. The transformation can be written as:

\begin{equation}
h_{\alpha \beta} = \frac{\partial \tilde{{x}}^{\alpha '}}{\partial {x}^{\alpha}} \frac{\partial \tilde{{x}}^{\beta '}}{\partial {x}^\beta}
\begin{pmatrix}
0 & s\\
s & 0

\end{pmatrix}_{\alpha ' \beta '}
\end{equation}

which yields the set of equations:

\begin{equation}
  \begin{aligned}
              2 s \frac{\partial \tilde{{x}}^{t}}{\partial {x}^{t}} \frac{\partial \tilde{{x}}^{n}}{\partial {x}^t} &= 0 \\
               s ( \frac{\partial \tilde{{x}}^{t}}{\partial {x}^{t}} \frac{\partial \tilde{{x}}^{n}}{\partial {x}^n} + \frac{\partial \tilde{{x}}^{t}}{\partial {x}^{n}} \frac{\partial \tilde{{x}}^{n}}{\partial {x}^t}) &= h_{n t}\\
               2s \frac{\partial \tilde{{x}}^{t'}}{\partial {x}^{n}} \frac{\partial \tilde{{x}}^{n}}{\partial {x}^n} &= h_{n n}
  \end{aligned}
\end{equation}

From the first equation, either $\frac{\partial \tilde{{x}}^{t}}{\partial {x}^{t}}$ or $ \frac{\partial \tilde{{x}}^{n}}{\partial {x}^t}$ must vanish. Let us consider the case where $ \frac{\partial \tilde{{x}}^{n}}{\partial {x}^t}=0$, and the other case is analogous. We then have that
\begin{equation}
\tilde{{x}}^n := \tilde{{x}}^n ({x}^n)
\end{equation}
and:
\begin{equation}
d \tilde{{x}}^t = \frac{h_{n t}}{s \dot{\tilde{{x}}}^n} d{x}^t + \frac{h_{n n}}{2 s \dot{\tilde{{x}}}^n} d{x}^n =: \frac{1}{s} \omega.
\end{equation}

For any one-form $\omega'$ in 2 dimensions, there always exists a function $\lambda$ such that $\lambda \omega '$ is locally an exact form, thus locally it always exists a function $s$ such that $d \tilde{{x}}^t$ is an exact form, and there is always a diffeomorphism $ \tilde{{x}}: ({x}^n, {x}^t) \rightarrow (\tilde{{x}}^n ({x}^n), \tilde{{x}}^t ({x}^n, {x}^t) )$ that satisfies \eqref{APP1} (locally). This is due to the Frobenius theorem for foliations applied to the case of a 1-form in 2 dimensions.

\section{Manipulation of constraints}\label{A:contraintmanipulation}
In this section we want to find the explicit expression of the functions $\tau^{\alpha\beta}({\th})$, introduced in Part 3 of the Proof of Theorem \ref{THM:PolyakovRPS}, such that the resulting combination $\tau^{\alpha\beta}f_{\alpha\beta}$ is a basic function in $\check{F}$, i.e.\ it only depends on the reduced variables $(X^\mu, J_\mu)$.

In order to do this, let us rewrite ${\th}^{\alpha \beta} (\partial_\alpha X^\nu)G_{\mu\nu} (\partial_\beta X^\mu) $ as a function of the reduced variables
\begin{equation}\label{A19}
\begin{aligned}
{\th}^{\alpha \beta} (\partial_\alpha X^\nu) G_{\mu\nu} (\partial_\beta X^\mu) 
 & =  {\th}^{t t} (\partial_t X^\nu) G_{\mu\nu} (\partial_t X^\mu) + 2 {\th}^{n t} (\partial_t X^\nu) G_{\mu\nu}  ({\th}^{n n} )^{-1} (  J^\mu - (\partial_t X^\mu) {\th}^{n t} )\\
& + {\th}^{n n} ({\th}^{n n})^{-2} (  J_\mu - (\partial_t X^\nu) G_{\mu\nu} {\th}^{n t} ) (  J^\mu - (\partial_t X^\mu) {\th}^{n t} )\\&\\
& =(\partial_t X^\nu) G_{\mu\nu} (\partial_t X^\mu) \bigg({\th}^{t t} - 2 \frac{({\th}^{n t})^{2}}{{\th}^{n n}}+ \frac{({\th}^{n t})^{2}}{{\th}^{n n}} \bigg) \\
& +\chi (\partial_t X^\nu) G_{\mu\nu} J^\mu \bigg( 2  \frac{{\th}^{n t}}{{\th}^{n n}} - 2  \frac{{\th}^{n t}}{{\th}^{n n}} \bigg) + J_\mu J^\mu ({\th}^{n n})^{-1} \\&\\
& = ({\th}^{n n})^{-1} (\varsigma(\partial_t X^\nu) G_{\mu\nu} (\partial_t X^\mu) + J_\mu J^\mu )
\end{aligned}
\end{equation}\ {}
where we used: $\partial_n X^\mu = ({\th}^{n n})^{-1} ( \chi J^\mu - {\th}^{n t} (\partial_t X^\mu) )$ and that $({\th}^{t t} - 2 \frac{({\th}^{n t})^{2}}{{\th}^{n n}}+ \frac{({\th}^{n t})^{2}}{{\th}^{n n}}) = (\th^{nn})^{-1} \mathrm{det}(\th)=(\th^{nn})^{-1}\varsigma $. Thus:
\begin{equation}\label{A20}
\begin{aligned}
\tau^{\alpha \beta} {\th}_{\alpha \beta} \frac{ {\th}^{\lambda \rho}}{2} (\partial_\lambda X^\nu) G_{\mu\nu} (\partial_\rho X^\mu)&=( \tau^{t t} {\th}_{t t} + 2 \tau^{n t}{\th}_{n t} +\tau^{n n} {\th}_{n n} )\frac{1}{2 {\th}^{n n}} (\varsigma (\partial_t X^\nu) G_{\mu\nu} (\partial_t X^\mu) + J_\mu J^\mu )\\
& =(\tau^{t t} \frac{{\th}_{t t}}{2 {\th}^{n n}} +  \tau^{n t} \frac{{\th}_{n t}}{{\th}^{n n}} +\tau^{n n} \frac{{\th}_{n n}}{2 {\th}^{n n}} ) (\varsigma (\partial_t X^\nu) G_{\mu\nu} (\partial_t X^\mu) + J_\mu J^\mu )\\
&=(  \tau^{t t} \frac{1}{2} -  \tau^{n t} \frac{{\th}^{n t}}{{\th}^{n n}} +\tau^{n n} \frac{{\th}^{t t}}{2 {\th}^{n n}} ) ((\partial_t X^\nu) G_{\mu\nu} (\partial_t X^\mu) + \varsigma J_\mu J^\mu )
\end{aligned}
\end{equation}
Where we used the metric $\th$ to raise the indices two write the last equality, to the effect of moving $\varsigma$ in front of $J^\mu J_\mu$.
Now we compute the second term:

\begin{equation}\label{A21}
\begin{aligned}
\tau^{\alpha \beta} (\partial_\alpha X^\nu) G_{\mu\nu} (\partial_\beta X^\mu) & =\tau^{t t} (\partial_t X^\nu) G_{\mu\nu} (\partial_t X^\mu) + 2 \tau^{n t} (\partial_t X^\nu) G_{\mu\nu} ({\th}^{n n})^{-1} ( \chi J^\mu - {\th}^{n t} (\partial_t X^\mu) )\\
&+\tau^{n n} ({\th}^{n n})^{-2} ( \chi J_\mu - {\th}^{n t} (\partial_t X^\nu) G_{\mu\nu} ) ( \chi J^\mu - {\th}^{n t} (\partial_t X^\mu) ) \\&\\
&=(\partial_t X^\nu) G_{\mu\nu} (\partial_t X^\mu) \bigg(\tau^{t t} - 2 \tau^{n t} \frac{{\th}^{n t}}{{\th}^{n n}} + \tau^{n n} \frac{({\th}^{n t})^2}{({\th}^{n n})^2}\bigg)  \\
&+ \chi (\partial_t X^\nu) G_{\mu\nu} J^\mu \bigg(  2 \tau ^{n t} \frac{1}{{\th}^{n n}} -  2 \tau^{n n} \frac{{\th}^{n t}}{({\th}^{n n})^2} \bigg) \\
&+ J_\mu J^\mu  \tau^{n n} ({\th}^{n n})^{-2}.
\end{aligned}
\end{equation}
Putting the pieces together:
\begin{equation}\label{A22}
\begin{aligned}
&\tau^{\alpha \beta} f_{\alpha \beta} = \tau^{\alpha \beta} {\th}_{\alpha \beta} \frac{1}{2} {\th}^{\lambda \tau} (\partial_\lambda X^\nu) G_{\mu\nu} (\partial_\tau X^\mu) - \tau^{\alpha \beta} (\partial_\alpha X^\nu) G_{\mu\nu} (\partial_\beta X^\mu) =\\&\\
&(\partial_t X^\nu) G_{\mu\nu} (\partial_t X^\mu) \bigg( \tau^{t t} \frac{1}{2} -  \tau^{n t} \frac{{\th}^{n t}}{{\th}^{n n}} +\tau^{n n} \frac{{\th}^{t t}}{2 {\th}^{n n}} -(\tau^{t t} - 2 \tau^{n t} \frac{{\th}^{n t}}{{\th}^{n n}} + \tau^{n n} \frac{({\th}^{n t})^2}{({\th}^{n n})^2})\bigg) + \\
&\chi (\partial_t X^\nu) G_{\mu\nu} J^\mu \bigg( -  2 \tau ^{n t} \frac{1}{{\th}^{n n}} +  2 \tau^{n n} \frac{{\th}^{n t}}{({\th}^{n n})^2} \bigg) +\\
&\varsigma J_\mu J^\mu  \bigg(  \tau^{t t} \frac{1}{2} -  \tau^{n t} \frac{{\th}^{n t}}{{\th}^{n n}} +\tau^{n n} \frac{{\th}^{t t}}{2 {\th}^{n n}} - \varsigma \tau^{n n} ({\th}^{n n})^{-2} \bigg) = \\&\\
&(\partial_t X^\nu) G_{\mu\nu} (\partial_t X^\mu) \bigg[ - \tau^{t t} \frac{1}{2} +  \tau^{n t} \frac{{\th}^{n t}}{{\th}^{n n}} +\tau^{n n} \bigg(\frac{{\th}^{t t}}{2 {\th}^{n n}} - \frac{({\th}^{n t})^2}{({\th}^{n n})^2}\bigg) \bigg] + \\
& \chi (\partial_t X^\nu) G_{\mu\nu} J^\mu \bigg[ -  2 \tau ^{n t} \frac{1}{{\th}^{n n}} +  2 \tau^{n n} \frac{{\th}^{n t}}{({\th}^{n n})^2} \bigg] +\\
&\varsigma J_\mu J^\mu  \bigg[  \tau^{t t} \frac{1}{2} -  \tau^{n t} \frac{{\th}^{n t}}{{\th}^{n n}} +\tau^{n n} \bigg( \frac{{\th}^{t t}}{2 {\th}^{n n}} -\varsigma\frac{1}{({\th}^{n n})^{2}}\bigg) \bigg].
\end{aligned}
\end{equation}
This expresses $\tau^{\alpha \beta} f_{\alpha \beta}$ in function of $({\th}_{\alpha \beta}, X^\mu, J_\mu)$.  We now compare it with Equation \eqref{A18} for $l^{\alpha \beta}$, which we report below for convenience:
$$
\tau^{\alpha \beta} ({\th}) f_{\alpha \beta} ({\th}, X^\mu, J_\mu) = l^{n n} J_\mu J^\mu + 2 l^{n t} (\partial_t X^\mu) J_\mu + l^{t t} (\partial_t X^\nu) G_{\mu\nu} (\partial_t X^\mu)
$$
so that:
\begin{equation}\label{21/04/A3}
\begin{aligned}
& l^{nt}= \chi\bigg( - \tau ^{n t} \frac{1}{{\th}^{n n}} + \tau^{n n} \frac{{\th}^{n t}}{({\th}^{n n})^2} \bigg)\\
&l^{t t} = - \tau^{t t} \frac{1}{2} +  \tau^{n t} \frac{{\th}^{n t}}{{\th}^{n n}} +\tau^{n n} \bigg(\frac{{\th}^{t t}}{2 {\th}^{n n}} - \frac{({\th}^{n t})^2}{({\th}^{n n})^2}\bigg)\\
& \varsigma l^{n n}= \tau^{t t} \frac{1}{2} -  \tau^{n t} \frac{{\th}^{n t}}{{\th}^{n n}} + \tau^{n n} \bigg( \frac{{\th}^{t t}}{2 {\th}^{n n}} - \varsigma \frac{1}{({\th}^{n n})^{2}}\bigg)
\end{aligned}
\end{equation}
The first two equations can be solved for $\tau^{n t}$ and $\tau^{t t}$. At this point, the third one yields a relation between $l_{n n}$ and $l_{t t}$: $l^{n n} = -\varsigma l^{ t t}$. The relations are:
\begin{equation}\label{A24}
\begin{aligned}
&\tau^{n n} = \tau^{n n} \frac{{\th}^{n n}}{{\th}^{n n}}\\
&\tau^{n t} = \tau^{n n} \frac{{\th}^{n t}}{{\th}^{n n}} -  \chi {\th}^{n n} l^{n t} \\
&\tau^{t t} =  \tau^{n n}\frac{{\th}^{t t}}{{\th}^{n n}} - 2 \chi  {\th}^{n t}  l^{n t} -2 l^{t t}
\end{aligned}
\end{equation}
The first term after the equal sign can be written as: $\frac{\tau^{n n}}{{\th}^{n n}} {\th}^{\alpha \beta} \propto  {\th}^{\alpha \beta}$. This term gives no contribution when contracted with $f_{\alpha \beta}$, since ${\th}^{\alpha \beta} f_{\alpha \beta} = 0$. \\
The relation $l^{n n} = -\varsigma l^{ t t}$ tells us that we can only obtain terms of this kind: $ l^{t t} ( (\partial_t X^\nu) G_{\mu\nu} (\partial_t X^\mu) - \varsigma J_\mu J^\mu ) + l^{n t} (\partial_t X^\nu) G_{\mu\nu} J^\mu$.  The resulting new constraints for Polyakov theory are given in Equation \eqref{eq:PolConstraints}.

\section{Nambu--Goto: lengthy algebraic manipulations}\label{AppendixLengthy}

In this appendix, we are going to show how eq. \eqref{A34} follows from eq. \eqref{A33} and \eqref{A35}. We rewrite here the second condition of the kernel \eqref{A36}:

\begin{equation}\label{e:NGKerRearranged}
\begin{aligned}
\delta X^\mu\colon \qquad \bigg\{ \sqrt{\mathsf{g}}\, \bigg[ (g^{\lambda t} g^{n \alpha} - g^{\lambda n} g^{\alpha t} - g&^{t n} g^{\alpha \lambda} ) \partial_\lambda X_\nu \partial_\alpha X_\mu + g^{n t} G_{\mu \nu} \bigg] \partial_t (\mathbb{X}_X)^\nu  + \\
\sqrt{\mathsf{g}}\, \bigg[ (\cancel{g^{\lambda n} g^{n \alpha}} - \cancel{ g^{\lambda n} g^{\alpha n}} - &g^{n n} g^{\alpha \lambda} ) \partial_\lambda X_\nu \partial_\alpha X_\mu + g^{n n} G_{\mu \nu} \bigg]  (\mathbb{X}_{\partial_n X})^\nu  + \\
\partial_t \bigg( \sqrt{\mathsf{g}}\, \bigg[ (g^{\lambda t} g^{n \alpha} - g^{\lambda n} g^{\alpha t} - g^{t n} &g^{\alpha \lambda} ) \partial_\lambda X_\mu \partial_\alpha X_\nu + g^{n t} G_{\nu \mu} \bigg] (\mathbb{X}_X)^\nu  \bigg) \bigg\} =\\
\delta X^\mu \sqrt{\mathsf{g}}\,(C&_{\mu \nu} - C_{\nu \mu}) (\mathbb{X}_X)^\nu 
\end{aligned}
\end{equation}
where $C_{\mu \nu}$ groups together the terms inside the square bracket in Equation \eqref{A33}.

We will see that  \eqref{e:NGKerRearranged} has no component on the plane tangent  to the surface defined by $X^\mu$ and thus all conditions are on the orthogonal part. This is in agreement with eq. \eqref{A34}, where the tangent part is free ($\beta_n \partial_n X^\mu + \beta_t \partial_t X^\mu$), while the orthogonal part is not. This will allow us to easily find a condition on $P^{\perp \mu}_{\nu}  (\mathbb{X}_{\partial_n X})^\nu $. In fact this term is exactly the second line of \eqref{e:NGKerRearranged}:

\begin{gather*}
g^{n n} P^{\perp \mu}_{\nu}  (\mathbb{X}_{\partial_n X})^\nu  = \bigg[ - g^{n n} g^{\alpha \lambda}  \partial_\lambda X_\nu \partial_\alpha X_\mu + g^{n n} G_{\mu \nu} \bigg]  (\mathbb{X}_{\partial_n X})^\nu 
\end{gather*}

The problem is thus reduced to writing the orthogonal projection of the remaining terms. Rearranging the second to last line of \eqref{e:NGKerRearranged} to have the $\mu$ and $\nu$ indexes in the same order as the previous lines, using $ P^\perp_{\mu \nu} : = G_{\mu \nu} - \partial^{\alpha} X_\mu \partial_\alpha X_\nu$ and $g^{\alpha \beta} \partial_\beta X_\mu =: \partial^\alpha X_\mu$, and recalling that $(\mathbb{X}_X)^\mu = \alpha_n \partial_n X^\mu + \alpha_t \partial_t X^\mu $, the left hand side of \eqref{e:NGKerRearranged} becomes:

\begin{equation}\label{A37}
\begin{aligned}
LHS = \bigg[ (\partial^t X_\nu \partial^n X_\mu - \partial^n X_\nu \partial^t X_\mu ) - g^{n t} P^\perp_{\mu \nu} \bigg] \partial_t (\mathbb{X}_X)^\nu  \sqrt{\mathsf{g}}\, - g^{n n} P^\perp_{\mu \nu} (\mathbb{X}_{\partial_n X})^\nu  \sqrt{\mathsf{g}}\, +\\
-\partial_t \bigg[ \sqrt{\mathsf{g}}\, (\partial^t X_\nu \partial^n X_\mu - \partial^n X_\nu \partial^t X_\mu) (\mathbb{X}_X)^\nu  \bigg] - \partial_t \bigg[\cancel{ g^{n t} P^\perp_{\mu \nu} (\mathbb{X}_X)^\nu  \sqrt{\mathsf{g}}\,} \bigg] =  \\
- \sqrt{\mathsf{g}}\, P^\perp_{\mu \nu} \bigg[ g^{n t} \partial_t (\mathbb{X}_X)^\nu  + g^{n n} (\mathbb{X}_{\partial_n X})^\nu  \bigg] - \partial_t \bigg[\sqrt{\mathsf{g}}\, (\partial^t X_\nu \partial^n X_\mu - \partial^n X_\nu \partial^t X_\mu) \bigg] (\mathbb{X}_X)^\nu  
\end{aligned}
\end{equation}
We want to show now that the second term in \eqref{A37} has the component on the tangent plane:

\begin{equation}\label{A38}
\begin{aligned}
P^{\parallel , \rho \mu} \partial_t \bigg[\sqrt{\mathsf{g}}\, (\partial^t X_\nu \partial^n X_\mu - \partial^n X_\nu \partial^t X_\mu) \bigg] (\mathbb{X}_X)^\nu  =\\
\frac{\sqrt{\mathsf{g}}\,}{2}(\alpha_t  \partial^n X^\rho - \alpha_n \partial^t X^\rho ) ( \partial_\lambda X^\mu \partial^\lambda X^\nu) (\partial_t G_{\mu \nu}),
\end{aligned}
\end{equation}
where $P^{\parallel , \rho \mu} := \partial^\alpha X^\rho \partial_\alpha X^\mu$. 
The term on the right hand side of \eqref{A38} vanishes in the case of a constant metric. We will see later that in the general case it gets canceled by the term on the right hand side of \eqref{e:NGKerRearranged}. This means that all tangent terms vanish, as we wanted to prove. 
Rearranging indexes and since from \eqref{A35}, $(\mathbb{X}_X)^\mu $ lies on the tangent plane, we can rewrite the left hand side of \eqref{A38} as:

\begin{gather}\label{A39}
P^{\parallel , \mu k} \partial_t \bigg[\sqrt{\mathsf{g}}\, (\partial^t X_k \partial^n X_w - \partial^n X_k \partial^t X_w)\bigg] P^{\parallel , w}_{\nu} (\mathbb{X}_X)^\nu 
\end{gather}
Let us then look at:

\begin{gather}
P^{\parallel , \mu k} \partial_t \bigg[\sqrt{\mathsf{g}}\, (\partial^t X_k \partial^n X_w - \partial^n X_k \partial^t X_w)\bigg] P^{\parallel , w}_{\nu}
\end{gather}
Let us calculate the internal term:

\begin{gather*}
\partial_t \bigg[\sqrt{\mathsf{g}}\, (\partial^t X_k \partial^n X_w - \partial^n X_k \partial^t X_w)\bigg] =\notag\\
\partial_t \bigg[\sqrt{\mathsf{g}}\, \bigg] (\partial^t X_k \partial^n X_w - \partial^n X_k \partial^t X_w) + (\sqrt{\mathsf{g}}\,) \partial_t \bigg[ (\partial^t X_k \partial^n X_w - \partial^n X_k \partial^t X_w)\bigg]=
\end{gather*}
\begin{gather*}
\sqrt{\mathsf{g}}\, g^{\rho \sigma} ( \partial_\rho X_\mu \partial_t \partial_\sigma X^\mu + \underline{ \frac{1}{2} \partial_\rho X^\mu \partial_\sigma X^\nu ( \partial_t G_{\mu \nu}) }) \ {} (\partial^t X_k \partial^n X_w - \partial^n X_k \partial^t X_w) +\notag\\
 \sqrt{\mathsf{g}}\, \bigg[ \partial_t \partial^t X_k \partial^n X_w + \partial^t X_k \partial_t \partial^n X_w - \partial_t \partial^n X_k \partial^t X_w - \partial^n X_k \partial_t \partial^t X_w \bigg] \propto
\end{gather*}
\begin{gather}
( \partial^\sigma X_\mu \partial_t \partial_\sigma X^\mu +  \underline{ \frac{1}{2} \partial^\sigma X^\mu \partial_\sigma X^\nu ( \partial_t G_{\mu \nu}) }) \ {} (\partial^t X_k \partial^n X_w - \partial^n X_k \partial^t X_w) +\notag\\
 \partial_t \partial^t X_k \partial^n X_w + \partial^t X_k \partial_t \partial^n X_w - \partial_t \partial^n X_k \partial^t X_w - \partial^n X_k \partial_t \partial^t X_w \label{A40}
\end{gather}
where in the last step the expression was multiplied by $(\sqrt{\mathsf{g}}\,)^{-1}$. The underlined term gives the term on the right hand side of \eqref{A38}, while the projection of the rest on the tangent plane vanishes. Removing the underlined part and applying the projectors:
\begin{gather*}
 P^{\parallel , \mu k} \bigg[ \bigg(\partial_t - \underline{ \frac{1}{2} \partial^\sigma X^\mu \partial_\sigma X^\nu ( \partial_t G_{\mu \nu}) }\bigg)  \bigg(\sqrt{\mathsf{g}}\, (\partial^t X_k \partial^n X_w - \partial^n X_k \partial^t X_w)\bigg) \bigg] P^{\parallel , w}_{\nu} =
\end{gather*}
\begin{gather*}
 \partial^\sigma X_k \partial_t \partial_\sigma X^k \ {} (\partial^t X^\mu \partial^n X_\nu - \partial^n X^\mu \partial^t X_\nu) +\\
+\partial^\alpha X^\mu \partial_\alpha X^k \partial_t \partial^t X_k \partial^n X_\nu + \partial^t X^\mu \partial_t \partial^n X_k \partial_\beta X^k \partial^\beta X_\nu +\\
-\partial^\alpha X^\mu \partial_\alpha X^k \partial_t \partial^n X_k \partial^t X_\nu - \partial^n X^\mu \partial_t \partial^t X_k \partial_\beta X^k \partial^\beta X_\nu=
\end{gather*}
\begin{gather}
 \partial^\sigma X_k \partial_t \partial_\sigma X^k \ {} (\partial^t X^\mu \partial^n X_\nu - \partial^n X^\mu \partial^t X_\nu) +\notag\\
\underbrace{\partial^n X^\mu \partial_n X^k \partial_t \partial^t X_k \partial^n X_\nu}_{A} + \partial^t X^\mu \partial_t X^k \partial_t \partial^t X_k \partial^n X_\nu+\notag\\
+\partial^t X^\mu \partial_t \partial^n X_k \partial_n X^k \partial^n X_\nu + \underbrace{\partial^t X^\mu \partial_t \partial^n X_k \partial_t X^k \partial^t X_\nu}_B+\notag\\
- \partial^n X^\mu \partial_n X^k \partial_t \partial^n X_k \partial^t X_\nu - \underbrace{\partial^t X^\mu \partial_t X^k \partial_t \partial^n X_k \partial^t X_\nu}_B + \notag\\
- \underbrace{\partial^n X^\mu \partial_t \partial^t X_k \partial_n X^k \partial^n X_\nu}_{A} - \partial^n X^\mu \partial_t \partial^t X_k \partial_t X^k \partial^t X_\nu \label{A41}
\end{gather}
The terms marked $A$ and $B$ cancel each other. We thus remain with:
\begin{gather}
 \partial^\sigma X_k \partial_t \partial_\sigma X^k \ {} (\partial^t X^\mu \partial^n X_\nu - \partial^n X^\mu \partial^t X_\nu) +\notag\\
\partial^t X^\mu \partial_t X^k \partial_t \partial^t X_k \partial^n X_\nu + \partial^t X^\mu \partial_t \partial^n X_k \partial_n X^k \partial^n X_\nu \label{A42}\\
- \partial^n X^\mu \partial_n X^k \partial_t \partial^n X_k \partial^t X_\nu - \partial^n X^\mu \partial_t \partial^t X_k \partial_t X^k \partial^t X_\nu\notag
\end{gather}
We can rearrange the terms in the second and third line and then use the Leibniz rule:

\begin{gather*}
 \partial^\sigma X_k \partial_t \partial_\sigma X^k \ {} (\partial^t X^\mu \partial^n X_\nu - \partial^n X^\mu \partial^t X_\nu) +\\
\partial^t X^\mu ( \partial_\sigma X^k \partial_t \partial^\sigma X_k ) \partial^n X_\nu - \partial^n X^\mu ( \partial_t \partial^\sigma X_k \partial_\sigma X^k ) \partial^t X_\nu = \\
 \partial^\sigma X_k \partial_t \partial_\sigma X^k \ {} (\partial^t X^\mu \partial^n X_\nu - \partial^n X^\mu \partial^t X_\nu) +\\
+  \partial_t \partial^\sigma X_k \partial_\sigma X^k \ {} (\partial^t X^\mu \partial^n X_\nu - \partial^n X^\mu \partial^t X_\nu) = \\
 \partial_t ( \partial^\sigma X_k  \partial_\sigma X^k ) \ {} (\partial^t X^\mu \partial^n X_\nu - \partial^n X^\mu \partial^t X_\nu) = \\
 \partial_t ( \delta^\sigma_{\sigma}) \ {} (\partial^t X^\mu \partial^n X_\nu - \partial^n X^\mu \partial^t X_\nu) \equiv 0
\end{gather*}
This means that the left hand side of \eqref{A36}, which we rewrote in \eqref{A37} as:

\begin{equation}\label{A43}
LHS = - \sqrt{\mathsf{g}}\, P^\perp_{\mu \nu} \bigg[ g^{n t} \partial_t (\mathbb{X}_X)^\nu  + g^{n n} (\mathbb{X}_{\partial_n X})^\nu  \bigg] - \partial_t \bigg[\sqrt{\mathsf{g}}\, (\partial^t X_\nu \partial^n X_\mu - \partial^n X_\nu \partial^t X_\mu) \bigg] (\mathbb{X}_X)^\nu  
\end{equation}
can be rewritten as:
\begin{gather}
LHS = - \sqrt{\mathsf{g}}\, P^{\perp, \mu}_{\nu} \bigg[ g^{n t} \partial_t (\mathbb{X}_X)^\nu  + g^{n n} (\mathbb{X}_{\partial_n X})^\nu  \bigg] + \notag\\
- P^{\perp, \mu k} \partial_t \bigg[\sqrt{\mathsf{g}}\, (\partial^t X_\nu \partial^n X_k - \partial^n X_\nu \partial^t X_k) \bigg] (\mathbb{X}_X)^\nu + \notag\\
- \frac{\sqrt{\mathsf{g}}\,}{2}(\alpha_t  \partial^n X^\rho - \alpha_n \partial^t X^\rho ) ( \partial_\lambda X^\mu \partial^\lambda X^\nu) (\partial_t G_{\mu \nu}) \label{A44}
\end{gather}
From \eqref{A40} we see that the second term simplifies considerably:

\begin{gather}
-P^{\perp, \mu k} \partial_t \bigg[\sqrt{\mathsf{g}}\, (\partial^t X_\nu \partial^n X_k - \partial^n X_\nu \partial^t X_k) \bigg] (\mathbb{X}_X)^\nu  =\notag
\end{gather}
\begin{gather}
\sqrt{\mathsf{g}}P^{\perp, \mu k}  \bigg[ \partial_t \partial^t X_k \partial^n X_\nu - \partial_t \partial^n X_k \partial^t X_\nu \bigg] (\mathbb{X}_X)^\nu =\notag
\end{gather}
\begin{gather}
\sqrt{\mathsf{g}}P^{\perp, \mu k}  \bigg[ \partial_t \partial^t X_k \partial^n X_\nu - \partial_t \partial^n X_k \partial^t X_\nu \bigg] ( \alpha_n \partial_n X^\nu + \alpha_t \partial_t X^\nu)=\notag
\end{gather}
\begin{gather}
\sqrt{\mathsf{g}}P^{\perp, \mu k}  \bigg[ \alpha_n  \partial_t \partial^t X_k - \alpha_t \partial_t \partial^n X_k \bigg]=\notag
\end{gather}
\begin{gather}
\sqrt{\mathsf{g}}P^{\perp, \mu k}  \bigg[ \alpha_n  \partial_t (g^{t \alpha}\partial_\alpha X_k) - \alpha_t \partial_t ( g^{n \alpha}\partial_\alpha X_k) \bigg]=\notag
\end{gather}
\begin{gather}
\sqrt{\mathsf{g}}P^{\perp, \mu k} \bigg[ \alpha_n  g^{t \alpha} \partial_t \partial_\alpha X_k - \alpha_t g^{n \alpha} \partial_t \partial_\alpha X_k \bigg]=\notag
\end{gather}
\begin{gather}
\sqrt{\mathsf{g}}P^{\perp, \mu k}  \bigg[ (g^{t t} \alpha_n - g^{n t} \alpha_t ) \partial_t \partial_t X_k +  (g^{n t} \alpha_n - g^{n n}\alpha_t )\partial_t \partial_n X_k \bigg]=\notag
\end{gather}
\begin{gather}
\sqrt{\mathsf{g}}P^{\perp, \mu}_{\ {} k}  \bigg[ (g^{t t} \alpha_n - g^{n t} \alpha_t ) \partial_t \partial_t X^k +  (g^{n t} \alpha_n - g^{n n}\alpha_t )\partial_t \partial_n X^k \bigg]+\notag\\
\sqrt{\mathsf{g}}P^{\perp, \mu}_{\ {} k'} G^{k' w} \dot{G}_{w k} \bigg[ (g^{t t} \alpha_n - g^{n t} \alpha_t ) \partial_t X^k +  (g^{n t} \alpha_n - g^{n n}\alpha_t )\partial_n X^k \bigg]\label{A45}
\end{gather}
We can also simplify the first part of the first term in \eqref{A44}:

\begin{equation}\label{A46}
\begin{aligned}
-P^{\perp, \mu}_{\nu} \partial_t (\mathbb{X}_X)^\nu  =-P^{\perp, \mu}_{\nu} g^{n t} &\partial_t ( \alpha_n \partial_n X^\nu + \alpha_t \partial_t X^\nu)=\\
-P^{\perp, \mu}_{\nu} g^{n t}  ( \alpha_n \partial_t \partial_n X&^\nu + \alpha_t \partial_t \partial_t X^\nu)
\end{aligned}
\end{equation}
Putting the pieces together \eqref{A44} becomes:

\begin{gather}
- \sqrt{\mathsf{g}}\, g^{n n} P^{\perp, \mu}_{\nu} (\mathbb{X}_{\partial_n X})^\nu  + \sqrt{\mathsf{g}}\,  P^{\perp, \mu}_{\nu} \bigg[ (g^{t t} \alpha_n - 2 g^{n t} \alpha_t ) \partial_t \partial_t X^\nu - g^{n n}\alpha_t \partial_t \partial_n X^\nu \bigg] +\notag\\
\sqrt{\mathsf{g}}\, P^{\perp, \mu}_{\nu} G^{\nu w} \dot{G}_{w k} \bigg[ (g^{t t} \alpha_n - g^{n t} \alpha_t ) \partial_t X^k +  (g^{n t} \alpha_n - g^{n n}\alpha_t )\partial_n X^k \bigg] +\notag \\
- \frac{\sqrt{\mathsf{g}}\,}{2}(\alpha_t  \partial^n X^\rho - \alpha_n \partial^t X^\rho ) ( \partial_\lambda X^\mu \partial^\lambda X^\nu) (\partial_t G_{\mu \nu}) \label{A47}
\end{gather}
and if the metric is constant the right hand side of \eqref{A36}  vanishes (which means that \eqref{A47}$=0$ and that its third row  identically vanishes) and we can write the kernel fields as:

\begin{gather}
(\mathbb{X}_X)^\mu  = \alpha_n \partial_n X^\mu + \alpha_t \partial_t X^\mu \label{A48oneBis}\notag \\
(\mathbb{X}_{\partial_n X})^\mu  = \beta_n \partial_n X^\mu + \beta_t \partial_t X^\mu +(g^{n n})^{-1} P^{\perp, \mu}_{\nu}  \bigg[   (g^{t t} \alpha_n - 2 g^{n t} \alpha_t ) \partial_t \partial_t X^\nu - g^{n n}\alpha_t \partial_t \partial_n X^\nu \bigg]  \notag\\
+ (g^{n n})^{-1} P^{\perp, \mu}_{\nu} G^{\nu w} \dot{G}_{w k} \bigg[ (g^{t t} \alpha_n - g^{n t} \alpha_t ) \partial_t X^k +  (g^{n t} \alpha_n - g^{n n}\alpha_t )\partial_n X^k \bigg] \label{A48twoBis}
\end{gather}
with $\{ \alpha_n, \alpha_t, \beta_n, \beta_t \}$ free parameters. The presence of $\beta_n$ and $\beta_t$ is due to the fact that there is no condition on the part of $(\mathbb X)_{\partial_n X^\mu}$ annihilated by $P^{\perp, \mu}_{\nu}$.

In the case of a general metric, $C_{\mu \nu}$ doesn't vanish, but we have \eqref{A47} $= \sqrt{\mathsf{g}}\, (C_{\mu \sigma} - C_{\sigma \mu})(\mathbb X)_{X^\sigma} $. If the right hand side of \eqref{e:NGKerRearranged} had a component parallel to the tangent plane that did not exactly cancel the third line of \eqref{A47}, there would be conditions over $\alpha_n$ and $\alpha_t$, which would mean that the kernel of $\omega$ could have different dimensions for different target metrics $G_{\mu \nu}$. We are going now to calculate the component of $(C_{\mu \sigma} - C_{\sigma \mu}) (\mathbb X)_{X^\sigma}$ parallel to the tangent plane. Let us observe that $C_{\sigma \mu} \partial_\beta X^\mu$ can be rewritten in a simplified way:

\begin{gather}
\bigg[ \frac{1}{2}\bigg(g^{\lambda \rho} \partial^n X^{\mu '} - g^{\lambda n} \partial^\rho X^{\mu '} - g^{\rho n} \partial^\lambda X^{\mu '} \bigg) \partial_\lambda X^{\nu '} \partial_\alpha X^\nu \frac{\partial G_{\nu ' \nu}}{\partial X^\sigma} G_{\mu ' \mu } \notag\\
+ \partial^n X^{\mu '} \frac{\partial G_{\mu ' \mu}}{\partial X^\sigma} \bigg] \partial_\beta X^\mu=\notag\\
 \frac{1}{2}\bigg(\partial_\lambda X^{\nu '} \partial^\lambda X^\nu \delta^{n}_\beta - 2 \partial^n X^{\nu '} \partial_\beta X^\nu \bigg)  \frac{\partial G_{\nu ' \nu}}{\partial X^\sigma} + \partial^n X^{\mu '}\partial_\beta X^\mu \frac{\partial G_{\mu ' \mu}}{\partial X^\sigma}=\notag\\
\frac{\delta^{n}_{\ {}\beta}}{2}\bigg( \partial_\lambda X^{\nu '} \partial^\lambda X^\nu \bigg)  \frac{\partial G_{\nu ' \nu}}{\partial X^\sigma} \label{A48B1}
\end{gather}
The component of $ (C_{\mu \sigma} - C_{\sigma \mu}) (\mathbb X)_{X^\sigma} $ parallel to the tangent plane is:

\begin{gather}\label{A48B2}
 P^{\parallel, \mu}_{\nu} \bigg(C_{\mu \sigma} - C_{\sigma \mu}\bigg) (\mathbb X)_{X^\sigma} 
\end{gather}
condition that can be developed using $P^{\parallel, \mu}_{\nu}:= \partial^\lambda X_\nu \partial_\lambda X^\mu$:

\begin{gather}
\partial^\lambda X_\nu \partial_\lambda X^\mu \bigg(C_{\mu \sigma} - C_{\sigma \mu}\bigg) \bigg(\alpha_n \partial_n X^\sigma + \alpha_t \partial_t X^\sigma\bigg) =\notag\\
\bigg[(\partial^\lambda X_\nu \partial_\lambda X^\sigma) \alpha_n - \partial^n X_\nu  (\alpha_n \partial_n X^\sigma + \alpha_t \partial_t X^\sigma) \bigg] \frac{1}{2}\partial_\rho X^{\nu '} \partial^\rho X^\nu \frac{\partial G_{\nu ' \nu}}{\partial X^\sigma}=\notag\\
\bigg[\alpha_n \partial^t X_\nu \partial_t X^\sigma - \alpha_t \partial^n X_\nu \partial_t X^\sigma \bigg]  \frac{1}{2}\partial_\rho X^{\nu '} \partial^\rho X^\nu \frac{\partial G_{\nu ' \nu}}{\partial X^\sigma}=\notag\\
\bigg[\alpha_n \partial^t X_\nu  - \alpha_t \partial^n X_\nu \bigg]   \frac{1}{2}\partial_\rho X^{\nu '} \partial^\rho X^\nu (\partial_t G_{\nu ' \nu}) \label{A48B3}
\end{gather}
This term cancels the third line of \eqref{A47}. We thus see that all the terms in \eqref{A36} have no component parallel to the tangent plane.
Let us now calculate the component of  $ (C_{\mu \sigma} - C_{\sigma \mu}) (\mathbb X)_{X^\sigma} $ orthogonal to the tangent plane:

\begin{gather}
P^{\perp, \mu}_{\nu}\bigg(C_{\mu \sigma} - C_{\sigma \mu}\bigg) \bigg(\alpha_n \partial_n X^\sigma + \alpha_t \partial_t X^\sigma\bigg) =\notag\\
 \alpha_n P^{\perp, \sigma}_{\nu} \frac{1}{2}\partial_\rho X^{\nu '} \partial^\rho X^\nu \frac{\partial G_{\nu ' \nu}}{\partial X^\sigma} -  P^{\perp, \mu}_{\nu} \partial^n X^{\mu '}  \frac{\partial G_{\mu ' \mu}}{\partial X^\sigma} (\alpha_n \partial_n X^\sigma + \alpha_t \partial_t X^\sigma) =\notag\\
 P^{\perp, \mu}_{\nu} \bigg[ \alpha_n  \frac{1}{2}\partial_\rho X^{\nu '} \partial^\rho X^\nu \frac{\partial G_{\nu ' \nu}}{\partial X^\mu} -   \partial^n X^{\mu '}  \frac{\partial G_{\mu ' \mu}}{\partial X^\sigma} (\alpha_n \partial_n X^\sigma + \alpha_t \partial_t X^\sigma) \bigg] \label{A48B4}
\end{gather}
 We can finally write the expressions \eqref{A35} and \eqref{A48two} for $(\mathbb{X}_X)^\mu $ and $(\mathbb{X}_{\partial_n X})^\mu $ respectively as:
\begin{gather}
(\mathbb{X}_X)^\mu  = \alpha_n \partial_n X^\mu + \alpha_t \partial_t X^\mu \label{A48B5}\notag \\
(\mathbb{X}_{\partial_n X})^\mu  = \beta_n \partial_n X^\mu + \beta_t \partial_t X^\mu + (g^{n n})^{-1} P^{\perp, \mu}_{\nu} \bigg[ (g^{t t} \alpha_n - 2 g^{n t} \alpha_t ) \partial_t \partial_t X^\nu -g^{n n} \alpha_t \partial_t \partial_n X^\nu \bigg]  \notag\\
+ (g^{n n})^{-1} P^{\perp, \mu}_{\nu} G^{\nu w} \dot{G}_{w k} \bigg[ (g^{t t} \alpha_n - g^{n t} \alpha_t ) \partial_t X^k +  (g^{n t} \alpha_n - g^{n n}\alpha_t )\partial_n X^k \bigg] +\notag\\
- (g^{n n})^{-1} P^{\perp, \mu}_{\nu} \bigg[ \alpha_n  \frac{1}{2}\partial_\rho X^{\nu '} \partial^\rho X^\nu \frac{\partial G_{\nu ' \nu}}{\partial X^\mu} -   \partial^n X^{\mu '}  \frac{\partial G_{\mu ' \mu}}{\partial X^\sigma} (\alpha_n \partial_n X^\sigma + \alpha_t \partial_t X^\sigma) \bigg] \label{A48B6}
\end{gather}


\section{The boundary action $S^\partial_P$}\label{AppendixBoundaryAction}
We need the  BFV action $S^\partial_P$ to complete the data comprising the BFV structure for Polyakov theory. We can calculate it using the fact that the modified classical master equation can be recast in the form: $\frac{1}{2}\iota_{Q_M} \iota_{Q_M} \Omega_M = \pi^* S^\partial_P$. $Q_M$ acts on the anti-ghosts and anti-fields as:
\begin{gather*}
Q_MX^\dag_k =  \partial_\beta \bigg(-2 {\th}^{\alpha \beta} \partial_\alpha X^\mu G_{\mu \nu}\bigg) + {\th}^{\alpha \beta} \partial_\alpha X^\mu \partial_\beta X^\nu \frac{\partial G_{\mu \nu}}{\partial X^k} - \partial_\beta \bigg( X^\dag_k \zeta^\beta \bigg)\\
Q_M \zeta^\dag_\lambda = - X^\dag_\mu \partial_\lambda X^\mu + \partial_\alpha \bigg( 2 {\th}^{\dag \alpha \beta} {\th}_{\lambda \beta} \bigg) - {\th}^{\dag \alpha \beta} \partial_\lambda {\th}_{\alpha \beta} + \zeta^\dag_\alpha \partial_\lambda \zeta^\alpha + \partial_\alpha \bigg( \zeta^\dag_\lambda \zeta^\alpha \bigg)\\
Q_M{\th}^{\dag  \rho \sigma} = 2 {\th}^{\dag \alpha \rho} \partial_\alpha \zeta^\sigma - \partial_\alpha \bigg( {\th}^{\dag  \rho \sigma} \zeta^\alpha \bigg) - {\th}^{\rho \alpha} {\th}^{\sigma \beta} \partial_\alpha X^\mu \partial_\beta X^\nu G_{\mu \nu} - {\th}^{\dag  \rho \sigma} \partial_\lambda \zeta^\lambda
\end{gather*}
and on the fields and ghosts as:

\begin{gather*}
Q_M X^\mu = \zeta^\lambda \partial_\lambda X^\mu\\
Q_M {\th}_{\alpha \beta} =   \partial_\alpha \zeta^{\lambda} {\th}_{\lambda \beta} + \partial_\beta \zeta^{\lambda} {\th}_{\lambda \alpha }  + \zeta^\lambda \partial_\lambda {\th}_{\alpha \beta} - \partial_\lambda \zeta^\lambda {\th}_{\alpha \beta} \\
Q_M \zeta^\alpha = \zeta^\lambda \partial_\lambda \zeta^\alpha
\end{gather*}
The graded symplectic form $\Omega_M$ is: $\Omega_M = \int_{M} \delta X^\mu \delta X^\dag_\mu + \delta {\th}^{\alpha \beta} \delta {\th}^\dag_{\alpha \beta} + \delta \zeta^\alpha \delta^\dag_\alpha$.  Let us now calculate (we drop the $M$ subscript) $\iota_Q \iota_Q \Omega_M $ and manipulate the result to obtain a total derivative:
\begin{gather*}
\iota_Q \iota_Q \Omega_M = \intl_M\zeta^\lambda \partial_\lambda X^\kappa \bigg[ \partial_\beta \bigg(-2 {\th}^{\alpha \beta} \partial_\alpha X^\mu G_{\mu \kappa}\bigg) + {\th}^{\alpha \beta} \partial_\alpha X^\mu \partial_\beta X^\nu \frac{\partial G_{\mu \nu}}{\partial X^\kappa} - \underbrace{\partial_\beta \bigg( X^\dag_\kappa \zeta^\beta \bigg)}_{A_1} \bigg]+\\
 \zeta^\rho \partial_\rho \zeta^\lambda \bigg[ \underbrace{- X^\dag_\mu \partial_\lambda X^\mu}_{A_2} + \partial_\alpha \bigg( 2 {\th}^{\dag \alpha \beta} {\th}_{\lambda \beta} \bigg) - {\th}^{\dag \alpha \beta} \partial_\lambda {\th}_{\alpha \beta} + \underbrace{ \zeta^\dag_\alpha \partial_\lambda \zeta^\alpha + \partial_\alpha \bigg( \zeta^\dag_\lambda \zeta^\alpha \bigg)}_{B}\bigg] +\\
 \bigg[ \partial_\rho \zeta^{\lambda} {\th}_{\lambda \sigma} + \partial_\sigma \zeta^{\lambda} {\th}_{\lambda \rho }  + \zeta^\lambda \partial_\lambda {\th}_{\rho \sigma} - \partial_\lambda \zeta^\lambda {\th}_{\rho \sigma} \bigg] \times\\
\bigg[ 2 {\th}^{\dag \alpha \rho} \partial_\alpha \zeta^\sigma - \partial_\alpha \bigg( {\th}^{\dag  \rho \sigma} \zeta^\alpha \bigg) - {\th}^{\rho \alpha} {\th}^{\sigma \beta} \partial_\alpha X^\mu \partial_\beta X^\nu G_{\mu \nu} - {\th}^{\dag  \rho \sigma} \partial_\lambda \zeta^\lambda\bigg]
\end{gather*}
The terms marked as $A_1$, $A_2$ and $B$ are equivalent respectively to:
\begin{gather*}
A_1+A_2 = - \partial_\beta \bigg[ \zeta^\lambda \partial_\lambda X^\mu X^\dag_\mu \zeta^\beta \bigg]\\
B = \partial_\beta \bigg[ \zeta^\rho \partial_\rho \zeta^\lambda \zeta^\dag_\lambda \zeta^\beta \bigg]
\end{gather*}
removing these terms and considering that:
\begin{gather*}
 \bigg[ \partial_\rho \zeta^{\lambda} {\th}_{\lambda \sigma} + \partial_\sigma \zeta^{\lambda} {\th}_{\lambda \rho }  + \zeta^\lambda \partial_\lambda {\th}_{\rho \sigma} \bigg] (- {\th}^{\dag  \rho \sigma} \partial_\lambda \zeta^\lambda)  +\\
 - \partial_\lambda \zeta^\lambda {\th}_{\rho \sigma} \bigg[ 2 {\th}^{\dag \alpha \rho} \partial_\alpha \zeta^\sigma - \partial_\alpha \bigg( {\th}^{\dag  \rho \sigma} \zeta^\alpha \bigg) - {\th}^{\rho \alpha} {\th}^{\sigma \beta}- {\th}^{\dag  \rho \sigma} \partial_\lambda \zeta^\lambda\bigg] =0
\end{gather*}
we are left with:
\begin{gather*}
\iota_Q \iota_Q \Omega_M = \intl_M\zeta^\lambda \partial_\lambda X^\kappa \bigg[ \partial_\beta \bigg(-2 {\th}^{\alpha \beta} \partial_\alpha X^\mu G_{\mu \kappa}\bigg) + {\th}^{\alpha \beta} \partial_\alpha X^\mu \partial_\beta X^\nu \frac{\partial G_{\mu \nu}}{\partial X^\kappa}  \bigg]+\\
 \zeta^\rho \partial_\rho \zeta^\lambda \bigg[ + \partial_\alpha \bigg( 2 {\th}^{\dag \alpha \beta} {\th}_{\lambda \beta} \bigg) - {\th}^{\dag \alpha \beta} \partial_\lambda {\th}_{\alpha \beta} \bigg] + \partial_\lambda \zeta^\lambda {\th}_{\rho \sigma}  {\th}^{\rho \alpha} {\th}^{\sigma \beta} \partial_\alpha X^\mu \partial_\beta X^\nu G_{\mu \nu}+ \\
 \bigg[ \partial_\rho \zeta^{\lambda} {\th}_{\lambda \sigma} + \partial_\sigma \zeta^{\lambda} {\th}_{\lambda \rho }  + \zeta^\lambda \partial_\lambda {\th}_{\rho \sigma} \bigg] \times\\
\bigg[ 2 {\th}^{\dag \alpha \rho} \partial_\alpha \zeta^\sigma - \partial_\alpha \bigg( {\th}^{\dag  \rho \sigma} \zeta^\alpha \bigg) - {\th}^{\rho \alpha} {\th}^{\sigma \beta} \partial_\alpha X^\mu \partial_\beta X^\nu G_{\mu \nu}\bigg].
\end{gather*}
We will consider separately the terms that have $X^\mu$ or its derivatives and the terms that do not. The ones that do are:
\begin{gather*}
\zeta^\lambda \partial_\lambda X^\kappa \bigg[ \partial_\beta \bigg(-2 {\th}^{\alpha \beta} \partial_\alpha X^\mu G_{\mu \kappa}\bigg) + {\th}^{\alpha \beta} \partial_\alpha X^\mu \partial_\beta X^\nu \frac{\partial G_{\mu \nu}}{\partial X^\kappa}  \bigg]+\\
 + \partial_\lambda \zeta^\lambda {\th}_{\rho \sigma}  {\th}^{\rho \alpha} {\th}^{\sigma \beta} \partial_\alpha X^\mu \partial_\beta X^\nu G_{\mu \nu}+ \\
- \bigg[ \partial_\rho \zeta^{\lambda} {\th}_{\lambda \sigma} + \partial_\sigma \zeta^{\lambda} {\th}_{\lambda \rho }  +  \zeta^\lambda \partial_\lambda {\th}_{\rho \sigma}\bigg]  {\th}^{\rho \alpha} {\th}^{\sigma \beta} \partial_\alpha X^\mu \partial_\beta X^\nu G_{\mu \nu} =\\
\partial_\beta \bigg[ - 2 \zeta^\lambda \partial_\lambda X^\nu {\th}^{\alpha \beta} \partial_\alpha X^\mu G_{\mu \nu} + \zeta^\beta {\th}^{\alpha \lambda} \partial_\alpha X^\mu \partial_\lambda X^\nu G_{\mu \nu} \bigg]
\end{gather*}
the remaining terms are:
\begin{gather*}
 \zeta^\rho \partial_\rho \zeta^\lambda \bigg[ + \partial_\alpha \bigg( 2 {\th}^{\dag \alpha \beta} {\th}_{\lambda \beta} \bigg) - {\th}^{\dag \alpha \beta} \partial_\lambda {\th}_{\alpha \beta} \bigg] + \\
 \bigg[ \partial_\rho \zeta^{\lambda} {\th}_{\lambda \sigma} + \partial_\sigma \zeta^{\lambda} {\th}_{\lambda \rho }  + \zeta^\lambda \partial_\lambda {\th}_{\rho \sigma} \bigg] \bigg[ 2 {\th}^{\dag \alpha \rho} \partial_\alpha \zeta^\sigma - \partial_\alpha \bigg( {\th}^{\dag  \rho \sigma} \zeta^\alpha \bigg)\bigg] =\\
\partial_\beta \bigg[  - \zeta^\lambda \partial_\lambda {\th}_{\rho \sigma} {\th}^{\dag  \sigma \rho} \zeta^\beta + 2 \zeta^\lambda \partial_\lambda \zeta^\sigma {\th}^{\dag  \beta \rho} {\th}_{\sigma \rho} - 2 \partial_\sigma \zeta^\lambda {\th}_{\lambda \rho} {\th}^{\dag  \rho \sigma} \zeta^\beta \bigg].
\end{gather*}

Putting all the pieces together we explicitly remain with the integral of a total derivative (we assume integration against a volume form on $M$):
\begin{gather*}
\iota_Q \iota_Q \Omega_M=
    \int_{M} \partial_\beta \bigg[ - \zeta^\lambda \partial_\lambda X^\mu X^\dag_\mu \zeta^\beta  - 2 \zeta^\lambda \partial_\lambda X^\nu {\th}^{\alpha \beta} \partial_\alpha X^\mu G_{\mu \nu} + \zeta^\beta {\th}^{\alpha \lambda} \partial_\alpha X^\mu \partial_\lambda X^\nu G_{\mu \nu} \bigg] + \\
\int_{M} \partial_\beta \bigg[ \zeta^\rho \partial_\rho \zeta^\lambda \zeta^\dag_\lambda \zeta^\beta - \zeta^\lambda \partial_\lambda {\th}_{\rho \sigma} {\th}^{\dag  \sigma \rho} \zeta^\beta + 2 \zeta^\lambda \partial_\lambda \zeta^\sigma {\th}^{\dag  \beta \rho} {\th}_{\sigma \rho} - 2 \partial_\sigma \zeta^\lambda {\th}_{\lambda \rho} {\th}^{\dag  \rho \sigma} \zeta^\beta \bigg]
\end{gather*}
and using Stoke's theorem:
\begin{gather*}
\iota_Q \iota_Q \Omega_M= \int_{\partial M}- \zeta^\lambda \partial_\lambda X^\mu X^\dag_\mu \zeta^n  - 2 \zeta^\lambda \partial_\lambda X^\nu {\th}^{\alpha n} \partial_\alpha X^\mu G_{\mu \nu} + \zeta^n {\th}^{\alpha \lambda} \partial_\alpha X^\mu \partial_\lambda X^\nu G_{\mu \nu}  + \\
\int_{\partial M} \zeta^\rho \partial_\rho \zeta^\lambda \zeta^\dag_\lambda \zeta^n - \zeta^\lambda \partial_\lambda {\th}_{\rho \sigma} {\th}^{\dag  \sigma \rho} \zeta^n + 2 \zeta^\lambda \partial_\lambda \zeta^\sigma {\th}^{\dag n \rho} {\th}_{\sigma \rho} - 2 \partial_\sigma \zeta^\lambda {\th}_{\lambda \rho} {\th}^{\dag  \rho \sigma} \zeta^n
\end{gather*}
we can express the first and second terms, and the third term as functions of the reduced variables:
\begin{gather*}
- \zeta^\lambda \partial_\lambda X^\mu X^\dag_\mu \zeta^n  - 2 \zeta^\lambda \partial_\lambda X^\nu {\th}^{\alpha n} \partial_\alpha X^\mu G_{\mu \nu} = -2 \bigg[ \sigma^t \partial_t X^\mu J_\mu + \sigma^n J_\mu J^\mu \bigg]\\
 \zeta^n {\th}^{\alpha \lambda} \partial_\alpha X^\mu \partial_\lambda X^\nu G_{\mu \nu} =  \sigma^n \bigg[ J_\mu J^\mu - \partial_t X^\mu \partial_t X^\nu G_{\mu \nu} \bigg]
\end{gather*}
while the second line can be expressed as:
\begin{gather*}
-2 \rho^{\dag n n} \sigma^t \partial_t \sigma^n  + 2 \rho^{\dag n t} \sigma^t \partial_t \sigma^t
\end{gather*}
and thus:
\begin{gather*}
\iota_Q \iota_Q \Omega_M= \int_{\partial M} -2  \sigma^t \partial_t X^\mu J_\mu - \sigma^n \bigg[ J_\mu J^\mu + \partial_t X^\mu \partial_t X^\nu G_{\mu \nu} \bigg] -2 \rho^{\dag n n} \sigma^t \partial_t \sigma^n  + 2 \rho^{\dag n t} \sigma^t \partial_t \sigma^t
\end{gather*}
which is the BFV boundary action of Equation \eqref{e:PolBoundaryAction}.

\printbibliography
\end{document}